\documentclass{article}
\usepackage[utf8]{inputenc}
\usepackage[margin=1in]{geometry}  
\usepackage{booktabs}
\usepackage{graphicx}
\usepackage{mathrsfs}
\usepackage{makecell}
\usepackage{blkarray}
\usepackage{tikz}
\usepackage{natbib}
\usepackage{caption}   
\usepackage{float}     
\usepackage{enumitem}
\usepackage{amsmath, amsthm, amssymb, algorithm}
\usepackage[noend]{algpseudocode}
\allowdisplaybreaks
\usepackage{setspace}
\usepackage{bbm}
\usepackage{mathtools}
\usepackage[colorlinks=true,linkcolor=blue,citecolor=blue,urlcolor=blue]{hyperref}
\usepackage{cleveref}

\newcommand\blfootnote[1]{%
  \begingroup
  \renewcommand\thefootnote{}
  \footnote{#1}%
  \addtocounter{footnote}{-1}
  \endgroup
}

\DeclarePairedDelimiter\norm{\lVert}{\rVert}

\title{Neyman Jackknife: Design-Based Variance Estimation \\ for Causal Inference under Interference}
\author{
Bryan Park\\
\texttt{bryan314@stanford.edu}
\and
Stefan Wager\\
\texttt{swager@stanford.edu}
}
\date{Stanford University \\ \today}

\begin{document}

\maketitle

\newtheorem{theorem}{Theorem}
\newtheorem{problem}{Problem}
\newtheorem{lemma}[theorem]{Lemma}
\newtheorem{definition}{Definition}
\newtheorem{example}{Example}
\newtheorem{observation}{Observation}
\newtheorem*{answer}{Answer}
\newtheorem{corollary}[theorem]{Corollary}
\newtheorem{proposition}[theorem]{Proposition}
\newtheorem{assumption}{Assumption}
\newtheorem*{remark}{Remark}
\newtheorem{condition}{Condition}
\newtheorem{part_label}{Part}

\let\emptyset\varnothing

\newcommand{\Var}{\operatorname{Var}}
\newcommand{\s}{\normalfont\text{s}_{\text{eff}}}
\newcommand{\UB}{\normalfont\text{UB}_{\text{oracle}}}
\newcommand{\PE}{\normalfont\text{Err}_{\text{proxy}}}
\newcommand{\HE}{\normalfont\text{Err}_{\text{het}}}
\newcommand{\Cov}{\operatorname{Cov}}
\newcommand{\C}{\mathcal{C}}
\newcommand{\E}{\mathbb{E}}
\newcommand{\N}{\mathbb{N}}
\newcommand{\W}{\mathbf{W}}
\newcommand{\U}{\mathbf{U}}
\newcommand{\G}{\mathcal{G}}
\newcommand{\F}{\mathcal{F}}
\renewcommand{\W}{\mathbf{W}}
\newcommand{\indep}{\perp \!\!\! \perp}
\newcommand{\FX}{\mathcal{F}^\mathbf{X}}
\renewcommand{\H}{\mathcal{H}}
\renewcommand{\L}{\mathcal{L}}
\renewcommand{\S}{\mathcal{S}}
\renewcommand{\N}{\mathcal{N}}
\newcommand{\Z}{\mathbf{Z}}
\newcommand{\T}{\mathcal{T}}
\renewcommand{\P}{\mathbb{P}}
\newcommand{\B}{\mathcal{B}}
\newcommand{\A}{\mathcal{A}}
\newcommand{\V}{\mathcal{V}}
\newcommand{\X}{\mathbf{X}}
\newcommand{\D}{\mathcal{D}}
\newcommand{\R}{\mathbf{R}}
\newcommand{\Y}{\mathbf{Y}}
\newcommand{\limd}[1][]{\xrightarrow[#1]{d}}
\newcommand{\limp}[1][]{\xrightarrow[#1]{p}}
\newcommand{\limas}[1][]{\xrightarrow[#1]{\mathrm{a.s.}}}
\makeatletter
\newcommand{\Mod}[1]{\ (\mathrm{mod}\ #1)}
\newcommand{\ed}{\,{\buildrel d \over =}\,}

\begin{abstract}
Conservative variance estimation is a key component of design-based causal inference. Under interference, where one unit's treatment may affect the outcome of another unit, variance estimation is particularly challenging due to the dependence between observed outcomes. Prior methods typically assume a specific form for the treatment-effect estimator, then provide a tailored variance analysis for each design-estimator pair. In this paper, we introduce the ``Neyman Jackknife," a flexible framework for design-based, conservative variance estimation under interference. After analyzing the design once, our framework applies to any treatment-effect estimator, viewing it as a black-box function of the treatment vector. Our variance estimators are built from recomputations of the treatment-effect estimator, each dropping a set of treatments and any outcomes they may affect. We prove that these estimators are finite-sample conservative in expectation. In classical settings, our approach recovers estimators closely related to Neyman's estimator under SUTVA and the Newey--West HAC variance estimator for time series. For more complex designs such as switchbacks and randomized graph cluster randomization, numerical experiments show that our framework is conservative and comparable with existing methods. Our approach also covers treatment-effect estimators that are nonlinear in the outcomes, such as synthetic difference in differences, for which conservative variance estimators were previously unavailable in the presence of spillovers.
\end{abstract}

\section{Introduction}
Even in the simplest randomized experiment,\blfootnote{\hspace{-6mm}This work was supported by ONR Grant N-00014-24-12091, as well as a Stanford Graduate Fellowship.}
the design-based variance of a treatment-effect estimator cannot be estimated without bias. This is due to the fundamental problem of causal inference where multiple potential outcomes for a single unit cannot be observed at the same time \citep{holland1986}. Thus, a general alternative is to derive variance estimators that are conservative in expectation. Classically, \citet{neyman1923applications} derived such an estimator under the stable unit treatment value assumption (SUTVA) \citep{rubin1980sutva} by upper-bounding unobservable terms in the target variance. 

More recently, design-based, conservative variance estimation has been studied in the setting of interference \citep{Hudgens,Manski}, where the treatment for one unit may affect the outcome of another unit. While various methods have been proposed, two main lines of work have emerged. The first builds on Neyman's idea of upper-bounding an explicit expression of the target variance \citep{Aronow,conflict2026,lu2025design, harshaw2026variance}. The second brings the heteroskedasticity and autocorrelation consistent (HAC) tradition from econometrics \citep{NeweyWest1987, conley1999, kojevnikov2021limit}
into design-based inference \citep{leung2022ani, gao2025network,  lin2025timeseries}.

In this paper, we present an alternative to these analytic approaches. We introduce the ``Neyman Jackknife," a flexible framework for design-based, conservative variance estimation under interference. The core idea is to exploit the fact that in the design-based setting, the sole randomness lies in the treatment assignment. By viewing the treatment-effect estimator as an arbitrary function of the treatment assignment, we use the Poincar\'{e} inequality \citep[Lemma 13.7]{levin2017markov} to upper-bound the variance, rather than assuming any specific form of the treatment-effect estimator. The resulting variance estimators are obtained by recomputing the original treatment-effect estimator with some treatments and any affected outcomes left out, generalizing the classical leave-one-out Jackknife \citep{tukey1958bias}.

Despite a different algorithmic motivation, our approach turns out to be closely related to many of the prior analytic methods. For instance, under SUTVA and the completely randomized design, our framework recovers variance estimators closely related to Neyman's estimator, namely a weighted sum of the within-arm sample variances \citep{neyman1923applications}. There are also close connections with HAC methods. Under circular $2M$-dependence and the Bernoulli design, our framework gives a scaled version of the Newey--West estimator with Bartlett weights \citep{NeweyWest1987}. For more general $M$-hop interference, our framework rederives the overlap weights proposed by \citet{kojevnikov2021bootstrap}. Both examples reflect the well-known connection between block-recompute and HAC methods in time series \citep{Kunsch1989, politis1997subsampling}. 

Finally, we illustrate our framework in numerical experiments across three settings. First, we consider switchbacks with burn-ins under decaying carryover effects \citep{bojinov2023design, HuWager2022}. Second, we consider spatial network experiments under randomized graph cluster randomization \citep{ugander2023}. In both settings, we compare the Neyman Jackknife with existing methods and find that it remains conservative while producing intervals of comparable width. Third, we consider design-based inference for the synthetic difference in differences (SDID) estimator \citep{sdid} applied to a randomized experiment with spillovers.  To the best of our knowledge, as the SDID estimator is nonlinear in the outcomes, conservative variance estimation for SDID in the presence of spillovers has no available solution in the literature. In contrast, the Neyman Jackknife applies directly and provides a solution where no baseline was available.

\subsection{Related Work}\label{related}
The literature on design-based, conservative variance estimation goes back to \cite{neyman1923applications}, who analyzed the difference-in-means estimator in a completely randomized design under SUTVA. Subsequent work sharpened Neyman's variance estimator by exploiting additional information regarding the unobserved finite-population pairing of potential outcomes \citep{robins1988confidence, AronowGreenLee2014, ImbensMenzel2021, nutz2022directional}. \cite{MukerjeeDasguptaRubin2018} and \cite{ChattopadhyayImbens2024} proposed frameworks that extend Neyman-style variance estimation to more general designs. In general, conservative variance estimation under SUTVA has been well studied, and we refer the reader to the references above for further background.

Under interference, variance estimation is more difficult as there is structural dependence between the observed outcomes. As aforementioned, two main lines of work have emerged to address this problem. The first line of work builds on Neyman's idea of beginning with an explicit (but non-estimable) characterization of the variance, then developing estimable upper bounds. For instance, \citet{Aronow} achieved this by expanding the covariance and applying Young's inequality to upper-bound each unobservable term. \citet{harshaw2026variance} refined this approach by upper-bounding the variance at the level of quadratic forms, instead of a term-by-term approach. Other analytic techniques to upper-bound the target variance have been proposed by \citet{conflict2026} and \citet{lu2025design}. We further review these methods in the \hyperref[app]{Appendix} (Section \ref{overview_prior_sec_appendix}).

The second line of work is based on the HAC method from econometrics. Roughly speaking, given a random vector $\boldsymbol{\psi},$ the HAC method estimates the variance of $\mathbf{1}'\boldsymbol{\psi}$ via $(\boldsymbol{\psi} - \mathbf{c})'K(\boldsymbol{\psi} - \mathbf{c})$ with some choice of centering $\mathbf{c}$ and kernel $K.$ In the econometrics literature, $\E[\boldsymbol{\psi}]$ is typically either known or assumed to have a common-mean structure. Centering by the true expectation of $\boldsymbol{\psi}$ or the sample mean, appropriate choices of the kernel $K$ enable consistent variance estimation \citep{NeweyWest1987, conley1999, kojevnikov2021bootstrap, kojevnikov2021limit}. In the design-based setting, however, $\E[\boldsymbol\psi]$ is usually the unknown treatment-effect vector $\boldsymbol\tau$, where the individual effects may be heterogeneous. Thus, while design-based HAC estimators center by the treatment-effect estimator $\hat\tau$, consistent variance estimation is generally not possible, and a new analysis is needed to guarantee conservativeness. Design-based HAC estimators under interference have been studied by \citet{leung2022ani, gao2025network, lin2025timeseries}, and \citet{wager2026causal}. Again, we further describe these methods in the \hyperref[app]{Appendix} (Section~\ref{overview_prior_sec_appendix}).

Our approach builds on the Jackknife tradition, including the original Jackknife for independent data \citep{tukey1958bias, miller1974jackknife,efron1982jackknife} and later extensions for dependent data \citep{Kunsch1989, politis1997subsampling, bloznelis2001orthogonal}. Our method builds on the insight that sampling uncertainty should be separated from treatment-assignment uncertainty \citep{AbadieAtheyImbensWooldridge2020, ImbensMenzel2021}. Rather than resampling or imputing potential outcomes directly, we delete some treatments and any affected outcomes, then recompute the treatment-effect estimator. In this sense, our procedure is a Jackknife for treatment assignments, while the underlying potential outcomes are held fixed. Our main formal result on conservativeness of the Neyman Jackknife is a generalization of the classical result of \citet{efron1981jackknife} on conservativeness of the leave-one-out Jackknife with independent data, both in its substance and in the probabilistic tools used: Our conservativeness argument relies on the Poincar\'{e} inequality \citep[Lemma 13.7]{levin2017markov}, which is a generalization of the Efron--Stein inequality underlying the result of \citet{efron1981jackknife}.

\section{The Neyman Jackknife}
\label{sec:NJ}

In Section \ref{setting}, we formally describe our causal setup. In Section \ref{NJ_introduce}, we introduce the Neyman Jackknife. 
\subsection{Problem Setting}\label{setting}

We consider a finite-population model for treatment-effect estimation with spillovers. There are $m$ intervention units labeled $1,\dots, m$ and $n$ outcome units labeled $1,\dots, n.$ Each intervention unit $j\in [m]$ receives treatment $W_j\in\{0,1\}$ and we observe $Y_i$ for each outcome unit $i\in [n].$ In the standard spillover setting,
there is a natural one-to-one pairing between intervention and outcome units, i.e., $m = n$ and $W_i$
is the treatment assigned to the $i$th unit \citep[e.g.,][]{Aronow,Hudgens,Manski}. However, our model also accommodates more general settings, such as spatial experiments \citep{pollmann2020causal} or bipartite experiments \citep{lu2025design}.
We write $\W$ and $\Y$ for the vectors of all treatments and observed outcomes respectively.
We make the following two assumptions throughout.

\begin{assumption}[Randomized experiment]
\normalfont
\label{assu:rand}
Treatment is randomized, i.e., $\W \sim \pi$ for some known distribution $\pi$ on $\{0,1\}^m$. We refer to $\pi$ as the design.
\end{assumption}

\begin{assumption}[Bipartite interference]
\normalfont
\label{assu:PO}
We posit deterministic potential outcomes $$\{Y_i(w): i\in [n], w\in\{0,1\}^m\}$$ and assume
$Y_i = Y_i(\W)$. Each outcome unit $i\in[n]$ has a known exposure set
$\N_i \subseteq [m]$ such that $Y_i$ is only responsive to treatments in $\N_i$. In other words, $Y_i(w) = Y_i(w')$ whenever $w_{\N_i} =w'_{\N_i}$ where we write $\mathbf{v}_E = (v_j)_{j\in E}$ for a vector $\mathbf{v}$ and subset $E $. 
\end{assumption}

This setting is well studied in causal inference, and several unbiased estimators for various
causal estimands are available (see \citet[Chapter 12]{wager2026causal} for a recent review). For example, one can obtain unbiased estimates of the global average treatment effect (GATE)
\begin{align}
    \tau_{\normalfont\text{GATE}} = \frac{1}{n} \sum_{i = 1}^n\left(Y_i(\mathbf{1}_m) - Y_i(\mathbf{0}_m)\right)
    \label{eq:GATE}
\end{align}
via inverse-propensity weighting,
\begin{align}
\hat\tau_{\normalfont\text{IPW}} = \frac{1}{n} \sum_{i = 1}^n \left(\frac{\mathbbm{1}\{\forall j\in \N_i, W_j=1\}}{\Pr_\pi(\{\forall j\in \N_i,W_j=1\})}-\frac{\mathbbm{1}\{\forall j\in \N_i, W_j=0\}}{\Pr_\pi(\{\forall j\in \N_i,W_j=0\})}\right)Y_i.\label{eq:HT}
\end{align}
When the intervention units match the outcome units ($m=n$), another standard estimand is the average direct effect \citep{Savje, Hu} 
\begin{align}
    \tau_{\normalfont\text{ADE}} = \frac{1}{n}\sum_{i=1}^n \E_\pi\left[Y_i(w_i=1;\W_{-i}) - Y_i(w_i=0; \W_{-i})\right].\label{eq:ade}
\end{align}
Under the Bernoulli design, the Horvitz--Thompson estimator
\begin{align*}
    \hat\tau_{\normalfont\text{HT}}=\frac{1}{n}\sum_{i=1}^n \left(\frac{W_i}{\Pr_\pi(W_i=1)} - \frac{1-W_i}{\Pr_\pi(W_i=0)} \right)Y_i
\end{align*}
is unbiased for the average direct effect.
Throughout this paper, we assume that the analyst has already chosen a point estimator that targets their causal quantity of interest, and is now interested in a conservative, design-based variance estimator that can be used for inference. The Neyman Jackknife paradigm provides a simple and versatile recomputing-based approach.

We note that Assumption \ref{assu:PO} falls under the exposure mapping framework of \citet{Aronow}. Given the exposure sets $\{\N_i\}_{i=1}^n,$ we can construct a bipartite graph $G$ with vertex set $[m]\sqcup [n]$; there is an edge from $j\in [m]$ to $i\in [n]$ if and only if $j\in \N_i$. While we do not assume a low-dimensional exposure mapping such as the fraction of treated neighbors, Assumption \ref{assu:PO} still posits that the outcome of $i\in [n]$ depends only on the treatments of its neighbors in $G.$ We refer to $G$ as the interference (or exposure) graph. When the intervention and outcome units coincide, $G$ collapses to a graph with vertex set $[n]$. In this case, $G$ may contain self-loops $(i\in \N_i)$ and also be directed, i.e., $j\in \N_i$ need not imply $i\in \N_j.$ 

\subsection{Constructing the Neyman Jackknife}\label{NJ_introduce}
The classical Jackknife \citep{tukey1958bias, miller1974jackknife, efron1982jackknife} estimates the variance of a statistic
by repeatedly recomputing leave-one-out versions of it. Consider a setting where
$\mathbf{X}=(X_1, \, \ldots, \, X_n)$ are independent (though not necessarily identically distributed)
observations, let $f(\cdot)$ be any measurable function of these random variables, and let
$f^{(-i)}(\mathbf{X}_{-i})$ be the recomputed version of $f(\cdot)$ without $X_i$. Let $f^*(\X)$ denote the average of the leave-one-out estimates. Then, \citet[Theorem 2]{efron1981jackknife}
showed that the Jackknife variance estimator
\begin{equation}
\label{eq:tukey_jack}
\widehat V_{\text{LOO}} = \frac{n-1}{n}\sum_{i = 1}^n \left(f^*(\X)-f^{(-i)}\left(\X_{-i}\right)\right)^2
\end{equation}
is conservative in expectation for the variance of $f^*(\X)$ (provided $f(\X)$ has a finite
second moment). 

However, there are some obstacles to using $\widehat{V}_{\normalfont\text{LOO}}$ in our setup described in Section \ref{setting}. First, $X_i$ is usually taken as the $i$th sampled unit. In our design-based setting, however, units and their potential outcomes are fixed, hence the unit itself is not an independent source of randomness. The natural objects to leave out are instead the treatments. Second, under interference a single treatment can affect multiple outcomes, so leaving it out requires ignoring all of those outcomes simultaneously. Finally, we target the variance of $f,$ not $f^*,$ and the two variances do not necessarily coincide.

Our proposed Neyman Jackknife is a generalization of the classical Jackknife that addresses each of these challenges. Let $\hat\tau = f(\W)$ denote any treatment-effect estimator whose
variance we wish to estimate. In our design-based setting, note that $\Y$ is a function of $\W$.

\begin{definition}[Neyman Jackknife]\label{NJ_def}
\normalfont
The Neyman Jackknife takes two user-specified inputs:
\begin{enumerate}
\item A subset-sampling rule $\mu$ that draws a random subset $S\subseteq [m]$ given $\W\sim\pi$, i.e.,
$$\Pr(S=A \mid \W = w) = \mu(w, A)$$
for each $w\in\{0,1\}^m$ in the support of $\W$ and $A\subseteq [m].$ 
\item For each $A\subseteq [m]$, any proxy $f^{(-A)}(\W_{-A})$ that is computable from $\W_{-A}$ alone. For instance, $f^{(-A)}$ can use $Y_i$ such that $\N_i \cap A = \emptyset$.
\end{enumerate}
The Neyman Jackknife variance estimator is then
\begin{align}
\widehat V 
&=\frac{m}{\s} \sum_{A\subseteq [m]} \mu(\W,A)\cdot{\left(f(\W) - f^{(-A)}(\W_{-A})\right)^2}\label{eq:NJ_main}
\end{align}
where $\s = \s(\pi,\mu)\in [0,m]$ is the effective subset size of the subset-sampling rule $\mu$ under $\pi$ (Definition~\ref{ess_def}).
\end{definition}

In Section~\ref{theory_machine}, we define $\s = m\lambda$ where $\lambda$ is the spectral gap of a certain reversible Markov chain induced by $\mu$. For now, it suffices that $\s$ is an explicit constant determined by the design $\pi$ and the subset-sampling rule $\mu.$ Our main theoretical result is that $\widehat{V}$ is conservative in expectation for any choice of inputs.

\begin{theorem}
\label{main_NJ}
Under Assumptions \ref{assu:rand} and \ref{assu:PO}, the Neyman Jackknife variance estimator~\eqref{eq:NJ_main} satisfies
$$\E_{\W\sim\pi}[\widehat V]\geq \Var_{\W\sim \pi}(f(\W)).$$
\end{theorem}

The sum in \eqref{eq:NJ_main} is cheap to compute when $\mu(\W,\cdot)$ is supported on few subsets, e.g., only the $m$ singletons as in leave-one-out. Alternatively, one can use the Monte Carlo approximation
\begin{align}
\widehat{V}_{\normalfont{\text{MC}}}=\frac{m}{\s}\cdot \frac{1}{B}\sum_{b=1}^B\left(f(\W)-f^{(-S_b)}(\mathbf{W}_{-{S_b}})\right)^2\label{eq:MC}
\end{align}
where the $S_b$'s are drawn i.i.d. from $\mu(\W,\cdot)$. Since $\E[\widehat{V}_{\normalfont\text{MC}}\mid \W] =\widehat{V},$
Theorem~\ref{main_NJ} implies that $\widehat{V}_{\text{MC}}$ is also conservative in expectation.

We conclude by describing a default procedure to build the proxy $f^{(-A)}$ in Definition~\ref{NJ_def}. 
\begin{definition}[Drop-and-recompute proxy]\label{drop_and_recompute}
\normalfont
    Given exposure sets $\{\N_i\}_{i=1}^n,$ define the halo of $A\subseteq [m]$ as 
    \begin{align*}
        \mathcal{H}_A = \{i\in [n]: \N_i\cap A\neq\emptyset\}.
    \end{align*}
    The drop-and-recompute proxy $f^{(-A)}$ recomputes $f$ after dropping treatments in $A$ and outcomes in $\mathcal{H}_A.$
\end{definition}

\begin{figure}[t]
\centering
\includegraphics[width=\textwidth]{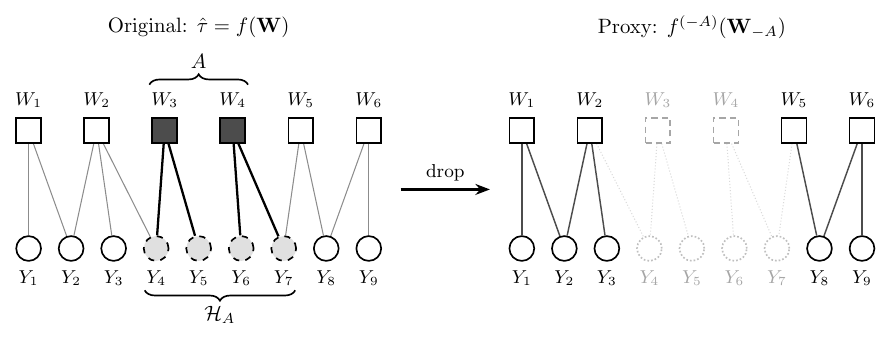}
\caption{Drop-and-recompute. Left: the original experiment with $A=\{3,4\}$
and $\mathcal{H}_A=\{4,5,6,7\}$. Right: the proxy $f^{(-A)}$ discards $\mathbf{W}_A$ and
every outcome in $\mathcal{H}_A$, then recomputes $\hat\tau$ on the remaining units. This ensures the proxy is computable from $\mathbf{W}_{-A}$ alone.}
\label{fig:drop_recompute}
\end{figure}

We illustrate Definition~\ref{drop_and_recompute} in Figure~\ref{fig:drop_recompute}. Since Theorem~\ref{main_NJ} holds for any computable proxy, the drop-and-recompute proxy is only one valid choice. In Sections~\ref{theory_machine} and \ref{covariate_sec}, we will show that better proxies (possibly involving covariates) can give tighter estimates of the variance.

\section{Theoretical Machinery of the Neyman Jackknife}\label{theory_machine}

In this section, we describe the underlying theory for the Neyman Jackknife. Recall that $\W\in\{0,1\}^m$ is the random treatment vector. We slightly generalize our goal to conservative variance estimation of $f(\W)$ where $\W\in\mathcal{X}^m$ is a random vector with known distribution and $\mathcal{X}$ is finite. 

We first review the classical Poincar\'{e} inequality for reversible Markov chains \citep[Lemma 13.7]{levin2017markov}. Let $\W\sim\pi$ be any $\S$-valued random variable where $\mathcal{S}$ is finite ($\W$ is not necessarily a random vector here). 
Let $\Omega=\{w\in\S:\pi(w)>0\}$ denote the support of $\pi$ and let $P$ denote any $|\Omega|\times|\Omega|$ transition matrix on $\Omega,$ meaning $P(w,w')\geq 0$ and 
$\sum_{w'\in\Omega}P(w,w')=1$
for all $w,w'\in\Omega.$ We further assume that $P$ is reversible with respect to $\pi$, and let $\lambda$ denote the spectral gap of $P.$

\begin{definition}[Reversibility]\label{reversible}\normalfont
    We say that $P$ is reversible with respect to $\pi$ if the detailed balance equation 
    $$\pi(w)P(w,w') = \pi(w')P(w',w)$$
    is satisfied 
    for all $w,w'\in\Omega.$
\end{definition}

\begin{definition}[Spectral gap]\label{spectral_gap}\normalfont
    If $P$ is reversible with respect to $\pi,$ then all eigenvalues are real. We let $\lambda$ denote the spectral gap of $P,$ which is one minus the second largest eigenvalue of $P.$ We also say $\lambda$ is the spectral gap of the transition corresponding to $P.$
\end{definition}

The Poincar\'{e} inequality provides a general upper bound of $\Var(f(\W))$ via reversible transitions.
\begin{lemma}[Poincar\'{e} inequality for finite state space]\label{original}
    Let $\S$ be finite and $\W\sim \pi$ be an $\S$-valued random variable supported on $\Omega.$ Let $P$ denote a transition matrix on $\Omega$ that is reversible with respect to $\pi.$ Let $\lambda$ denote the spectral gap of $P.$ Finally, let $\W'$ denote the one-step transition of $\W$ according to $P,$ meaning
    \begin{align*}
        \Pr(\W'=w'\mid \W=w) = P(w,w')
    \end{align*}
    for all $w,w'\in\Omega.$ Then, we have
    \begin{align}
        \Var(f(\W))\leq \frac{1}{2\lambda}\E\left[\left(f(\W)-f(\W')\right)^2\right]\label{eq:Poincare_inequality}
    \end{align}
    for any measurable $f:\Omega\to\mathbb{R}$ such that $f(\W)$ has finite variance.
\end{lemma}

To make use of this inequality for $\W\in\mathcal{X}^m$, our first task is to associate the Neyman Jackknife with a reversible Markov chain. We do so via the following construction, which we refer to as Gibbs rerandomization.\footnote{We emphasize that Gibbs rerandomization is an analytic tool. Computing the Neyman Jackknife does not require carrying out Gibbs rerandomization.} The Neyman Jackknife as given in Definition \ref{NJ_def} involves a subset-sampling rule that gives a random subset $S$ of treatments to mask. Gibbs rerandomization starts from the same subset-sampling rule, but imagines rerandomizing treatments in $S$ rather than just masking them. Gibbs rerandomization provides a reversible Markov chain over $\W\in \mathcal{X}^m$ with stationary distribution $\pi$. We apply the Poincar\'e inequality to Gibbs rerandomization and find that it lets us characterize the behavior of the Neyman Jackknife.

\begin{definition}[Gibbs rerandomization]\label{Gibbs_def}
\normalfont
Let $\W\in\mathcal{X}^m$ with $\mathcal{X}$ finite. The Gibbs rerandomization for $\W\sim\pi$ with update set $S$ is the following two-part transition from $\W$ to $\W'$:
\begin{part_label}\label{part_one}\normalfont
    Given $\W$ and the subset-sampling rule $\mu,$ draw $S\sim \mu(\W,\cdot).$
\end{part_label}

\begin{part_label}\label{part_two}\normalfont
Define
$\W' = (\mathbf{W}'_S, \mathbf{W}_{-S})$ where 
\begin{align*}
\mathbf{W}'_S\sim \mathcal{L}(\mathbf{W}_S\mid S,\mathbf{W}_{-S})\qquad\text{and}\qquad
\mathbf{W}'_S \indep \mathbf{W}_S \mid S, \mathbf{W}_{-S}.
\end{align*}
In other words, $\W_S$ and $\W'_{S}$ are i.i.d. given $(S,\W_{-S}).$ 
\end{part_label}
\end{definition}

Note that Gibbs rerandomization differs from classical Gibbs sampling \citep[Chapter 3.3]{levin2017markov} by allowing the update set $S$ to depend on $\W$. The following lemma is helpful for understanding the transition. 
\begin{lemma}\label{ed}
    Let $\W'$ be obtained from $\W$ by Gibbs rerandomization with update set $S$. Then, 
    $$(S,\W,\W')\ed (S,\W',\W).$$
\end{lemma}
\begin{proof}
Given $S=A$ and $\W_{-A}=x$, we have
\begin{align*}
    (S,\W,\W') &= (A,(\W_A,x),(\W'_A,x))\\
    & \ed (A,(\W'_A,x),(\W_A,x))\\
    &= (S,\W',\W)
\end{align*}
since $\W_A$ and $\W'_A$ are i.i.d. given $S=A$ and $\W_{-A}=x.$ As the equality in distribution holds unconditionally as well, we conclude the proof.
\end{proof}

Let $P$ denote the $|\Omega|\times |\Omega|$ matrix describing the transition from $\W$ to $\W'$ given by Gibbs rerandomization with update set $S.$ In other words, 
\begin{align}
    P(w,w')= \Pr(\W'=w'\mid \W = w) \label{eq:stochastic_matrix}
\end{align}
for any $w,w'\in \Omega.$ By Lemma \ref{ed}, we see that $(\W,\W')$ is exchangeable. In particular, this implies that $P$ is reversible with respect to $\pi$, hence the Poincar\'{e} inequality is applicable. We show that the upper bound in \eqref{eq:Poincare_inequality}  can further be expressed in terms of $\W$ alone.

\begin{lemma}\label{Poincare_mod}
     Let $\W\sim\pi$ and $\lambda$ denote the spectral gap of Gibbs rerandomization for $\W$ with update set $S$. Then, we have
    \begin{align}
        \Var(f(\W))\leq \frac{1}{\lambda}\E\left[(f(\W)-\E[f(\W)\mid S,\mathbf{W}_{-S}])^2\right]\label{eq:oracle_Poincare}
    \end{align}
    for any measurable $f:\Omega\to\mathbb{R}$ such that $f(\W)$ has finite variance.
\end{lemma}
\begin{proof}
    Let $\W'$ be obtained from $\W$ by Gibbs rerandomization with update set $S.$ We can write
\begin{align*}
    \frac{1}{2\lambda}\left(f(\mathbf{W})-f(\mathbf{W'})\right)^2 &= \frac{1}{2\lambda}\big[\left(f(\mathbf{W})-h(S,\mathbf{W}_{-S})\right)+ \left(h(S,\mathbf{W}_{-S})-f(\mathbf{W'})\right)\big]^2\\
    &= \frac{1}{2\lambda}\left(f(\mathbf{W})-h(S,\mathbf{W}_{-S})\right)^2 + \frac{1}{2\lambda}\left(f(\mathbf{W'})-h(S,\mathbf{W}_{-S})\right)^2 + \frac{1}{\lambda}R
\end{align*}
where $h(S,\mathbf{W}_{-S})=\E[f(\mathbf{W})\mid S,\mathbf{W}_{-S}]$ and 
\begin{align*}
    R=\left(f(\mathbf{W})-h(S,\mathbf{W}_{-S})\right)\left(h(S,\mathbf{W}_{-S})-f(\mathbf{W'})\right).
\end{align*}
Since $\W_S$ and $\W'_S$ are i.i.d. given $(S,\W_{-S}),$ we see that
\begin{align*}
    \E\left[(f(\mathbf{W})-h(S,\mathbf{W}_{-S}))^2\right] =\E\left[(f(\mathbf{W'})-h(S,\mathbf{W}_{-S}))^2\right]. 
\end{align*}
Moreover, we have
\begin{align*}
    \E[R\mid S,\mathbf{W}_{-S}]&= \E\left[f(\mathbf{W})-h(S,\mathbf{W}_{-S})\mid S,\mathbf{W}_{-S}\right]\cdot \E\left[h(S,\mathbf{W}_{-S})-f(\mathbf{W'})\mid S,\mathbf{W}_{-S}\right] = 0
\end{align*}
and thus $\E[R]=0.$ Bringing everything together and applying Lemma \ref{original} gives the desired inequality.
\end{proof}

While the upper bound in \eqref{eq:oracle_Poincare} only involves $\W$, it contains the term $\E[f(\W)\mid S,\W_{-S}]$ which we generally cannot compute from a single observation of $\W$. Our key idea is to approximate $\E[f(\W)\mid S,\W_{-S}]$ via any computable proxy $f^{(-S)}=g(S,\W_{-S}).$ The conservative guarantee still holds as the conditional expectation minimizes the mean squared error. This insight gives our main theoretical result; note that Theorem \ref{main_NJ} is simply the setting of Theorem \ref{Main} with $\mathcal{X}=\{0,1\}.$
 In stating this result, we note that the spectral gap $\lambda$ of Gibbs rerandomization is in fact fully determined by the design $\pi$ and subset-sampling rule $\mu$. In order to abstract away from Gibbs rerandomization in the method specification, we refer to $m \lambda$ as the effective subset size $\s$ of $\mu$ under $\pi$ when convenient.

\begin{definition}
    \label{ess_def}
\normalfont
Given any distribution $\pi$ on $\mathcal{X}^m$ where $\mathcal{X}$ is finite, the effective subset size $\s=\s(\pi, \, \mu)$ of the subset-sampling rule $\mu$ under $\pi$ (as in Definition \ref{NJ_def}) is $m \lambda$ where $\lambda$ is the spectral gap of the Gibbs rerandomization process (Definition \ref{Gibbs_def}) induced by $\pi$ and $\mu$.
\end{definition}

\begin{theorem}\label{Main}
    Let $\W\sim\pi$ be any random vector in $\mathcal{X}^m$ where $\mathcal{X}$ is finite.
    Let $S \sim \mu$ be any subset-sampling rule with effective subset size $\s$ as in Definition \ref{ess_def}. Finally, let $f^{(-S)}=g(S,\W_{-S})$ denote any computable proxy. Then, the variance estimate
    \begin{align}
        \widehat V&= \frac{m}{\s}\cdot\E\left[(f(\W)-g(S,\mathbf{W}_{-S}))^2\mid \mathbf{W}\right]\nonumber\\
        &=\frac{m}{\s}\sum_{A\subseteq [m]}\mu(\W,A)\cdot \left(f(\W)-g(A,\mathbf{W}_{-A})\right)^2\label{eq:main}
    \end{align}
    is conservative in expectation, i.e., $\E[\widehat V]\geq \Var(f(\W))$.
\end{theorem}

\begin{proof}
    Write $\lambda = \s/m$. By orthogonality, we have the decomposition
     \begin{align}
        \E[\widehat V] = \frac{1}{\lambda}\E\left[(f(\W)-\E[f(\W)\mid S,\mathbf{W}_{-S}])^2\right] + \frac{1}{\lambda}\E\left[(g(S,\W_{-S})-\E[f(\W)\mid S,\mathbf{W}_{-S}])^2\right]. \label{eq:decomp}
    \end{align}
    Applying Lemma \ref{Poincare_mod} concludes the proof.
\end{proof}

The decomposition in \eqref{eq:decomp} provides useful insight on choosing the proxy $f^{(-S)}=g(S,\W_{-S})$. Namely, the first term is the oracle Poincar\'{e} bound (Lemma~\ref{Poincare_mod}), which does not involve the choice of proxy at all. Hence, the proxy adds conservativeness only through the second term, and the proxy should target the conditional expectation $\E[f(\W)\mid S,\W_{-S}].$ Thus, while the drop-and-recompute proxy (Definition~\ref{drop_and_recompute}) is a natural choice, better proxies may reduce the conservative bias of the Neyman Jackknife variance estimator. We study this further in Section~\ref{covariate_sec}.

\section{Classical Designs and Connections}\label{classical_sec}

Recall the Neyman Jackknife variance estimator
\begin{align*}
    \widehat{V} = \frac{m}{\s}\sum_{A\subseteq [m]}\mu(\W,A)\cdot \left(f(\W) - f^{(-A)}(\W_{-A})\right)^2
\end{align*}
from \eqref{eq:NJ_main}. The subset-sampling rule $\mu$ is chosen by the analyst. Given the exposure sets $\{\N_i\}_{i=1}^n$, the analyst can also construct the proxy $f^{(-A)}$ by dropping treatments in $A$ and outcomes in $\H_A$ (Definition~\ref{drop_and_recompute}). The most non-trivial step in using the Neyman Jackknife is computing the effective subset size $\s$ of $\mu$ under the design $\pi.$

In Section~\ref{recipe_sec}, we describe various subset-sampling rules and the resulting $\s$ under the Bernoulli, completely randomized, and stratified designs. Hence, an analyst working with these classical designs may simply use the given subset-sampling rules and the corresponding $\s$ for any treatment-effect estimator. In particular, we provide explicit expressions for the Neyman Jackknife variance estimator resulting from concrete subset-sampling rules.

In Sections~\ref{Neyman_connect} and \ref{NW}, we show that the Neyman Jackknife variance estimator is closely related to classical variance estimators in standard settings. Section~\ref{Neyman_connect} assumes a setting with SUTVA while Section~\ref{NW} assumes a setting with interference.

\subsection{The Neyman Jackknife for Classical Designs}\label{recipe_sec}

We begin with the Bernoulli design where each intervention unit $j\in [m]$ independently receives treatment $W_j\sim\text{Ber}(\pi_j)$ with $\pi_j\in(0,1)$. In this case, we can characterize $\s$ whenever the subset-sampling rule draws $S$ independent of $\W.$

\begin{lemma}\label{spectral_ber}
    Let $\W$ follow the Bernoulli design on $\{0,1\}^m$. Assume the subset-sampling rule $\mu$ is independent of $\W$ so that $S\indep \W.$ Then,
    \begin{align*}
        \s = m\cdot \min_{j\in [m]}\Pr(j\in S).
    \end{align*}
\end{lemma}

In the classical leave-one-out transition, $\mu$ chooses $S$ to be a uniformly random singleton, independent of $\W.$ In this case, $\Pr(j\in S) = 1/m$ for each $j\in [m]$ and thus $\s = 1$ by Lemma~\ref{spectral_ber}. More generally, one can choose $S$ to be a uniformly random part from a fixed partition of $[m]$. Let $V_1,\dots, V_K$ denote such a partition and let the subset-sampling rule $\mu$ draw $S$ as a uniformly random part, independent of $\W.$ Then, $\s = m/K$ by Lemma~\ref{spectral_ber} and the resulting Neyman Jackknife variance estimator is 
    \begin{align}
        \widehat{V} 
        &=\sum_{k=1}^K {\left(f(\W) - f^{(-V_k)}(\W_{-V_k})\right)^2}.\label{ber_est}
    \end{align}
Next, consider the completely randomized design on $\{0,1\}^m$ with $n_1$ treated units and $n_0$ control units so that $n_1 + n_0 = m.$ Let $\T,\mathcal{C}$ respectively denote the set of treated and control units. Our next result considers a natural subset-sampling rule under the completely randomized design. For more general subset-sampling rules, see the \hyperref[app]{Appendix} (Lemma~\ref{spectral_CRT}).
\begin{proposition}\label{spectral_CRT_k}
    Let $\W$ follow the completely randomized design on $\{0,1\}^m$ so that $n_1+n_0=m$. Fix integers $1\leq s_1\leq n_1$ and $1\leq s_0\leq n_0.$ Let $\mu(\W,\cdot)$ draw $S$ as a uniformly random set of $s_1$ treated and $s_0$ control units. Then,
    \begin{align*}
        \text{s}_{\mathrm{eff}} = m\cdot \frac{1/n_1 + 1/n_0}{1/s_1 + 1/s_0}.
    \end{align*}
    The resulting Neyman Jackknife variance estimator is
    \begin{align}
        \widehat{V} 
        &=\frac{s_1+s_0}{m}\cdot \frac{1}{\binom{n_1-1}{s_1-1}\binom{n_0-1}{s_0-1}}\sum_{A\subseteq \T: |A|=s_1} \sum_{B\subseteq \mathcal{C}: |B|=s_0}\left(f(\W) - f^{(-A\sqcup B)}(\W_{-A\sqcup B})\right)^2.\label{eq:crd_est}
    \end{align}
    Moreover, the Monte Carlo version with $B$ iterations is
    \begin{align}
        \widehat{V}_{\normalfont\text{MC}} = \frac{1/s_1+1/s_0}{1/n_1+1/n_0}\cdot \frac{1}{B}\sum_{b=1}^B\left(f(\W)-f^{(-S_b)}(\mathbf{W}_{-{S_b}})\right)^2\label{eq:CRD_est_MC}
    \end{align}
    where the $S_b$'s are drawn i.i.d. from $\mu(\W,\cdot).$
    
\end{proposition}

If $s_1:s_0 = n_1:n_0,$ note that $\s = s_1 + s_0$ in Proposition~\ref{spectral_CRT_k}. 
Finally, we consider the stratified completely randomized design on $\{0,1\}^m$. Here, we have strata $I_1,\dots, I_H$ that partition $[m].$ Independently across strata, $\W_{I_h}$ follows the completely randomized design with $n_{1h}$ treated and $n_{0h}$ control units so that $n_{1h}+n_{0h} = |I_h| = n_h.$ Let $\T_h, \C_h$ denote the set of treated and control units in $I_h.$ We provide a subset-sampling rule that selects $s_{1h}$ treated and $s_{0h}$ control units within a randomly selected stratum. Again, more general subset-sampling rules are given in the \hyperref[app]{Appendix} (Lemma~\ref{spectral_stratified_k}).

\begin{proposition}\label{strat_spec_main}
    Let $\W$ follow the stratified completely randomized design on $\{0,1\}^m$. For each $h\in[H],$ fix integers $1\leq s_{1h}\leq n_{1h}$, $1\leq s_{0h}\leq n_{0h}$ and let $$c_h = \frac{1/n_{1h}+1/n_{0h}}{1/s_{1h}+1/s_{0h}},\qquad c = \sum_{h=1}^H c_h^{-1}.$$
    Let $\mu(\W,\cdot)$ draw $S$ by picking a stratum $I_h$ with probability $c_h^{-1}/c,$ then selecting a uniformly random set of $s_{1h}$ treated and $s_{0h}$ control units in $I_h.$ Then, $\s = m/c$ and the Neyman Jackknife variance estimator is 
    \begin{align}
        \widehat{V} = \sum_{h=1}^H \Bigg[\frac{s_{1h}+s_{0h}}{n_{h}}\cdot \frac{1}{\binom{n_{1h}-1}{s_{1h}-1}\binom{n_{0h}-1}{s_{0h}-1}}\sum_{A\subseteq \T_h: |A|=s_{1h}}\sum_{ B\subseteq \C_h: |B|=s_{0h}}\left(f(\W) - f^{(-A\sqcup B)}(\W_{-A\sqcup B})\right)^2\Bigg].\label{eq:stratified_est}
    \end{align}
    Moreover, the Monte Carlo version with $B$ iterations is
    \begin{align}
        \widehat{V}_{\normalfont\text{MC}} =  c\cdot \frac{1}{B}\sum_{b=1}^B\left(f(\W)-f^{(-S_b)}(\mathbf{W}_{-{S_b}})\right)^2\label{eq:stratified_est_MC}
    \end{align}
    where the $S_b$'s are drawn i.i.d. from $\mu(\W,\cdot).$
\end{proposition}

The variance estimators in \eqref{ber_est}--\eqref{eq:stratified_est_MC} are obtained by plugging in the value of $\s$ into \eqref{eq:NJ_main}, summing over the support of $\mu(\W,\cdot)$ (or considering Monte Carlo versions), and simplifying expressions. By Theorem~\ref{main_NJ}, the variance estimators in \eqref{ber_est}--\eqref{eq:stratified_est_MC} are conservative in expectation.

\subsection{Neyman's SUTVA Setting}\label{Neyman_connect}
Consider Neyman's classical setting under SUTVA. There are $n$ units, each $i\in [n]$ receives treatment $W_i,$ and we observe outcome $Y_i.$ In particular, $\N_i = \{i\}$ for each $i\in [n].$ We use the completely randomized design with $n_1$ treated units and $n_0$ control units so that $n_1 + n_0 = n$. The treatment-effect estimator $\hat\tau = f(\W)$ is difference-in-means, i.e.,
\begin{align*}
    \hat\tau
    &= \frac{1}{n_1}\sum_{i\in \T} Y_i - \frac{1}{n_0}\sum_{i\in \C} Y_i.
\end{align*}
The Neyman Jackknife variance estimator in \eqref{eq:crd_est} is directly applicable. Recall that the subset-sampling rule for \eqref{eq:crd_est} draws $S = A\sqcup B$ where $A,B$ are uniformly random subsets of $s_1$ treated and $s_0$ control units, respectively. We choose the drop-and-recompute proxy 
\begin{align}
    f^{(-A\sqcup B)}(\W_{-A\sqcup B}) = \frac{1}{n_1 -s _1}\sum_{i\in \T\setminus A}Y_i - \frac{1}{n_0-s_0}\sum_{i\in \C\setminus B}Y_i\label{proxy_SUTVA}
\end{align}
where we note that $\mathcal{H}_{A\sqcup B} = A\sqcup B$. The following result relates the Neyman Jackknife variance estimator to Neyman's classical variance estimator.

\begin{proposition}\label{Neyman_connection}
Let $(s_1,s_0)$ be integer parameters satisfying $1\leq s_1< n_1$ and $1\leq s_0< n_0$. Let $\widehat{V}$ denote the Neyman Jackknife variance estimator in \eqref{eq:crd_est} with the drop-and-recompute proxy in \eqref{proxy_SUTVA}. Let $\widehat V_{\normalfont\text{Neyman}}$ denote Neyman's classical variance estimator. Then,
\begin{align*}
    &\widehat V = \alpha\cdot\frac{\hat S_\T^2}{n_1} + \beta\cdot\frac{\hat S_\C^2}{n_0},\qquad
    \widehat V_{\normalfont\text{Neyman}} = \frac{\hat S_\T^2 }{n_1} + \frac{ \hat S_\C^2}{n_0},
\end{align*}    
where $\hat S^2_\T$ and $\hat S^2_\C$ are the within-arm sample variances and 
\begin{align*}
    \alpha &= \frac{s_1+s_0}{s_0}\cdot \frac{n_0}{n}\cdot \frac{n_1}{n_1-s_1},\qquad \beta = \frac{s_1+s_0}{s_1}\cdot \frac{n_1}{n}\cdot \frac{n_0}{n_0-s_0}.
\end{align*}
If $s_1:s_0 = n_1:n_0,$ we further have
$$\widehat{V} = \frac{n}{n-s_1-s_0}\cdot \widehat{V}_{\normalfont\text{Neyman}}.$$
\end{proposition}
Note that $\widehat{V}$ estimates the variance of $f(\W)$ with recomputations involving $(n - s_1 -s_0)$ units. Thus, if we apply the finite-sample correction as in the leave-one-out Jackknife \citep{tukey1958bias,miller1974jackknife,efron1982jackknife}, the proportional case $s_1:s_0 = n_1:n_0$ exactly recovers
\begin{align*}
    \frac{n-s_1-s_0}{n}\cdot \widehat{V} = \widehat{V}_{\normalfont\text{Neyman}}.
\end{align*}
In the \hyperref[app]{Appendix} (Proposition~\ref{linear}), we further show that $(1-\s/n)\cdot \widehat{V}$ is admissible for general $s_1,s_0$ where $\s$ is given in Proposition~\ref{spectral_CRT_k}.

\subsection{$M$-Hop Interference}\label{NW}

Consider a setting where $n$ units are placed on a simple, undirected graph $G_0$. Each $i\in [n]$ receives treatment $W_i\in\{0,1\}$ and we observe outcome $Y_i.$ Assume that $Y_i$ only depends on the treatments of units within graph-distance $M\geq 0.$ In other words, we have $\N_i = B_M(i) = \{j\in [n]: d_{G_0}(i,j)\leq M\}.$ We target the average direct effect~\eqref{eq:ade} using the Bernoulli design, where unit $i\in [n]$ is independently treated with probability $\pi_i\in(0,1).$ The treatment-effect estimator $\hat\tau=f(\W)$ is the Horvitz--Thompson estimator
\begin{align*}
    \hat\tau &= \frac{1}{n}\sum_{i=1}^n \psi_i=\frac{1}{n}\sum_{i=1}^n \left(\frac{W_i}{\pi_i}- \frac{1-W_i}{1-\pi_i}\right)Y_i.
\end{align*}
For simplicity, we consider the case where $G_0$ is the cycle graph on $[n]$ (see Section~\ref{M_hop_sec_appendix} of the \hyperref[app]{Appendix} for general $G_0$). Then, $\{Y_i\}_{i=1}^n$ can be viewed as a circular $2M$-dependent time series. 

We apply the Neyman Jackknife to this setting. Consider the subset-sampling rule $\mu(\W,\cdot)$ that draws $S$ as a uniformly random block of size $s\geq 1$ on the cycle graph. Formally, $S=\{\ell,\ell+1,\dots, \ell+s-1\}$ where $\ell\sim \text{Unif}\{1,\dots, n\}$ is independent of $\W$ and elements of $S$ wrap around from $n$ to 1. Since $\Pr(j\in S) = s/n$ for each $j\in [n],$ Lemma~\ref{spectral_ber} gives $\s = s$. Finally, we choose the drop-and-recompute proxy 
\begin{align*}
    f^{(-S)}(\W_{-S}) = \frac{1}{n - |\H_S|}\sum_{i\notin \H_S} \psi_i
\end{align*} 
where $\H_S$ is the block of size $s+2M$ that pads $S$ on both sides, each with $M$ additional units. The resulting Neyman Jackknife variance estimator is 
\begin{align}
    \widehat V &= \frac{1}{s}\sum_{i=1}^n \left(f(\W) - f^{(-S_i)}(\W_{-S_i})\right)^2\label{eq:NW_start}
\end{align}
where $S_i$ is the block of size $s$ that begins at unit $i$ and runs clockwise. The following result relates the Neyman Jackknife variance estimator to the Newey--West estimator.

\begin{proposition}\label{NW_connection}
Fix an integer parameter $s\geq 1$ such that $s + 2M\leq n/2$. Let $\widehat V$ be the corresponding Neyman Jackknife variance estimator in \eqref{eq:NW_start} and let $\widehat{V}_{\normalfont\text{NW}}$ be the sample-centered Newey--West estimator with bandwidth $s+2M$ and circular Bartlett weights. Then,
\begin{align}
\widehat V = \frac{s+2M}{s}\cdot \left(\frac{n}{n-s-2M}\right)^2\cdot \widehat V_{\normalfont\text{NW}}.\label{eq:NW_relation}
\end{align}
\end{proposition}

In particular, $\widehat V$ is asymptotically equivalent to $\widehat V_{\normalfont\text{NW}}$ if $M \ll s \ll n.$ In the \hyperref[app]{Appendix} (Proposition~\ref{Mhop}), we discuss a similar connection with network HAC estimators \citep{kojevnikov2021bootstrap} for general $G_0$.

\section{Two-Stage Network Experiments}\label{network_exp}

Network experiments have become a standard tool for studying causal effects in the presence of interference (e.g., social networks). Here, we study the Neyman Jackknife for such experiments in the setting of neighborhood interference \citep{conflict2026}. We consider designs that first draw a random vector $\W\in\mathcal{X}^m,$ then determine a treatment vector $\Z\in\{0,1\}^n.$ Throughout this section, a treatment solely refers to an entry of $\Z,$ and we refer to entries of $\W$ as randomizations.

Formally, place $n$ units on a simple, undirected graph $G_0$. Each $i\in [n]$ receives treatment $Z_i\in\{0,1\}$ and we observe outcome $Y_i.$ We assume that $Y_i$ only depends on $\Z_{B_1(i)}$ where $B_r(i)=\{j\in [n]: d_{G_0}(i,j)\leq r\}.$ Arguably the most popular estimand in this setting is the global average treatment effect (GATE)
\begin{align*}
    \tau_{\normalfont\text{GATE}} = \frac{1}{n} \sum_{i = 1}^n\left(Y_i(\Z=\mathbf{1}_n) - Y_i(\Z=\mathbf{0}_n)\right).
\end{align*}
While the Bernoulli design for $\Z\in\{0,1\}^n$ is natural for estimating the average direct effect, it is suboptimal for estimating the GATE as the probability of $B_1(i)$ being in full-treatment or full-control decays exponentially in the degree of $i.$ In Section~\ref{network_designs}, we describe a Neyman Jackknife for network designs that address the above weakness of the Bernoulli design. In Section~\ref{bias_variance_sec}, we study the bias and variance of the Neyman Jackknife variance estimator for network experiments. In Section~\ref{covariate_sec}, we show that proxies involving covariates can sometimes be used to reduce the conservative bias of the Neyman Jackknife variance estimator.

\subsection{The Neyman Jackknife for Network Designs}\label{network_designs}
Many network designs are two-stage designs. In the first stage we draw the randomization vector $\W\in\mathcal{X}^m$, and in the second stage we let the treatment vector $\Z\in\{0,1\}^n$ be some deterministic function of $\W.$ Thus, all randomness comes from $\W$ and the treatment-effect estimator satisfies $\hat\tau = f(\W)$ for some $f.$ The key idea is to Neyman Jackknife the random vector $\W,$ rather than $\Z.$ 

The most transparent case of two-stage network designs is cluster-randomized designs \citep{ugander2013graph}. Partition $[n]$ into $m$ clusters $\mathcal{I}_1,\dots, \mathcal{I}_m.$ In the first stage, draw $\W\in\{0,1\}^m.$ In the second stage, let $Z_i = W_j$ for each $j\in [m]$ and $i\in\mathcal{I}_j$. Then, the results in Section~\ref{recipe_sec} apply verbatim if $\W$ follows the Bernoulli, completely randomized, or stratified design on $\{0,1\}^m.$ In particular, $Y_i$ only depends on randomizations in $\N_i\subseteq [m]$ where  $\N_i=\{j\in [m]: B_1(i)\cap \mathcal{I}_j\neq\emptyset\}.$\footnote{Recall that $\N_i$ was originally defined in Assumption~\ref{assu:PO}. Here, it is defined regarding randomizations, not treatments.}

In some network designs, the first-stage randomization $\W\in\mathcal{X}^m$ involves infinite $\mathcal{X}$, hence Theorem~\ref{Main} does not apply directly. We first extend the Neyman Jackknife to encompass $\W\in\mathcal{X}^m$ where $\mathcal{X}$ is infinite and $\W$ has independent entries.

\begin{theorem}\label{main_indep}
Let $\mathcal{X}$ be a (possibly infinite) measurable space and let $\W\in\mathcal{X}^m$ be a random vector with independent
entries. Let $\mu$ denote any subset-sampling rule that draws $S \subseteq [m]$ independent of $\mathbf{W}$ and let
\begin{align}
    \s &=m\cdot \min_{j\in [m]}\Pr(j\in S). \label{eq:explicit_eff}
\end{align}
For any proxies 
$f^{(-A)}(\W_{-A})$, define the Neyman Jackknife variance estimator
\begin{align}
  \widehat{V}
  &= \frac{m}{\s} \sum_{A \subseteq [m]} \Pr(S = A) \cdot \big(f(\W) - f^{(-A)}(\W_{-A})\big)^2. \label{eq:estimator_general}
\end{align}
Then, we have $\mathbb{E}[\widehat{V}] \geq \operatorname{Var}(f(\mathbf{W}))$ assuming $f(\W)$ has finite variance.
\end{theorem}

Thus, when $\mathcal{X}$ is infinite and $\W\in\mathcal{X}^m$ has independent entries, we require the subset-sampling rule to be independent of $\W$ and define $\s$ explicitly via \eqref{eq:explicit_eff}. In the remaining section, we describe randomized graph cluster randomization (RGCR) \citep{ugander2023}, the conflict graph design (CGD) \citep{conflict2026}, and a variant of the CGD which we term the requester design. For each design, we consider versions where the first-stage randomization $\W\in\mathcal{X}^n$ has $n$ independent entries. Let $V_1,\dots, V_K$ denote a fixed partition of $[n]$ and choose $S$ as a uniformly random part independent of $\W$. Then, $\s = n/K$ for each design by \eqref{eq:explicit_eff}, and we obtain the Neyman Jackknife variance estimator
\begin{align}
    \widehat{V} = \sum_{k=1}^K \left(f(\W) - f^{(-V_k)}(\W_{-V_k})\right)^2 \label{eq:two_stage_NJ}
\end{align}
by Theorem~\ref{main_indep}.  

When targeting the GATE under each design, the analyst typically chooses exposure indicators $T_i$ and $C_i$ such that $T_i Y_i = T_i Y_i(\Z=\mathbf{1}_n)$ and $C_i Y_i = C_i Y_i(\Z=\mathbf{0}_n).$ For instance, one can take $T_i = \mathbbm{1}\{\Z_{B_1(i)} = \mathbf{1}_{B_1(i)}\}$ and $C_i = \mathbbm{1}\{\Z_{B_1(i)} = \mathbf{0}_{B_1(i)}\}$. The analyst uses $(T_i, T_i Y_i)$ for all $i$ with $T_i=1$ and similarly $(C_i,C_iY_i)$ for all $i$ with $C_i = 1.$ We assume $T_i C_i = 0.$ Let $T_i$ only depend on randomizations (entries of $\W$) in $\N_{i,1}$ and $C_i$ only depend on randomizations in $\N_{i,0}.$ Given $S=A\subseteq [n],$ we can define the halos
\begin{align*}
    \H_{A,w} = \{i\in [n]: \N_{i,w}\cap A\neq\emptyset\}
\end{align*}
for $w\in\{0,1\}.$ When computing the proxy $f^{(-A)},$ one can delete $(T_i,T_iY_i)$ for all $i\in\H_{A,1},$ delete $(C_i,C_iY_i)$ for all $i\in \H_{A,0},$ then recompute the original treatment-effect estimator. For each design below, we describe its mechanism for generating $(\W,\Z),$ the indicators $(T_i,C_i),$ and the exposure sets $(\N_{i,1},\N_{i,0}).$

\begin{example}[RGCR with 1-hop-max clustering and independent randomization]\label{RGCR_ex}
\normalfont
    The first-stage randomization is $\W=((A_j,B_j))_{j=1}^n$ where $(A_j,B_j)$ are i.i.d., $A_j\sim U[0,1],$ and independently $B_j\sim \text{Ber}(p)$ with $p\in (0,1).$ In the second stage, for each $i\in [n]$ we assign its cluster label $c_i$ as $\text{argmax}_{j\in B_1(i)}\{A_j\}$.\footnote{One can fix a tie-breaking rule, but a tie occurs with probability zero.} Then, we assign the treatment $Z_i = B_{c_i}.$ For the exposure indicators, we use $T_i = \mathbbm{1}\{\Z_{B_1(i)} = \mathbf{1}_{B_1(i)}\}$ and $C_i = \mathbbm{1}\{\Z_{B_1(i)} = \mathbf{0}_{B_1(i)}\}$. Finally, we have $\N_{i,1} = \N_{i,0} = B_2(i).$
\end{example}

\begin{example}[CGD for estimating GATE]
\normalfont
    We are given a fixed importance ordering $\sigma:[n]\to [n]$ (a bijection where $\sigma(i)<\sigma(j)$ indicates that $i$ is more important than $j$). The first-stage randomization is $\W = (W_j)_{j=1}^n$ where each $W_j$ independently draws from the symbol set $\{e_1,e_0,*\}$ with fixed probabilities (see \citet{conflict2026} for more details). In the second stage, initialize $\Z=\mathbf{0}_n.$ For each $i\in [n]$ such that $W_i = e_1$ and $W_j = *$ for all $j\in B_2(i)\setminus \{i\}$ more important than $i,$ we update $\Z_{B_1(i)} = \mathbf{1}_{B_1(i)}$. For the exposure indicators, we use
\begin{align*}
    T_i &= \mathbbm{1}\{W_i = e_1, W_j = * \text{ for all } j\in B_2(i)\setminus\{i\}\text{ with }\sigma(j)<\sigma(i)\},\\
    C_i &= \mathbbm{1}\{W_i = e_0, W_j = * \text{ for all } j\in B_2(i)\setminus\{i\}\text{ with }\sigma(j)<\sigma(i)\}.
\end{align*}
Finally, we see that $\N_{i,1} = \N_{i,0} = \{j\in B_2(i): \sigma(j)\leq \sigma(i)\}.$
\end{example}

\begin{example}[Requester design with independent randomization]\label{requester_ex}
\normalfont
We describe a variant of the CGD, which we term the ``requester design." Here, each $i\in [n]$ requests $\Z_{B_1(i)}$ to be treated, and we grant the requests of randomly chosen units. Unlike the CGD, requests never conflict and no importance ordering is needed. Formally, the first-stage randomization is $\W=(W_j)_{j=1}^n$ where each $W_j$ independently draws from $\{e_1,*\}$. In the second stage, we initialize $\Z=\mathbf{0}_n$ and set $\Z_{B_1(i)} = \mathbf{1}_{B_1(i)}$ whenever $W_i = e_1.$ For the exposure indicators, we use
\begin{align*}
    T_i &= \mathbbm{1}\{W_i = e_1\},\\
    C_i &= \mathbbm{1}\{W_j = *\text{ for all }j\in B_2(i)\}.
\end{align*}
In contrast with the previous two designs, we see that $\N_{i,1}=\{i\}$ while $\N_{i,0} = B_2(i)$. In Section \ref{panel_exp}, we exploit this asymmetry of the requester design and apply it to the synthetic difference in differences (SDID) estimator in a setting with pre-treatment outcomes.    
\end{example}

\subsection{Bias-Variance Analysis}\label{bias_variance_sec}
For the network designs described in Examples \ref{RGCR_ex}--\ref{requester_ex}, one can use the Neyman Jackknife variance estimator 
\begin{align*}
    \widehat{V} = \sum_{k=1}^K \left(f(\W) - f^{(-V_k)}(\W_{-V_k})\right)^2
\end{align*}
from \eqref{eq:two_stage_NJ}. Classically, both the block Jackknife \citep{Kunsch1989} and cluster Jackknife \citep{mackinnon2023cluster} also drop an entire part and recompute the original estimator. The Neyman Jackknife estimator $\widehat{V}$ differs by dropping not only a part but also its halo (Definition~\ref{drop_and_recompute}). In this section, we analyze the bias and variance of $\widehat{V}$, with an eye towards choosing the number of parts $K$. Similar to prior work, we show that the bias is controlled by the size of the boundary of the partition \citep{leung2023network}, along with a heterogeneity term in the average treatment effects across parts \citep{abadie2023clustering}. Moreover, similar to the choice of block length for the block Jackknife \citep{hall1995blocking}, we show a bias-variance tradeoff in the number of parts $K$. 

As in Examples \ref{RGCR_ex}--\ref{requester_ex}, let the first-stage randomization $\W\in\mathcal{X}^n$ have $n$ independent entries. For the purpose of our bias-variance analysis, we focus on the IPW estimator $\hat\tau = f(\W)$ given by
\begin{align*}
    \hat\tau &= \frac{1}{n}\sum_{i=1}^n \psi_i= \frac{1}{n}\sum_{i=1}^n\left(\frac{T_i}{p_i} - \frac{C_i}{q_i}\right)Y_i
\end{align*}
where $p_i = \Pr(T_i=1)$ and $q_i = \Pr(C_i=1).$ Let $\N_i = \N_{i,1}\cup \N_{i,0}$ so that $\psi_i$ is measurable with respect to $\W_{\N_i}.$ We choose the drop-and-recompute proxy 
\begin{align*}
    f^{(-A)}(\W_{-A}) = \frac{1}{n - |\H_A|}\sum_{i\notin \H_A}\psi_i
\end{align*}
where $\H_A = \{i\in [n]: \N_i\cap A\neq\emptyset\}$ is the halo of $A.$\footnote{Instead of dropping units in $\H_{A,1}$ and $\H_{A,0}$ separately as in Section~\ref{network_designs}, we drop $\H_A$ for a simpler analysis.} 

Finally, we define some notation. Write $D_k = \H_{V_k}$ for each $k\in [K]$ and let
\begin{align*}
     \tau_i = Y_i(\Z=\mathbf{1}_n) - Y_i(\Z=\mathbf{0}_n),\qquad \overline{\tau}_{-k}=\frac{1}{n-|D_k|}\sum_{j\notin D_k}\tau_j,
\end{align*}
and $\tau = \tau_{\normalfont\text{GATE}}.$ The following result characterizes the bias and variance of $\widehat{V}$.

\begin{lemma}\label{bias_var_delete}
     Let $\nu_i = |\{k\in[K]: \N_i\cap V_k\neq\emptyset\}|$ denote the number of parts overlapping with $\N_i.$ Let $h_\star = \max_{i\in [n]}|\H_{\{i\}}|$ and assume $|\psi_i|\leq C$. Then, 
    \begin{align*}
        0\leq \E[\widehat{V}] - \Var(\hat\tau) - H_\mathcal{P}\leq \frac{C^2 h_\star}{n^2}\left(\sum_{i=1}^n (\nu_i - 1) + \sum_{k=1}^K \frac{|D_k|^2}{n-|D_k|}\right)
    \end{align*}
    where
    \begin{align*}
        H_{\mathcal{P}} = \sum_{k=1}^K \left(\overline{\tau}_{-k}-\tau\right)^2=  \frac{1}{n^2}\sum_{k=1}^K \left(\sum_{i\in D_k}(\tau_i - \overline{\tau}_{-k})\right)^2.
    \end{align*}
    Also write $\N_\star = \max_{i\in [n]}|\N_i|,$ $\nu_{\star} = \max_{i\in [n]}\nu_i,$ and $d_{\star} = \max_{k\in [K]}|D_k|.$ Then, we have 
    \begin{align*}
        \Var(\widehat{V})\leq 8C^2 h_\star \cdot\frac{\N_\star\nu_{\star} d_{\star}}{(n-d_{\star})^2}\cdot \E[\widehat{V}].
    \end{align*}
\end{lemma}
We refer to $H_\mathcal{P}$ as the heterogeneity term, analogous to the heterogeneity term in Neyman's classical analysis under SUTVA. With bounded degrees and outcomes along with balanced parts, we get the following simplification.

\begin{corollary}\label{bounded_deg_bv}
    Assume that $\N_i\subseteq B_2(i)$, the maximum degree of $G_0$ is $O(1)$, and $|\psi_i|\leq C = O(1)$. Also assume $d_{\star} = O(n/K)$ and $n-d_{\star} = \Omega(n).$ Then,
    \begin{align*}
        0\leq \E[\widehat{V}] - \Var(\hat\tau) - H_{\mathcal{P}}&\lesssim \frac{|\partial_{\mathcal{P}}|}{n^2} + \frac{1}{nK},\\
        \Var(\widehat{V})&\lesssim \frac{\E[\widehat{V}]}{nK}
    \end{align*}
    where $\partial_{\mathcal{P}} = \{i\in [n]: \nu_i\geq 2\}.$
\end{corollary}
We refer to $\partial_{\mathcal{P}}$ as the boundary of the partition. Note that the condition $\N_i\subseteq B_2(i)$ is satisfied for Examples \ref{RGCR_ex}--\ref{requester_ex}. In the remaining section, we assume the setting of Corollary~\ref{bounded_deg_bv}. We also assume that $\Var(\hat\tau)\asymp 1/n$. First, Corollary~\ref{bounded_deg_bv} reveals a tradeoff in choosing the number of parts $K.$ On one hand, smaller $K$ corresponds to smaller $|\partial_{\mathcal{P}}|.$ On the other hand, smaller $K$ leads to larger $1/(nK).$  A simple heuristic for choosing $K$ (based on Corollary~\ref{bounded_deg_bv}) minimizes
\begin{align}
    \frac{|\partial_{\mathcal{P}}|^2}{n^2} + \frac{1}{K}\label{eq:optimize}
\end{align}
which is derived from noting
\begin{align*}
    \left(\frac{|\partial_{\mathcal{P}}|}{n^2} + \frac{1}{nK}\right)^2 + \frac{1}{n^2K} \lesssim\frac{1}{n^2}\cdot \left(\frac{|\partial_{\mathcal{P}}|^2}{n^2} + \frac{1}{K}\right).
\end{align*}
Let the partition consist of balanced, compact parts. If $G_0$ is 1-dimensional as in time series, then we get $|\partial_\mathcal{P}|\asymp K$ and thus \eqref{eq:optimize} is minimized when $K\asymp n^{2/3}.$ On the other hand, if $G_0$ is $2$-dimensional as in planar networks, then we get $|\partial_\mathcal{P}|\asymp \sqrt{nK},$ suggesting $K\asymp \sqrt{n}$ to minimize \eqref{eq:optimize}. We use these heuristics for our experiments in Sections~\ref{temp_exp} and \ref{spatial_exp}.

The above heuristic does not involve the heterogeneity term $H_\mathcal{P},$ which we separately consider here. First assume that the residuals $\delta_i^{(-V_k)} = \tau_i-\overline{\tau}_{-k}$ cancel within $D_k$ so that $|\sum_{i\in D_k}\delta_i^{(-V_k)}|\leq C_1 \sqrt{|D_k|}$ for some constant $C_1$ (independent of $n,k$). Then, we get
\begin{align*}
    H_\mathcal{P}\leq \frac{C_1^2}{n^2} \sum_{k=1}^K |D_k| \lesssim \frac{1}{n}.
\end{align*}
On the other hand, if the residuals $\delta_i^{(-V_k)}$ pile up in $D_k$ so that $|\sum_{i\in D_k}\delta_i^{(-V_k)}|\geq C_2|D_k|,$ we get
\begin{align*}
    H_{\mathcal{P}}\geq \frac{C_2^2}{n^2} \sum_{k=1}^K |D_k|^2\geq \frac{C_2^2}{n^2K}\left(\sum_{k=1}^K|D_k|\right)^2 \gtrsim\frac{1}{K}.
\end{align*}
Thus, if $\tau_i$ is smooth at the level of parts, the heterogeneity term can be of larger order than the target variance; one should aim to choose parts at a scale larger than the smoothness of $\tau_i$. If this is not possible, one can use covariates to capture the spatial heterogeneity. We discuss this further in Section~\ref{covariate_sec}.

Finally, Corollary~\ref{bounded_deg_bv} gives $\Var(\widehat{V})/\E[\widehat{V}]^2 = O(1/K)$ and thus $\widehat{V}/\E[\widehat{V}]\limp 1$ if $K\to\infty.$ Combined with the dependency graph CLT \citep{ross2011stein}, we obtain asymptotic coverage.

\begin{proposition}\label{coverage}
    Consider the setting of Corollary~\ref{bounded_deg_bv} and assume $\Var(\hat\tau)\gtrsim 1/n.$ If $K\to\infty,$ then for any $\alpha\in(0,1),$ we have
    $$\liminf_{n\to\infty}\Pr\Big(\tau \in \big[\hat\tau - z_{1-\alpha/2}\sqrt{\widehat V},\ \hat\tau + z_{1-\alpha/2}\sqrt{\widehat V}\big]\Big) \ge 1-\alpha.$$
\end{proposition}

\subsection{Computable Proxies via Covariates}\label{covariate_sec}

By Theorem \ref{main_indep}, the Neyman Jackknife estimator $\widehat{V}$ from \eqref{eq:two_stage_NJ} is conservative in expectation for any choice of the partition and computable proxies. However, better choices can lead to tighter estimates of the variance. In the previous section, we saw that a careful choice of the number of parts $K$ can be helpful. Here, we explore another pathway and show that proxies involving covariates can alleviate the heterogeneity term $H_\mathcal{P}$, and thus reduce the conservative bias of $\widehat{V}.$

We adopt the exact setting of Section \ref{bias_variance_sec}, except we now allow a general proxy when computing $\widehat{V}$ in \eqref{eq:two_stage_NJ}. Applying the orthogonal decomposition as in \eqref{eq:decomp} to $\widehat{V},$ we get 
$$\E[\widehat{V}] = \UB + \PE$$ where 
\begin{align*}
    \UB &= \sum_{k=1}^K\E\left[ \Big(f(\W) - \E[f(\W)\mid \W_{-V_k}]\Big)^2\right],\\
    \PE &= \sum_{k=1}^K\E\left[ \Big(f^{(-V_k)}(\W_{-V_k}) - \E[f(\W)\mid \W_{-V_k}]\Big)^2\right].
\end{align*}
The first term $\UB$ is the oracle upper bound for the variance, where $\UB\geq \Var(\hat\tau)$ is the Poincar\'{e} inequality. The second term $\PE$ is the proxy error. We first bound $\UB - \Var(\hat\tau),$ which does not depend on the choice of proxy. 

\begin{proposition}\label{oracle_bias_prop}

    Let $\nu_i = |\{k\in[K]: \N_i\cap V_k\neq\emptyset\}|$ denote the number of parts overlapping with $\N_i.$ Let $h_\star = \max_{i\in [n]}|\H_{\{i\}}|$ and assume $|\psi_i|\leq C$. Then, we have
    \begin{align*}
        0\leq \UB - \Var(\hat\tau)\leq \frac{C^2h_\star}{n^2}\sum_{i=1}^n (\nu_i-1).
    \end{align*}
\end{proposition}
Thus, we see that the oracle bias is controlled by the boundary term that also appears in Lemma~\ref{bias_var_delete}, regardless of the choice of proxy. We now study the proxy error. Given $S=A,$ we want $f^{(-A)}(\W_{-A})$ to target $\E[f(\W)\mid\W_{-A}].$ A simple calculation gives 
\begin{align}
    \E[f(\W)\mid\W_{-A}] = \frac{1}{n}\sum_{i\notin \H_A}\psi_i + \frac{1}{n}\sum_{i\in \H_A} \left(\frac{\tilde{p}_i}{p_i}\cdot Y_i(\Z=\mathbf{1}_n) - \frac{\tilde{q}_i}{q_i}\cdot Y_i(\Z=\mathbf{0}_n)\right) \label{eq:CE_form}
\end{align}
where $\tilde{p}_i = \E[T_i\mid \W_{-A}]$ and $\tilde{q}_i = \E[C_i\mid \W_{-A}]$. As both $\tilde{p}_i$ and $\tilde{q}_i$ are computable from $\W_{-A}$, \eqref{eq:CE_form} suggests proxies of the form 
\begin{align}\label{eq:proxy_form}
    f^{(-A)}(\W_{-A}) = \frac{1}{n}\sum_{i\notin \H_A}\psi_i + \frac{1}{n}\sum_{i\in \H_A} \left(\frac{\tilde{p}_i}{p_i}\cdot \hat\mu_{i,1}^{(-A)} - \frac{\tilde{q}_i}{q_i}\cdot \hat\mu_{i,0}^{(-A)}\right)
\end{align}
where $\hat\mu_{i,1}^{(-A)},\hat\mu_{i,0}^{(-A)}$ are functions of $\W_{-A}.$

For instance, $\hat\mu_{i,1}^{(-A)}$ may be obtained by regressing $Y_j$ on deterministic covariates $X_j$ for $j\notin \H_A$ with $T_j = 1.$ The drop-and-recompute proxy from Section \ref{bias_variance_sec} can also be written similar to \eqref{eq:proxy_form}, i.e.,
$$\frac{1}{n - |\H_A|}\sum_{i\notin \H_A}\psi_i = \frac{1}{n}\sum_{i\notin \H_A}\psi_i + \frac{1}{n}\sum_{i\in \H_A} \left(\hat\mu_1^{(-A)} - \hat\mu_0^{(-A)}\right)$$
with 
\begin{align}
    \hat\mu_1^{(-A)} = \frac{1}{n-|\H_A|}\sum_{j\notin \H_A}\frac{T_j}{p_j}\cdot Y_j,\qquad \hat\mu_0^{(-A)} = \frac{1}{n-|\H_A|}\sum_{j\notin \H_A}\frac{C_j}{q_j}\cdot Y_j.\label{drop_recompute_mu_hat}
\end{align} 
Note that the drop-and-recompute proxy replaces $\tilde{p}_i/p_i$ and $\tilde{q}_i/q_i$ in \eqref{eq:proxy_form} with one. 

Assuming deterministic $\hat\mu_{i,w}^{(-A)}$, the following result shows that the heterogeneity term comes from $\PE.$ 
\begin{proposition}\label{proxy_error_prop}
    Let $\nu_i = |\{k\in[K]: \N_i\cap V_k\neq\emptyset\}|$ denote the number of parts overlapping with $\N_i$ and let $h_\star = \max_{i\in [n]}|\H_{\{i\}}|$. For deterministic imputations $\hat\mu_{i,w}^{(-A)}$, write $$\hat{\psi}_{i,k} = \frac{T_i}{p_i}\left(Y_i(\Z=\mathbf{1}_n) - \hat\mu_{i,1}^{(-V_k)}\right) - \frac{C_i}{q_i}\left(Y_i(\Z=\mathbf{0}_n) - \hat\mu_{i,0}^{(-V_k)}\right)$$ and assume $|\hat{\psi}_{i,k}|\leq C'.$ Choosing proxies $f^{(-A)}$ according to \eqref{eq:proxy_form}, we have 
    \begin{align*}
        0\leq \PE - H_{\mathcal{P},\hat\mu}\leq \frac{(C')^2h_\star}{n^2}\sum_{i=1}^n (\nu_i-1)
    \end{align*}
    where 
    $$H_{\mathcal{P},\hat\mu} = \frac{1}{n^2}\sum_{k=1}^K \left(\sum_{i\in D_k}\left(\tau_i - \hat\tau_{i,-k}\right)\right)^2$$
    and $\hat\tau_{i,-k} = \hat\mu_{i,1}^{(-V_k)} - \hat\mu_{i,0}^{(-V_k)}.$
\end{proposition}

In practice, the imputations $\hat\mu_{i,w}^{(-V_k)}$ are random. Thus, we expect Proposition~\ref{proxy_error_prop} to capture the main effect of covariates on $\mathrm{Err}_{\mathrm{proxy}}$, while the randomness of $\hat\mu^{(-V_k)}_{i,w}$ contributes an additional estimation-error term that we do not analyze here. To conclude, the heterogeneity term $H_{\mathcal{P},\hat\mu}$ involves the residuals
$$\delta_i^{(-V_k)} = \tau_i - \hat\tau_{i,-k},$$
instead of $\tau_i - \overline{\tau}_{-k}$ as in $H_\mathcal{P}$. In particular, the imputations $\hat\tau_{i,-k}$ may differ for each $i,$ possibly involving covariates $X_i.$ Thus, if the covariates capture the spatial heterogeneity of treatment effects, then the residuals $\delta_i^{(-V_k)}$ may cancel within $D_k$ instead of piling up, and $H_{\mathcal{P},\hat\mu}$ may be much smaller than $H_\mathcal{P}$. 

\section{Numerical Experiments}

In this section, we illustrate the Neyman Jackknife via numerical experiments. In Section \ref{temp_exp}, we consider switchback experiments for the GATE with carryovers decaying over time. We use block randomization with burn-ins and the H\'{a}jek estimator. In Section \ref{spatial_exp}, we consider network experiments for the GATE under neighborhood interference. We use RGCR and the H\'{a}jek estimator. In Section \ref{panel_exp}, we again consider a network experiment for the GATE, but in a panel data setting with pre-treatment outcomes. We use the requester design and SDID.

For the first two settings, we compare against the prior methods described in Section~\ref{related} (see Section~\ref{overview_prior_sec_appendix} for more details). This includes Young's inequality \citep{Aronow}; optimized PSD bound \citep{harshaw2026variance}; spectral
norm\footnote{This method was originally proposed for the conflict graph design and IPW estimator, where it can be much tighter.} \citep{conflict2026}; Cauchy--Schwarz \citep{lu2025design}; dependency graph HAC \citep{leung2022ani};
dependency graph HAC with PSD correction \citep{gao2025network}. As the treatment-effect estimator for the first two settings is the H\'{a}jek estimator, we obtain baselines by following \citet[Section 7.2]{Aronow} and applying each prior method to the linearized H\'{a}jek estimator. The resulting baseline variance estimators are asymptotically conservative, with the dependency graph HAC without PSD correction as an exception (see Section~\ref{overview_prior_sec_appendix}).

For the third setting, to the best of our knowledge, conservative variance estimation for SDID in the presence of spillovers has no solution in the literature as SDID is nonlinear in the outcomes. Thus, we use the SDID jackknife and bootstrap \citep{sdid} as heuristic baselines. These methods were originally proposed in a model-based setup, rather than a design-based setup. Moreover, they were proposed in a setup without spillovers, whereas our third setting assumes neighborhood interference.

\subsection{Switchback Experiments}\label{temp_exp}

In a switchback experiment, a treatment is repeatedly turned on and off over time for a whole system. Early work goes back to \citet{brandt1938tests} and \citet{cochran1941double}. More recently, switchback experiments have been widely used to study market-level interventions, for instance comparing two different pricing schemes. The main difficulty of analyzing switchback experiments lies in carryovers, where the treatment at one time can affect the outcome of the system at a later time. The literature typically restricts such carryovers by assuming finite-range dependence \citep{bojinov2023design} or decaying carryovers \citep{HuWager2022}. 

In our experiment, we build on the design-based approach of \citet{bojinov2023design} along with the decaying carryover assumption of \citet{HuWager2022}. We first describe the design-based setup. Consider $T=10{,}000$ time units and let $Z_t\in\{0,1\}$ denote the treatment of $t\in [T].$ In other words, $Z_t$ is the treatment of the system at time $t.$ Next, let $Y_t = Y_t(Z_1,\dots, Z_t)$ denote the observed outcome at time $t,$ where the underlying potential outcomes are fixed. In the \hyperref[app]{Appendix} (Section~\ref{temp_exp_details}), we describe our DGP for decaying carryovers. We illustrate our potential outcomes in Figure \ref{fig:switch_dgp}.

\begin{figure}[t]
\centering
\includegraphics[width=\textwidth]{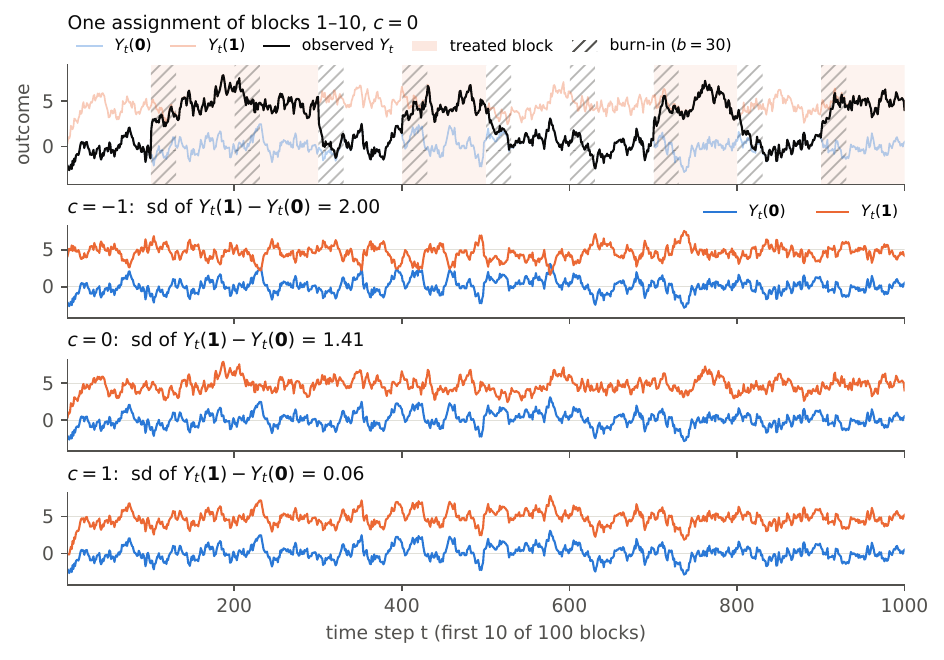}
\caption{Potential outcomes for a switchback experiment in one realized temporal world. In each plot, $c\in [-1,1]$ is approximately the sample correlation between the realized potential outcomes $(Y_t(\mathbf{0}))_{t=1}^T$ and $(Y_t(\mathbf{1}))_{t=1}^T$. When the treatment switches, the black curve rises (or drops) to the underlying orange (or blue) curve. This is due to the decaying carryovers in our DGP.}
\label{fig:switch_dgp}
\end{figure}

Our target estimand is the GATE
\begin{align*}
    \tau = \frac{1}{T}\sum_{t=1}^T\left(Y_t(\mathbf{Z}=\mathbf{1}_T) - Y_t(\mathbf{Z}=\mathbf{0}_T)\right).
\end{align*}
We perform inference by following the approach of \citet{HuWager2022}. First, we run the switchback experiment via block Bernoulli randomization, which is a two-stage design. Divide the $T$ time steps into $m=100$ blocks of size $T/m = 100.$ In the first stage, we independently draw $W_j\sim\text{Ber}(1/2)$ for each $j\in [m].$ In the second stage, we let $Z_t = W_j$ for all $t$ in block $j.$

Our point estimator is also based on \citet{HuWager2022}. Since there are carryovers, we need to be careful about using outcomes right after the system has switched from treated to control (or vice versa). As we assume decaying carryovers, we discard the first $b=30$ outcomes after such a switch (i.e., we assume the strength of a carryover is negligible after 30 timesteps). Formally, for each $j\geq 2,$ define the first $b=30$ units of block $j$ as burn-in units. We refer to the remaining units as focal units. For each $t$ in block $j\geq 2,$ define $\N_t=\{j-1,j\}$ if $t$ is a burn-in unit and $\N_t = \{j\}$ otherwise. For each $t$ in block 1, set $\N_t = \{1\}.$ We estimate the GATE using the H\'{a}jek estimator
\begin{align*}
    \hat\tau =  \hat\mu_1 - \hat\mu_0 = \frac{\sum_{t=1}^TT_tY_t/p_t}{\sum_{t=1}^T T_t/p_t} - \frac{\sum_{t=1}^T C_tY_t/q_t}{\sum_{t=1}^T C_t/q_t} 
\end{align*}
where $T_t$ indicates that $W_j=1$ for all $j\in\N_t$, $C_t$ indicates that $W_j=0$ for all $j\in\N_t,$ and $p_t,q_t$ are the respective propensities. This is analogous to the bias-corrected estimator proposed in \citet{HuWager2022}.

For inference, we consider the block Jackknife proposed by \citet[Algorithm 2]{HuWager2022} as a baseline along with the prior methods from Section~\ref{related}. This block Jackknife recomputes $\hat\tau$ after dropping a block of time units, rather than randomization units.\footnote{The block Jackknife of \citet[Algorithm 1]{HuWager2022} was proved valid for a setting where the burn-in periods make carryovers asymptotically negligible. The theory for their Algorithm 2 was left open.} Our method of interest is the Neyman Jackknife variance estimator, which we apply to the first-stage randomization $\W = (W_1,\dots, W_m).$ Choose the update set as $S=\{\ell,\ell+1,\dots, \ell+s-1\}\subseteq [m]$ where $\ell\sim \text{Unif}\{1,\dots, m\}$ is independent of $\W$ and elements of $S$ wrap around from $m$ to 1. We set $s = 5$ so that $m/s = 20 \approx m^{2/3}$, based on the discussion following \eqref{eq:optimize}. By Lemma~\ref{spectral_ber}, we have $\s = s.$ 

For the proxy, we define the halo $\H_S=\{t\in [T]: \N_t\cap S\neq\emptyset\}.$ For instance, if $j\in S$ then all time units in block $j$ and all burn-in units in block $j+1$ are included in $\H_S.$ We use the drop-and-recompute proxy
\begin{align*}
    f^{(-S)} =   \frac{\sum_{t\notin \H_S} T_tY_t/p_t}{\sum_{t\notin \H_S} T_t/p_t} - \frac{\sum_{t\notin \H_S} C_tY_t/q_t}{\sum_{t\notin \H_S}C_t/q_t}.
\end{align*}
Since the carryovers do not vanish completely, $f^{(-S)}$ is not measurable with respect to $\W_{-S}$. However, the carryovers decay quickly enough so that $f^{(-S)}$ is almost measurable with respect to $\W_{-S}.$ In the \hyperref[app]{Appendix} (Section~\ref{robust_sec}), we include a formal section on approximate measurability and asymptotic conservativeness. 

Our experiment results are given in Table~\ref{fig:switch_results}. We see that the Neyman Jackknife is conservative and also competitive in width with existing methods. When $c=1,$ the HAC estimator without PSD correction along with the block Jackknife (for time units) are anti-conservative. The distinction between the block Jackknife and Neyman Jackknife is particularly informative. For each $j=1,\dots,m-2,$ the block Jackknife drops outcomes in the focal part of block $j$, both parts of block $j+1,$ and the burn-in part of block $j+2.$ This is similar to the halo of dropping $W_j$ and $W_{j+1},$ but the difference is that the burn-in part of block $j$, while depending on $W_j,$ is kept. Thus, the ``proxy" used by the block Jackknife is not measurable with respect to $\W_{-\{j,j+1\}}$. We provide a formal description of the above block Jackknife in Section~\ref{temp_exp_details}.

\begin{table}[t]
\centering
\includegraphics[width=\textwidth]{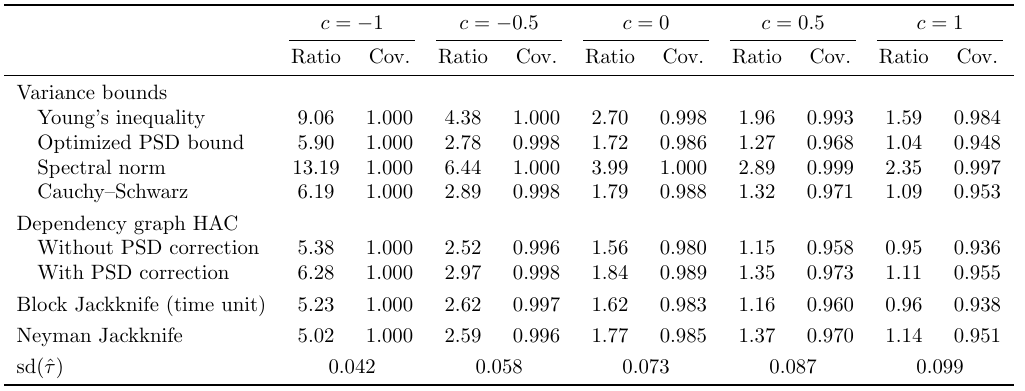}
\caption{Switchback experiment results. For each variance estimator $\widehat{V}$, Ratio is $\E[\widehat{V}]/\Var(\hat\tau)$ and Cov. is the coverage of the nominal $95\%$ interval. Columns are marked by the arm correlation $c\in [-1,1].$ The potential outcomes are drawn once and fixed. The values are computed over $200{,}000$ randomizations.}
\label{fig:switch_results}
\end{table}

\subsection{Spatial Network Experiments}\label{spatial_exp}

In many applications, units interact due to various reasons, whether it be personal connections or spatial proximity. Such interactions can be captured by an underlying network, where an edge between two units indicates that they interact. A network experiment is a randomized experiment on such a network. In the simplest setting, one assumes that units interact if and only if there is an edge between them. This model is termed neighborhood interference, and has been well studied in the design-based literature \citep{ugander2023, conflict2026}. The main challenge in this setting is to obtain precise inference, despite the outcomes being entangled by the network.

Here, we consider a network experiment, following the approach of \citet{ugander2023}. We first describe our design-based setup. Consider $n = 1{,}000$ units, placed independently and uniformly on $[0,1]^2$. Two units are connected when they are within distance $0.045$, giving an average degree of about $5.9$. The resulting graph is a random geometric graph, which we draw once and consider fixed. Let $Z_i\in\{0,1\}$ denote the treatment for $i\in [n]$ and let $Y_i = Y_i(\mathbf{Z}_{B_1(i)})$ denote the observed outcome for unit $i$, where $B_1(i)$ is the closed 1-hop neighborhood of $i$. We assume fixed potential outcomes along with observed covariates $x_{0i},x_{1i}$ for each $i\in [n].$ Our DGP for the potential outcomes is described in the \hyperref[app]{Appendix} (Section~\ref{spatial_exp_details}). In particular, our treatment effects are spatially correlated, as illustrated in Figure~\ref{fig:spatial_RGCR_DGP}.

\begin{figure}[t]
\centering
\includegraphics[width=\textwidth]{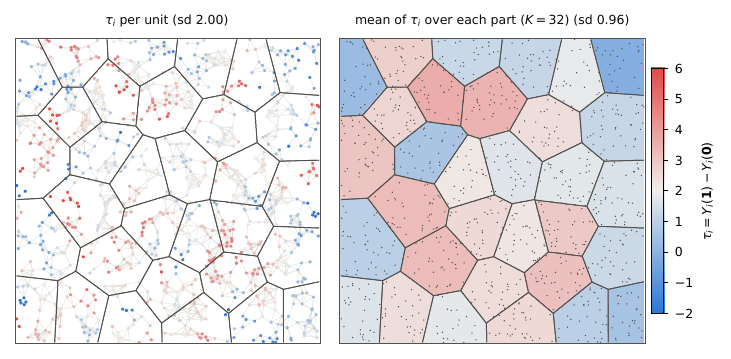}
\caption{Heterogeneity of individual treatment effects $\tau_i$ in one realized spatial world. The partition indicates the update sets used by the Neyman Jackknife (with $K = 32 \approx \sqrt{n}$ parts). The partition is chosen and fixed before the treatment assignment. The figure corresponds to arm correlation $c=-1,$ which is the most heterogeneous case.}
\label{fig:spatial_RGCR_DGP}
\end{figure}

Our causal estimand is the GATE
\begin{align*}
    \tau = \frac{1}{n}\sum_{i=1}^n \left(Y_i(\mathbf{Z} = \mathbf{1}_n) - Y_i(\mathbf{Z} = \mathbf{0}_n)\right).
\end{align*}
To target the GATE, we use the design-estimator pair studied by \citet{ugander2023}. Our design is RGCR with 1-hop-max clustering and independent randomization (with $p=1/2$), as described in Example~\ref{RGCR_ex}. Our point estimator is the H\'{a}jek estimator 
\begin{align*}
    \hat\tau = \frac{\sum_{i=1}^n T_iY_i/p_i}{\sum_{i=1}^n T_i/p_i} - \frac{\sum_{i=1}^n C_iY_i/q_i}{\sum_{i=1}^n C_i/q_i}
\end{align*}
where $T_i$ indicates that $B_1(i)$ is fully treated, $C_i$ indicates that $B_1(i)$ is fully in control, and $p_i,q_i$ are the respective propensities.\footnote{For RGCR, the propensities $p_i,q_i$ need to be estimated via Monte Carlo. We use $40{,}000$ Monte Carlo clusterings.} 

\begin{table}[t]
\centering
\includegraphics[width=\textwidth]{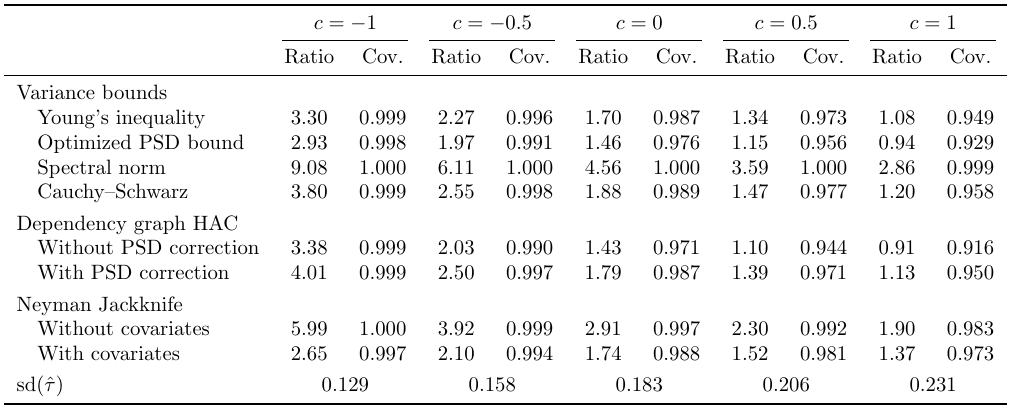}
\caption{Spatial network experiment results. For each variance estimator $\widehat{V}$, Ratio is $\E[\widehat{V}]/\Var(\hat\tau)$ and Cov. is the coverage of the nominal $95\%$ interval. Columns indicate the arm correlation $c.$ The potential outcomes are drawn once and fixed. The values are computed over $10{,}000$ randomizations.}
\label{fig:spatial_RGCR_results}
\end{table}

For variance estimation, we again take the prior methods in Section~\ref{related} as baselines. Our method of interest is the Neyman Jackknife variance estimator in \eqref{eq:two_stage_NJ}, with a partition $V_1,\dots, V_K$ obtained by using k-means on the node positions. We choose $K=32\approx\sqrt{n}$ according to the discussion following \eqref{eq:optimize}. Finally, we consider two proxies, one without covariates and the other with covariates. The first is the drop-and-recompute proxy 
\begin{align*}
    f^{(-S)} = \frac{\sum_{j\notin \H_S}T_jY_j/p_j}{\sum_{j\notin \H_S}T_j/p_j} - \frac{\sum_{j\notin \H_S}C_jY_j/q_j}{\sum_{j\notin \H_S}C_j/q_j}
\end{align*}
where $\H_S = \{i\in [n]: \N_i\cap S\neq\emptyset\}$ and $\N_i=B_2(i)$. The second is the proxy with covariates $x_{0i},x_{1i}$, based on Section \ref{covariate_sec}. We provide the explicit form in Section~\ref{spatial_exp_details}.

The experiment results are in Table~\ref{fig:spatial_RGCR_results}. Without covariates, the Neyman Jackknife is more conservative than most existing methods. This is partly due to $D_k=\H_{V_k}$ being the closed 2-hop neighborhood of $V_k.$ In our graph, the average $|D_k|$ is roughly $60.6$ while the average $|V_k|$ is only $31.2$. With covariates, however, the Neyman Jackknife yields much tighter estimates, as discussed in Section~\ref{covariate_sec}. For either proxy, the Neyman Jackknife is provably finite-sample conservative in expectation. In contrast, when $c=1$, we see that the optimized PSD bound and the HAC estimator without PSD correction are anti-conservative.

\subsection{SDID with Spillovers}\label{panel_exp}

\begin{figure}[t]
\centering
\includegraphics[width=\textwidth]{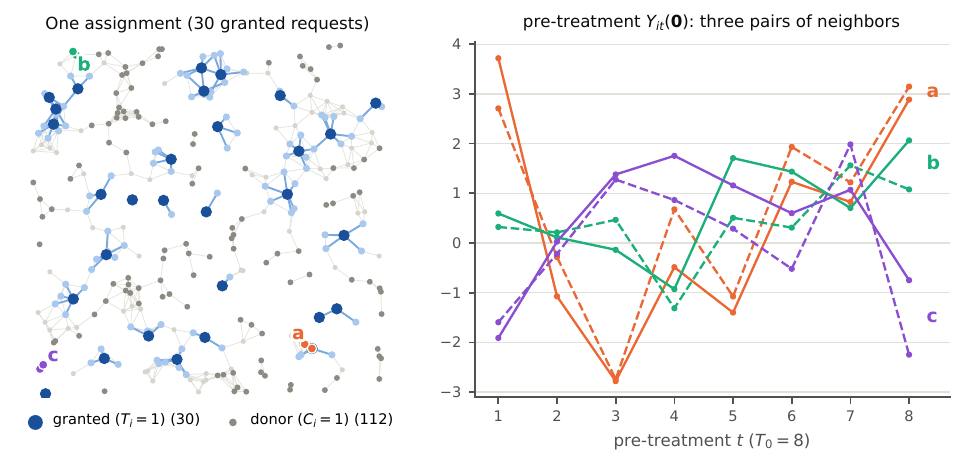}
\caption{One realized world of the panel setting. Left: one draw of the requester design with 30 granted requests. Right: pre-treatment outcomes $Y_{it}(\mathbf 0)$ of three pairs of neighboring units (solid and dashed lines, marked in the left panel). Neighbors share the spatial fields that drive the untreated outcomes and move together, while distant pairs respond differently.}
\label{fig:panel_dgp}
\end{figure}

Synthetic control \citep{abadie2010synthetic} and synthetic difference in differences \citep{sdid} are standard methods for treatment-effect estimation when there are relatively few treatments. Such methods assume a panel data setting, where we observe pre-treatment outcomes for each unit. While these methods were originally developed in the observational setting, recent work has also considered synthetic control from a design-based perspective \citep{bottmer2024design}. In this section, we go a step further and apply the SDID estimator to a network experiment, namely a setting with spillovers.

Our design-based setup is similar to Section~\ref{spatial_exp}. There are $n = 300$ units, placed independently and uniformly on $[0,1]^2$. Two units are connected when their distance is among the $2n$ smallest pairwise distances, so the graph has exactly $2n$ edges and average degree $4$. We draw the graph once and consider it fixed. There are $T_0=8$ pre-treatment periods and one post-treatment period. In the post-treatment period, each $i\in[n]$ receives treatment $Z_i\in\{0,1\}$ and we observe outcome $Y_{i,T_0+1} = Y_{i,T_0+1}(\mathbf{Z}_{B_1(i)})$. We describe the DGP for the pre-treatment outcomes and post-treatment potential outcomes in Section~\ref{panel_exp_details}. An illustration of the pre-treatment outcomes is given in Figure~\ref{fig:panel_dgp}.

\begin{table}[t]
\centering
\includegraphics[width=\textwidth]{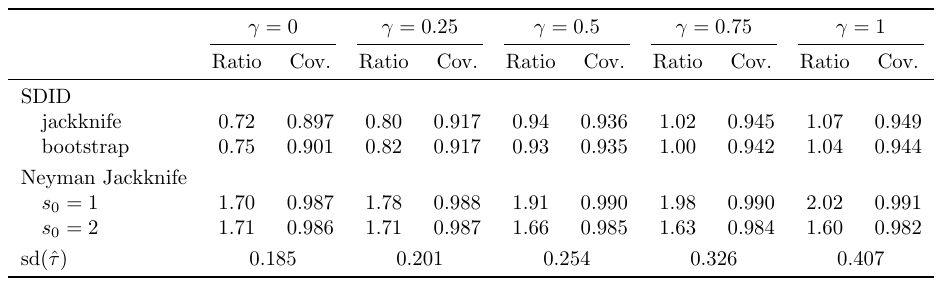}
\caption{Panel results for the SDID estimator. For each variance estimator $\widehat{V}$, Ratio is $\E[\widehat{V}]/\Var(\hat\tau)$ and Cov. is the coverage of the nominal $95\%$ interval. Columns indicate the spillover strength $\gamma.$ The potential outcomes are drawn once and fixed. The ratios and coverage are computed over $10{,}000$ randomizations. A separate $20{,}000$ randomizations are used to compute the standard deviation.}
\label{fig:panel_results}
\end{table}

To target the GATE
\begin{align*}
    \tau = \frac{1}{n}\sum_{i=1}^n \left(Y_{i,T_0+1}(\mathbf{Z}=\mathbf{1}_n) - Y_{i,T_0+1}(\mathbf{Z}=\mathbf{0}_n)\right),
\end{align*}
we use a completely randomized version of the requester design described in Example \ref{requester_ex}. Instead of independent randomization, we completely randomize the first-stage randomization $\W=(W_j)_{j=1}^n$ so that $|\{j:W_j=e_1\}|=n_1=30.$ As shown in Figure~\ref{fig:panel_dgp}, a typical draw of the requester design results in many more units in the control arm than the treated arm. The main idea is that the treated exposures are roughly independent, while the abundant units in the control arm can serve as donors. 

Our point estimator is the SDID estimator \citep{sdid}
\begin{align*}
    \hat\tau = \frac{1}{n_1}\sum_{i: T_i=1}\left(Y_{i,T_0+1} - \hat{Y}_{i,T_0+1}(\mathbf{0}_{B_1(i)})\right)
\end{align*}
where $T_i = \mathbbm{1}\{W_i = e_1\}$, $C_i = \mathbbm{1}\{B_1(i)\text{ is fully in control}\}$ and 
\begin{align*}
    \hat Y_{i,T_0+1}(\mathbf{0}_{B_1(i)}) = \sum_{t=1}^{T_0} \lambda_t Y_{i,t} + \sum_{j: C_j = 1}\omega_j\left(Y_{j,T_0+1} - \sum_{t=1}^{T_0}\lambda_t Y_{j,t}\right).
\end{align*}
The weights $\omega = (\omega_j)_{j:C_j=1}$ and $\boldsymbol{\lambda} = (\lambda_t)_{t=1}^{T_0}$ are obtained via 
\begin{align*}
    (\omega_0, \omega) &= \underset{\omega_0 \in \mathbb{R},\ \omega \in \Delta}{\mathrm{argmin}}\sum_{t=1}^{T_0} \Bigl( \omega_0 + \sum_{j : C_j = 1} \omega_j Y_{jt} - \frac{1}{n_1}\sum_{i : T_i = 1} Y_{it} \Bigr)^2 + \zeta^2 T_0 \|\omega\|_2^2,\\
    (\lambda_0, {\boldsymbol{\lambda}}) &= \underset{\lambda_0 \in \mathbb{R},\ \boldsymbol\lambda \in \Delta'}{\mathrm{argmin}} \sum_{j : C_j = 1} \Bigl( \lambda_0 + \sum_{t=1}^{T_0} \lambda_t Y_{jt} - Y_{j,T_0+1} \Bigr)^2
\end{align*}
where $\Delta,\Delta'$ denote the simplex over donors and $T_0$ periods respectively.

For inference, the baseline methods described in Section~\ref{related} no longer apply directly. Hence, we take the SDID jackknife and bootstrap \citep{sdid} as heuristic baselines. Our main method of interest is the Neyman Jackknife, which we apply to the first-stage randomization $\W = (W_1,\dots, W_n).$ Say that unit $i$'s request is granted if $W_i = e_1$ and ungranted otherwise. We choose the update set $S$ as a uniformly random set of $s_1$ granted and $s_0$ ungranted units, which gives 
$$\s = n\cdot \frac{1/n_1 + 1/n_0}{1/s_1 + 1/s_0}$$
by Proposition \ref{spectral_CRT_k}. For our experiments, we fix $s_1 = 1$ and consider $s_0 = 1$ or $2.$ Given $S$, we construct the proxy by drop-and-recompute. Let $\mathcal T^{(-S)} = \{ i \notin S : T_i = 1 \}$ be the units with granted requests outside $S$, and let $\mathcal C^{(-S)} = \{ j : C_j = 1,\ d(j, S) > 2 \}$ be the donors with no unit of $S$ within two hops. We define $f^{(-S)}$ as the SDID estimator computed with treated set $\mathcal T^{(-S)}$ and donor pool $\mathcal C^{(-S)}$, re-solving the weights $\omega$, $\boldsymbol\lambda$ and the regularization $\zeta$. In this new sample, there are $n_1 - 1$ units with granted requests. 

The experiment results are in Table~\ref{fig:panel_results}. Unlike Sections~\ref{temp_exp} and \ref{spatial_exp}, there are no existing methods with formal guarantees to compare against. The Neyman Jackknife fills this gap, and is provably finite-sample conservative in expectation across all spillover strengths. In contrast, the SDID jackknife and bootstrap under-cover for various spillover strengths. We note that the SDID jackknife and bootstrap were proposed for a model-based setting without spillovers. Hence, the above results do not indicate any shortcoming of these methods in the original setting they were proposed. 

\section{Discussion}
In this paper, we introduced the Neyman Jackknife, a framework for design-based, conservative variance estimation under interference. Given a design $\pi,$ choose a subset-sampling rule and analyze its effective subset size once. Then, assuming known exposure sets, our framework gives conservative variance estimators for any treatment-effect estimator $\hat\tau$ under $\pi$. The key idea is to drop some treatments and any outcomes they affect, then recompute $\hat\tau$ via a proxy only involving the retained information. Any proxy leads to a conservative variance estimator, whereas better proxies can lead to tighter estimates.

We discuss some directions for future work. First, the Neyman Jackknife may be useful for randomization inference. Studentizing a randomization test by a conservative variance estimate can extend its validity from sharp to weak null hypotheses \citep{wu2021randomization}. It would be interesting to study whether the Neyman Jackknife can serve this role under interference.

Second, it would be useful to study the effective subset size for various subset-sampling rules under adaptive designs or rerandomization. While an exact characterization of the effective subset size may be difficult for such designs, one could derive lower bounds that would still yield conservative variance estimators. For adaptive designs, we note that recomputation-based methods that exploit the filtration structure seem promising. For instance, the prediction-based variance estimation strategy of \cite{zeng2026sequentially} also uses the fact that the conditional expectation minimizes the mean squared error.

Finally, a third direction involves the requester design for network experiments. Here, a unit's treated exposure depends only on its own request, and all deletion bookkeeping for the Neyman Jackknife falls on the control side. Open questions include how to choose optimal proxies for the Neyman Jackknife, and how the requester design compares in efficiency with existing designs such as RGCR and the conflict graph design.

\section*{Statement on AI Use}
We used AI (Claude, ChatGPT) as an auxiliary tool in three ways. First, we formulated or proved some mathematical results in Sections 4 and 5 with the assistance of AI. For instance, AI helped formulate parts of Proposition 7, Lemma 12, and Proposition S3; AI also suggested the technique of martingale differences to simplify the proof of Lemma 12. Second, the simulation and plotting code for Section 6 was developed using AI assistance. Finally, AI was used for limited wording suggestions or sentence-level editing in parts of the manuscript. The authors reviewed all AI-assisted material, wrote the final statements and proofs, and are fully responsible for the content of the manuscript.

\bibliographystyle{plainnat} 
\bibliography{references} 

@article{efron1981jackknife,
  title={The jackknife estimate of variance},
  author={Efron, B. and Stein, C.},
  journal={Ann. Statist.},
  pages={586--596},
  year={1981},
  publisher={JSTOR}
}

@article{bojinov2023design,
  title={Design and analysis of switchback experiments},
  author={Bojinov, I. and Simchi-Levi, D. and Zhao, J.},
  journal={Manag. Sci.},
  volume={69},
  number={7},
  pages={3759--3777},
  year={2023},
  publisher={INFORMS}
}

@article{neyman1923applications,
  title={Sur les applications de la th{\'e}orie des probabilit{\'e}s aux exp\'{e}riences agricoles: Essai des principes},
  author={Neyman, J.},
  journal={Roczniki Nauk Rolniczych},
  volume={10},
  pages={1--51},
  year={1923}
}

@misc{wager2026causal,
  author = {Wager, S.},
  title = {Causal Inference: A Statistical Learning Approach},
  year = {2026},
  howpublished = {\url{https://web.stanford.edu/~swager/causal_inf_book.pdf}}
}

@misc{pollmann2020causal,
  title={Causal inference for spatial treatments},
  author={Pollmann, M.},
  archivePrefix = {arXiv},
  note          = {Preprint},
  year={2020}
}

@misc{lu2025design,
  title={Design-based causal inference in bipartite experiments},
  author={Lu, S. and Shi, L. and Fang, Y. and Zhang, W. and Ding, P.},
  archivePrefix = {arXiv},
  note          = {Preprint},
  year={2025}
}

@article{Aronow,
  author  = {Aronow, P. and Samii, C.},
  title   = {Estimating average causal effects under general interference, with application to a social network experiment},
  journal = {Ann. Appl. Stat.},
  volume  = {11},
  number  = {4},
  pages   = {1912--1947},
  year    = {2017}
}

@article{Hu,
  author  = {Hu, Y. and Li, S. and Wager, S.},
  title   = {Average direct and indirect causal effects under interference},
  journal = {Biometrika},
  volume  = {109},
  number  = {4},
  pages   = {1165--1172},
  year    = {2022}
}

@article{Hudgens,
  author  = {Hudgens, M. and Halloran, M.},
  title   = {Toward causal inference with interference},
  journal = {J. Amer. Statist. Assoc.},
  volume  = {103},
  number  = {482},
  pages   = {832--842},
  year    = {2008}
}

@article{Manski,
  author  = {Manski, C.},
  title   = {Identification of treatment response with social interactions},
  journal = {Econom. J.},
  volume  = {16},
  number  = {1},
  pages   = {S1--S23},
  year    = {2013}
}

@article{Savje,
  author  = {S{\"a}vje, F. and Aronow, P. and Hudgens, M.},
  title   = {Average treatment effects in the presence of unknown interference},
  journal = {Ann. Statist.},
  volume  = {49},
  number  = {2},
  pages   = {673--701},
  year    = {2021}
}

@article{gao2025network,
  author  = {Gao, M. and Ding, P.},
  title   = {Causal inference in network experiments: Regression-based analysis and design-based properties},
  journal = {J. Econometrics},
  volume  = {252},
  number  = {Part A},
  pages   = {106119},
  year    = {2025}
}

@article{leung2022ani,
  author  = {Leung, M.},
  title   = {Causal inference under approximate neighborhood interference},
  journal = {Econometrica},
  volume  = {90},
  number  = {1},
  pages   = {267--293},
  year    = {2022}
}

@article{filmus2016basis,
  author  = {Filmus, Y.},
  title   = {An orthogonal basis for functions over a slice of the boolean hypercube},
  journal = {Electron. J. Comb.},
  volume  = {23},
  number  = {1},
  year    = {2016}
}

@article{harshaw2026variance,
  author  = {Harshaw, C. and Middleton, J. and S{\"a}vje, F.},
  title   = {Optimized variance estimation under interference and complex experimental designs},
  journal = {J. Amer. Statist. Assoc.},
  year    = {2026}
}

@book{levin2017markov,
  title     = {Markov Chains and Mixing Times},
  author    = {Levin, D. and Peres, Y. and Wilmer, E.},
  edition   = {2nd},
  year      = {2017},
  publisher = {American Mathematical Society}
}

@article{NeweyWest1987,
  author  = {Newey, W. and West, K.},
  title   = {A Simple, Positive Semi-Definite, Heteroskedasticity and Autocorrelation Consistent Covariance Matrix},
  journal = {Econometrica},
  year    = {1987},
  volume  = {55},
  number  = {3},
  pages   = {703--708}}

@article{Kunsch1989,
  author  = {K{\"u}nsch, H.},
  title   = {The Jackknife and the Bootstrap for General Stationary Observations},
  journal = {Ann. Statist.},
  year    = {1989},
  volume  = {17},
  number  = {3},
  pages   = {1217--1241}
}

@article{ImbensMenzel2021,
  author  = {Imbens, G. and Menzel, K.},
  title   = {A Causal Bootstrap},
  journal = {Ann. Statist.},
  year    = {2021},
  volume  = {49},
  number  = {3},
  pages   = {1460--1488}
}

@article{AronowGreenLee2014,
  author  = {Aronow, P. and Green, D. and Lee, D.},
  title   = {Sharp Bounds on the Variance in Randomized Experiments},
  journal = {Ann. Statist.},
  year    = {2014},
  volume  = {42},
  number  = {3},
  pages   = {850--871}
}

@article{AbadieAtheyImbensWooldridge2020,
  author  = {Abadie, A. and Athey, S. and Imbens, G. and Wooldridge, J.},
  title   = {Sampling-Based versus Design-Based Uncertainty in Regression Analysis},
  journal = {Econometrica},
  year    = {2020},
  volume  = {88},
  number  = {1},
  pages   = {265--296}
}

@article{HuWager2022,
  author       = {Hu, Y. and Wager, S.},
  title        = {Switchback Experiments under Geometric Mixing},
  journal = {Manag. Sci.},
  volume = {forthcoming},
  year = {2026}
}

@article{robins1988confidence,
  author  = {Robins, J.},
  title   = {Confidence Intervals for Causal Parameters},
  journal = {Stat. Med.},
  volume  = {7},
  number  = {7},
  pages   = {773--785},
  year    = {1988}
}

@article{nutz2022directional,
  author  = {Nutz, M. and Wang, R.},
  title   = {The Directional Optimal Transport},
  journal = {Ann. Appl. Probab.},
  volume  = {32},
  number  = {2},
  pages   = {1400--1420},
  year    = {2022}
}

@article{MukerjeeDasguptaRubin2018,
  author  = {Mukerjee, R. and Dasgupta, T. and Rubin, D.},
  title   = {Using Standard Tools from Finite Population Sampling to Improve Causal Inference for Complex Experiments},
  journal = {J. Amer. Statist. Assoc.},
  year    = {2018},
  volume  = {113},
  number  = {522},
  pages   = {868--881}
}

@misc{ChattopadhyayImbens2024,
  author        = {Chattopadhyay, A. and Imbens, G.},
  title         = {Neymanian Inference in Randomized Experiments},
  year          = {2024},
  eprint        = {2409.12498},
  archivePrefix = {arXiv},
  primaryClass  = {stat.ME},
  note          = {Preprint}
}

@misc{lin2025timeseries,
  author        = {Lin, Z. and Ding, P.},
  title         = {Unifying regression-based and design-based causal inference in time-series experiments},
  year          = {2025},
  eprint        = {2510.22864},
  archivePrefix = {arXiv},
  primaryClass  = {stat.ME},
  note          = {Preprint}
}

@article{politis1997subsampling,
  title   = {Subsampling for heteroskedastic time series},
  author  = {Politis, D. and Romano, J. and Wolf, M.},
  journal = {J. Econometrics},
  volume  = {81},
  number  = {2},
  pages   = {281--317},
  year    = {1997}
}

@book{odonnell2014analysis,
  author    = {O'Donnell, R.},
  title     = {Analysis of Boolean Functions},
  publisher = {Cambridge University Press},
  year      = {2014}
}

@book{efron1982jackknife,
  author    = {Efron, B.},
  title     = {The Jackknife, the Bootstrap and Other Resampling Plans},
  publisher = {Society for Industrial and Applied Mathematics},
  year      = {1982}
}

@article{holland1986,
  author    = {Holland, P.},
  title     = {Statistics and Causal Inference},
  journal = {J. Amer. Statist. Assoc.},
  year    = {1986},
  volume  = {81},
  number  = {396},
  pages   = {945--960}
}

@article{rubin1980sutva,
  author  = {Rubin, D.},
  title   = {Randomization analysis of experimental data: {T}he {F}isher randomization test comment},
  journal = {J. Amer. Statist. Assoc.},
  volume  = {75},
  number  = {371},
  pages   = {591--593},
  year    = {1980}
}

@misc{conflict2026,
  author        = {Kandiros, V. and Pipis, C. and Daskalakis, C. and Harshaw, C.},
  title         = {The Conflict Graph Design: Estimating Causal Effects under Arbitrary Neighborhood Interference},
  year          = {2024},
  eprint        = {2411.10908v3},
  archivePrefix = {arXiv},
  primaryClass  = {stat.ME},
  note          = {Preprint}
}

@article{conley1999,
  author  = {Conley, T. G.},
  title   = {{GMM} estimation with cross sectional dependence},
  journal = {J. Econometrics},
  volume  = {92},
  number  = {1},
  pages   = {1--45},
  year    = {1999}
}

@article{kojevnikov2021limit,
  author  = {Kojevnikov, D. and Marmer, V. and Song, K.},
  title   = {Limit theorems for network dependent random variables},
  journal = {J. Econometrics},
  volume  = {222},
  number  = {2},
  pages   = {882--908},
  year    = {2021}
}

@misc{kojevnikov2021bootstrap,
  author        = {Kojevnikov, D.},
  title         = {The Bootstrap for Network Dependent Processes},
  year          = {2021},
  eprint        = {2101.12312v1},
  archivePrefix = {arXiv},
  primaryClass  = {econ.EM},
  note          = {Preprint}
}

@article{ugander2023,
  author        = {Ugander, J. and Yin, H.},
  title         = {Randomized graph cluster randomization},
  journal = {J. Causal Inference},
  volume  = {11},
  number  = {1},
  year    = {2023}
}

@article{sdid,
  author  = {Arkhangelsky, D. and Athey, S. and Hirshberg, D. A. and Imbens, G. W. and Wager, S.},
  title   = {Synthetic difference-in-differences},
  journal = {Amer. Econ. Rev.},
  volume  = {111},
  number  = {12},
  pages   = {4088--4118},
  year    = {2021}
}

@book{van,
  author    = {van der Vaart, A. W.},
  title     = {Asymptotic Statistics},
  publisher = {Cambridge University Press},
  year      = {1998}
}

@article{wu2021randomization,
  author  = {Wu, J. and Ding, P.},
  title   = {Randomization tests for weak null hypotheses in randomized experiments},
  journal = {J. Amer. Statist. Assoc.},
  volume  = {116},
  number  = {536},
  pages   = {1898--1913},
  year    = {2021}
}

@article{miller1974jackknife,
  title     = {The Jackknife---A Review},
  author    = {Miller, R.},
  journal   = {Biometrika},
  volume    = {61},
  number    = {1},
  pages     = {1--15},
  year      = {1974}
}

@inproceedings{ugander2013graph,
  title     = {Graph Cluster Randomization: Network Exposure to Multiple Universes},
  author    = {Ugander, J. and Karrer, B. and Backstrom, L. and Kleinberg, J.},
  booktitle = {Proceedings of the 19th ACM SIGKDD International Conference on Knowledge Discovery and Data Mining},
  pages     = {329--337},
  year      = {2013}
}

@article{tukey1958bias,
  title     = {Bias and Confidence in Not-Quite Large Samples (Abstract)},
  author    = {Tukey, J.},
  journal   = {Ann. Math. Stat.},
  volume    = {29},
  number    = {2},
  pages     = {614},
  year      = {1958}
}

@article{leung2023network,
  author  = {Leung, M.},
  title   = {Network cluster-robust inference},
  journal = {Econometrica},
  volume  = {91},
  number  = {2},
  pages   = {641--667},
  year    = {2023}
}

@article{mackinnon2023cluster,
  author  = {MacKinnon, J. and Nielsen, M. and Webb, M.},
  title   = {Cluster-robust inference: A guide to empirical practice},
  journal = {J. Econometrics},
  volume  = {232},
  number  = {2},
  pages   = {272--299},
  year    = {2023}
}

@article{abadie2023clustering,
  author  = {Abadie, A. and Athey, S. and Imbens, G. and Wooldridge, J.},
  title   = {When should you adjust standard errors for clustering?},
  journal = {Q. J. Econ.},
  volume  = {138},
  number  = {1},
  pages   = {1--35},
  year    = {2023}
}

@article{ross2011stein,
  author  = {Ross, N.},
  title   = {Fundamentals of {S}tein's method},
  journal = {Probab. Surv.},
  volume  = {8},
  pages   = {210--293},
  year    = {2011}
}

@article{hall1995blocking,
  author  = {Hall, P. and Horowitz, J. and Jing, B.},
  title   = {On blocking rules for the bootstrap with dependent data},
  journal = {Biometrika},
  volume  = {82},
  number  = {3},
  pages   = {561--574},
  year    = {1995}
}

@article{brandt1938tests,
  author  = {Brandt, A.},
  title   = {Tests of significance in reversal or switchback trials},
  journal = {Res. Bull. Iowa Agric. Home Econ. Exp. Stn.},
  volume  = {21},
  number  = {234},
  pages   = {1},
  year    = {1938}
}

@article{cochran1941double,
  author  = {Cochran, W. and Autrey, K. and Cannon, C.},
  title   = {A double change-over design for dairy cattle feeding experiments},
  journal = {J. Dairy Sci.},
  volume  = {24},
  number  = {11},
  pages   = {937--951},
  year    = {1941}
}

@article{abadie2010synthetic,
  author  = {Abadie, A. and Diamond, A. and Hainmueller, J.},
  title   = {Synthetic control methods for comparative case studies: Estimating the effect of {California's} tobacco control program},
  journal = {J. Amer. Statist. Assoc.},
  volume  = {105},
  number  = {490},
  pages   = {493--505},
  year    = {2010}
}

@article{bottmer2024design,
  author  = {Bottmer, L. and Imbens, G. and Spiess, J. and Warnick, M.},
  title   = {A design-based perspective on synthetic control methods},
  journal = {J. Bus. Econ. Statist.},
  volume  = {42},
  number  = {2},
  pages   = {762--773},
  year    = {2024}
}

@misc{zeng2026sequentially,
  author        = {Zeng, Z. and Adjaho, C. and Bucarey, A. and Qin, C. and Zhang, R. and Hoban, P. and Johari, R. and Wager, S.},
  title         = {Sequentially-Rerandomized Switchback Experiments},
  year          = {2026},
  eprint        = {2604.02489},
  archivePrefix = {arXiv},
  primaryClass  = {stat.ME},
  note          = {Preprint}
}

@article{bloznelis2001orthogonal,
  author  = {Bloznelis, M. and G{\"o}tze, F.},
  title   = {Orthogonal decomposition of finite population statistics and its applications to distributional asymptotics},
  journal = {Ann. Statist.},
  volume  = {29},
  number  = {3},
  pages   = {899--917},
  year    = {2001}
}

\newpage

\appendix\phantomsection\label{app}
\renewcommand{\thesection}{S\arabic{section}}   
\renewcommand{\theequation}{S\arabic{equation}} 
\renewcommand{\thefigure}{S\arabic{figure}}
\renewcommand{\thetable}{S\arabic{table}}
\renewcommand{\thetheorem}{S\arabic{theorem}}   
\renewcommand{\theHtheorem}{S\arabic{theorem}}
\setcounter{section}{0}
\setcounter{equation}{0}
\setcounter{figure}{0}
\setcounter{table}{0}
\setcounter{theorem}{0}

\section{Overview of Prior Variance Estimation Methods}\label{overview_prior_sec_appendix}

We gather various methods for design-based variance estimation in the literature and illustrate their key ideas. Assume neighborhood interference, i.e., place $n$ units on a simple, undirected graph $G_0$. We assume that the maximum degree of $G_0$ is $O(1).$ We also assume the Bernoulli design for expository purposes. Each $i\in [n]$ receives treatment $W_i\sim\text{Ber}(\pi_i)$ and we observe outcome $Y_i.$  We consider the IPW estimator $\hat\tau = f(\W)$ given by 
\begin{align*}
    \hat\tau &= \frac{1}{n}\sum_{i=1}^n\psi_i = \frac{1}{n}\sum_{i=1}^n \left(\frac{T_i}{p_i}-\frac{C_i}{q_i}\right)Y_i
\end{align*}
where $T_i = \mathbbm{1}\{\W_{\N_i}=\mathbf{1}_{\N_i}\}$, $C_i = \mathbbm{1}\{\W_{\N_i}=\mathbf{0}_{\N_i}\},$ and $p_i,q_i$ are the respective propensities. Indeed, $\N_i=B_1(i)$ is the closed 1-hop neighborhood of $i.$

Each method we review aims to construct $\widehat V$ such that $\E_\pi[\widehat V]\geq \Var_\pi(\hat\tau)$, or at least asymptotically. Again, while the joint choice of design and estimator above is far from optimal, we use it as a clean setting to describe various methods for variance estimation. Moreover, we note that many of the existing methods were proposed as part of a larger program for inference, proposing new estimators or designs as well. Thus, while we review each method purely as a variance-estimation strategy, their original development may have involved more complex designs or estimators.

\subsection{Estimable Upper Bounds}
This section reviews the first line of work described in Section \ref{related}. The main idea behind these methods is to derive an exact (but non-estimable) expression of the true variance, then construct an estimable upper bound. We begin with \citet{Aronow}. Note that the exact variance of $\hat\tau$ is given by 
\begin{align*}
    \Var(\hat\tau) = \frac{1}{n^2}\sum_{i,j=1}^n \Cov(\psi_i,\psi_j).
\end{align*}
The problem is that $\Cov(\psi_i,\psi_j)$, when expanded, includes terms involving $Y_i(w)Y_j(w')$ where $w,w'\in\{\mathbf{0},\mathbf{1}\}$. Such a product is non-estimable if $Y_i(w)$, $Y_j(w')$ cannot be observed together. The proposed solution is to track each non-estimable product in the covariance expansion and apply
\begin{align*}
    |Y_i(w)Y_j(w')|\leq \frac{Y_i(w)^2 + Y_j(w')^2}{2}.
\end{align*}
This produces an estimable upper bound $V_{\normalfont\text{UB}}\geq \Var(\hat\tau)$. One can then construct an unbiased estimator for $V_{\normalfont\text{UB}}$ via inverse-propensity weighting.

Next, we review the methods of \citet{conflict2026} and \citet{harshaw2026variance}. Here, the idea is to write 
\begin{align*}
    \hat\tau = \frac{1}{n}\sum_{i=1}^n \left(\frac{\mathbbm{1}\{\W_{\N_i}=\mathbf{1}_{\N_i}\}}{p_i}Y_i(\mathbf{1}) - \frac{\mathbbm{1}\{\W_{\N_i}=\mathbf{0}_{\N_i}\}}{q_i}Y_i(\mathbf{0})\right) = \frac{1}{n}\mathbf{v}'\boldsymbol{\theta}
\end{align*}
where $\boldsymbol\theta = (Y_i(\mathbf{1}),Y_i(\mathbf{0}))_{i=1}^n$ is deterministic but only partially observed while $\mathbf{v}\in \mathbb{R}^{2n}$ is a random vector involving the indicators and propensities. Then, we can write 
\begin{align*}
    \Var(\hat\tau) = \frac{1}{n^2}\boldsymbol\theta' A \boldsymbol\theta
\end{align*}
where $A=\Cov(\mathbf{v})$ is fully known.

We now construct estimable upper bounds for the variance, viewed as a quadratic form. First, \citet{conflict2026} use $\lambda_1(A),$ the largest eigenvalue of $A$, to get
\begin{align*}
    \frac{1}{n^2}\boldsymbol{\theta}' A \boldsymbol\theta&\leq \frac{\lambda_1(A)}{n^2}\cdot \norm{\boldsymbol{\theta}}_2^2\\
    &=\frac{\lambda_1(A)}{n^2} \cdot \sum_{i=1}^n \left(Y_i(\mathbf{1})^2 + Y_i(\mathbf{0})^2\right).
\end{align*}
The upper bound is estimable via inverse-propensity weighting. Second, \citet{harshaw2026variance} choose a $2n\times 2n$ matrix $B$ so that $B\succeq A$ (in Loewner ordering), obtaining the upper bound $n^{-2}\boldsymbol\theta' B\boldsymbol\theta.$ In order to ensure the upper bound is estimable via inverse-propensity weighting, they require $B_{k\ell}=0$ whenever $\theta_k,\theta_\ell$ are not jointly observable. Moreover, they choose an optimal $B$ by minimizing an objective function involving $B.$ 

Finally, we review the method of \citet{lu2025design}. Writing 
\begin{align*}
    \hat\tau &= \hat\mu_1 - \hat\mu_0
    =\frac{1}{n}\sum_{i=1}^n\frac{T_i}{p_i}Y_i - \frac{1}{n}\sum_{i=1}^n\frac{C_i}{q_i}Y_i,
\end{align*}
their idea is that $\Var(\hat\mu_1)$ and $\Var(\hat\mu_0)$ can each be estimable. By Cauchy--Schwarz, they obtain
\begin{align*}
    \Var(\hat\mu_1-\hat\mu_0)\leq \left(\sqrt{\Var(\hat\mu_1)} + \sqrt{\Var(\hat\mu_0)}\right)^2.
\end{align*}
Moreover, they provide a consistent estimator for the upper bound, and thus an asymptotically conservative variance estimator.

\subsection{HAC Methods}
Here, we review the second line of work described in Section \ref{related}. The first idea here is to write
\begin{align*}
    \Var(\hat\tau) = \frac{1}{n^2}\sum_{i,j=1}^n \mathcal{G}_{ij} \cdot \Cov(\psi_i,\psi_j)
\end{align*}
where $\G_{ij}=0$ if $\psi_i,\psi_j$ are known to be uncorrelated and $\G_{ij}=1$ otherwise. Since $G_0$ is known, we have a good handle on the dependency graph $\G.$ Under the Bernoulli design, $\G_{ij}=0$ whenever $d_{G_0}(i,j)>2$. The next idea is the approximation 
\begin{align*}
    \Cov(\psi_i,\psi_j)\approx (\psi_i-\overline\psi_n)(\psi_j-\overline\psi_n)
\end{align*}
where we plug in $\overline{\psi}_n=\hat\tau$ instead of $\E[\psi_i]$ which is unknown. This gives the dependency-graph based HAC estimator \citep{leung2022ani}
\begin{align*}
    \widehat{V}_{\text{HAC}} = \frac{1}{n^2}(\boldsymbol{\psi} - \overline{\psi}_n\mathbf{1})'\G( \boldsymbol{\psi} - \overline{\psi}_n\mathbf{1}).
\end{align*}
Unlike methods from the first line of work, $\widehat{V}_{\text{HAC}}$ is not guaranteed to be conservative, even asymptotically. This is due to the treatment-effect heterogeneity term
\begin{align}
    \frac{1}{n^2}(\boldsymbol\tau - \tau \mathbf{1})'\G (\boldsymbol\tau - \tau \mathbf{1})\label{eq:HAC_het}
\end{align}
in $\E[\widehat{V}_{\text{HAC}}] - \Var(\hat\tau)$ where $\boldsymbol{\tau}=(\tau_i)_{i=1}^n$ is the treatment-effect vector and $\tau=n^{-1}\sum_{i=1}^n \tau_i$. Here, \eqref{eq:HAC_het} may be negative since $\G$ is not necessarily positive semidefinite.

\citet{gao2025network} address this problem by replacing $\G$ with $\G^+,$ where $\G^+$ is obtained by zeroing out all negative eigenvalues of $\G$ so that $\G^+$ is PSD. The modified HAC estimator is given by 
\begin{align*}
    \widehat{V}_{\text{HAC}^+} = \frac{1}{n^2}(\boldsymbol\psi - \overline{\psi}_n\mathbf{1})'\G^+ (\boldsymbol\psi - \overline{\psi}_n\mathbf{1}),
\end{align*}
which can be shown to be asymptotically conservative. Finally, in time-series experiments, \cite{lin2025timeseries} follow the classical econometric route to obtain asymptotic conservativeness. Instead of choosing the kernel as the dependency graph $\G$, they use the Bartlett kernel to construct their HAC estimator.

\section{Proofs of Theoretical Results}\label{proof_sec_appendix}

\subsection{Proofs for Section~\ref{recipe_sec}}\label{spectral_proof}

For the results in this section, it suffices to derive the spectral gap $\lambda$ for the various subset-sampling rules and designs considered. One can then use the relation $\s = m\lambda.$ We begin with some definitions. Take $\W\in\{0,1\}^m$ with $\W\sim\pi$ and let $\mathscr{H}=L^2(\sigma(\W))$ denote the space of all square-integrable random variables measurable with respect to $\W.$ Let $\Omega\subseteq\{0,1\}^m$ denote the support of $\W.$ Given a subset-sampling rule $\mu$, let $\W'$ be obtained by Gibbs rerandomization and let $P$ denote the induced transition matrix. Let $\mathscr{T}:\mathscr{H}\to\mathscr{H}$ denote the linear operator
    \begin{align*}
        \mathscr{T}(f(\W)) = \E[f(\W')\mid \W].
    \end{align*}
Finally, for any $w\in\Omega,$ note that
\begin{align*}
     \E[f(\W')\mid \W=w]= \sum_{w'\in\Omega}P(w,w')f(w')= (Pf)_w
\end{align*}
where the last expression views $f=(f(w))_{w\in\Omega}$ as a $|\Omega|$-dimensional vector and $Pf$ is matrix-vector multiplication. 
    
Since $\mathscr{T}(f(\W)) = (Pf)_{\W},$ there is a one-to-one correspondence between eigenvalues of $\mathscr{T}$ and eigenvalues of $P.$ When $\pi$ is the Bernoulli design or the completely randomized design, we have explicit expressions for the orthogonal basis of $\mathscr{H}.$ For $S$ in the settings of Lemmas \ref{spectral_ber}, \ref{spectral_CRT}, and \ref{spectral_stratified_k}, each basis element is an eigenfunction for $\mathscr{T}$, allowing explicit calculations of eigenvalues. 

One consequence is that eigenvalues of $P$ are nonnegative. To see this, say $\E[f(\W')\mid\W] = c \cdot f(\W)$ for some $c.$ Then, we have 
\begin{align*}
    \E[f(\W)f(\W')] = \E\Big[\E[f(\W')\mid \W]\cdot f(\W)\Big ] = c\cdot \E[f(\W)^2].
\end{align*}
On the other hand, we also have 
\begin{align*}
    \E[f(\W)f(\W')] &= \E\Big[\E\left[f(\W)f(\W')\mid S,\W_{-S}\right]\Big]\\
    &= \E\Big[\E\left[f(\W)\mid S,\W_{-S}\right]\cdot \E\left[f(\W')\mid S,\W_{-S}\right]\Big]\\
    &= \E\Big[\Big(\E\left[f(\W)\mid S,\W_{-S}\right]\Big)^2\Big]\geq 0
\end{align*}
since $\W_S,\W'_S$ are i.i.d. given $(S,\W_{-S}).$ This gives $c\geq 0$ as desired. We begin with the Bernoulli design.

\begin{proof}[Proof of Lemma~\ref{spectral_ber}]
    The orthogonal basis of $\mathscr{H}$ is given by $\{\phi_A\}_{A\subseteq [m]}$ where 
    \begin{align*}
        \phi_A = \prod_{j\in A}\frac{W_j-\pi_j}{\sqrt{\pi_j(1-\pi_j)}}
    \end{align*}
    (see \cite{odonnell2014analysis}). We check that each $\phi_A$ is an eigenfunction of $\mathscr{T}.$ First, note that
    \begin{align*}
        \mathscr{T}(\phi_A) &= \E[\phi_A(\W')\mid \W]= \E[\E[\phi_A(\W')\mid S,\W]\mid \W].
    \end{align*}
    Given $(S,\W),$ the entries of $\W'_S$ are independent, hence
    \begin{align*}
        \E[\phi_A(\W')\mid S,\W] &= \phi_{A\setminus S}(\W)\cdot \E[\phi_{A\cap S}(\W')\mid S,\W]
        =
        \begin{cases}
            \phi_A(\W) & A\cap S = \emptyset,\\
            0 & \text{otherwise}.
        \end{cases}
    \end{align*}
    This gives
    \begin{align*}
        \mathscr{T}(\phi_A) = \E[\phi_A(\W)\mathbbm{1}_{A\cap S=\emptyset}\mid \W] = \Pr(A\cap S=\emptyset)\cdot \phi_A(\W),
    \end{align*}
    showing that $\phi_A$ is an eigenfunction with corresponding eigenvalue $\Pr(A\cap S = \emptyset).$ In particular, these are the eigenvalues of $P$ as well. The largest eigenvalue is when $A = \emptyset$ so that $\Pr(A\cap S = \emptyset) = 1.$ The second largest eigenvalue is the maximum over $A\neq\emptyset,$ and the spectral gap is 
    \begin{align*}
        1 - \max_{A\neq\emptyset}\Pr(A\cap S = \emptyset) = \min_{A\neq\emptyset} \Pr(A\cap S \neq\emptyset)
    \end{align*}
    as desired. The simplification to the minimum over $j\in [m]$ follows by monotonicity, and we conclude the proof.
\end{proof}

Next, we consider the completely randomized design. The following result is for subset-sampling rules more general than those described in the main text.

\begin{lemma}\label{spectral_CRT}
    Let $\W$ follow the completely randomized design with $n_1$ treated units and $n_0$ control units so that $n_1+n_0=m$. Let $\mathcal{T},\mathcal{C}$ respectively denote the set of treated and control units. Assume the subset-sampling rule $\mu$ is permutation-invariant, meaning
    $$\Pr(S=A\mid \W=w) = \Pr(S=\sigma A\mid \W=\sigma w)$$
    for every permutation $\sigma$ of $[m].$ Then we have
    \begin{align*}
        \lambda = \frac{n_1+n_0}{n_1n_0}\cdot \E\left[\frac{|S\cap \T|\cdot |S\cap \C|}{|S|}\right]
    \end{align*}
    where $|S\cap \T|\cdot |S\cap \C|/|S| = 0$ if $S=\emptyset.$
\end{lemma}

\begin{proof}
    Let $\Omega = \{w\in\{0,1\}^m: \norm{w}_1 = n_1\}$ denote the support of $\W$. Recall our assumption that $\Pr(S=A\mid \W=w) = \Pr(S=\sigma A\mid \W = \sigma w)$ for any permutation $\sigma.$ We first check that this gives $\Pr(\W'=w'\mid \W=w) = \Pr(\W'=\sigma w'\mid \W=\sigma w)$ as well. For any $\sigma\in S_m,$ we see that 
    \begin{align*}
        \Pr(\W'=\sigma w'\mid \W=\sigma w) &= \sum_{A\subseteq [m]} \Pr(S=A\mid \W = \sigma w)\cdot \mathbbm{1}_{(\sigma w)_{-A}=(\sigma w')_{-A}}\cdot \frac{1}{\binom{|A|}{\sum_{j\in A}(\sigma w)_j}}\\
        &= \sum_{B\subseteq [m]} \Pr(S=\sigma B\mid \W = \sigma w)\cdot \mathbbm{1}_{(\sigma w)_{-\sigma B}=(\sigma w')_{-\sigma B}}\cdot \frac{1}{\binom{|\sigma B|}{\sum_{j\in \sigma B}(\sigma w)_j}}\\
        &= \sum_{B\subseteq [m]} \Pr(S=B\mid \W = w)\cdot \mathbbm{1}_{w_{- B}=w'_{-B}}\cdot \frac{1}{\binom{| B|}{\sum_{j\in  B} w_j}}\\
        &= \Pr(\W'=w'\mid \W=w).
    \end{align*}

    Next, this implies that $\Pr(\W'=w'\mid \W =w)$ only depends on $w\cdot w' = \sum_{j=1}^m w_jw'_j$, the number of common treated units. To see this, pick any $v,v',w,w'\in\Omega$ such that $v\cdot v' = w\cdot w' = t.$ Considering the pairs $\{(v_j,v'_j)\}_{j=1}^m$, the counts of $(1,1), (1,0), (0,1), (0,0)$ are $t, n_1-t, n_1-t, m - 2n_1 + t.$ The counts must be the same for $\{(w_j,w'_j)\}_{j=1}^m$. Thus, there is a permutation $\sigma$ sending $v$ to $w$ and $v'$ to $w'$, hence $\Pr(\W'=w'\mid \W=w) = \Pr(\W'=v'\mid \W = v)$.

    We now consider eigenfunctions of $\mathscr{T}$. For each integer $0\leq d\leq \min\{n_1,n_0\},$ define
    \begin{align*}
        \chi_d = \prod_{i=1}^d (W_{2i-1}-W_{2i}).
    \end{align*}
    As $\Pr(\W'=w'\mid \W=w)$ only depends on $w\cdot w',$ $P$ satisfies Definition 17 of \citet{filmus2016basis}. By Lemma 18 of \citet{filmus2016basis}, each $\chi_d$ is an eigenfunction of $\mathscr{T}$, and their corresponding eigenvalues generate the entire set of eigenvalues of $P$ (without multiplicity). Fix $d$ and let $w_{\star}=(1,0,1,0,...,1,0,\mathbf{1}_{n_1-d},\mathbf{0}_{n_0-d})\in \Omega$ where there are $d$ pairs of $(1,0)$ in the beginning. Since
    \begin{align*}
        \E\left[\chi_d(\W')\mid \W\right] = \theta_d\cdot \chi_d(\W),
    \end{align*}
    we see that
    \begin{align*}
    \E\left[\chi_d(\W')\mid \W = w_\star\right] = \theta_d\cdot \chi_d(w_\star) = \theta_d.
    \end{align*}

    To calculate the conditional expectations, we check two basic facts. First, conditional on $$(\W=w, |S\cap \T| = s_1, |S\cap \C| = s_0),$$ the distribution of $S$ is given by uniformly selecting $s_1$ treated and $s_0$ control units each. To see this, given $\W = w$, let $A,B$ be subsets of $[m]$ such that $|A\cap \T| = |B\cap \T| = s_1$ and $|A\cap \C| = |B\cap \C| = s_0$. Then, we can swap the intersections and obtain a permutation $\sigma$ such that $\sigma w = w$ and $\sigma A = B.$ Thus, 
    \begin{align*}
        \Pr(S=B\mid \W =w) = \Pr(S=\sigma A\mid \W = \sigma w) = \Pr(S = A\mid \W = w).
    \end{align*}
    
    Next, we check that the joint distribution of $(|S\cap \T|, |S\cap \C|)$ does not depend on $\W.$ Fixing $s_1,s_0,$ write $\mathcal{A}(w)=\{A\subseteq [m]: |A\cap \T(w)| = s_1, |A\cap \C(w)| = s_0\}.$ For any $w,w'\in \Omega,$ choose a permutation $\sigma$ such that $w' = \sigma w.$ Then, we have $\mathcal{A}(w') = \sigma \mathcal{A}(w).$ Thus, we have 
    \begin{align*}
       \Pr(|S\cap \T(w)| = s_1, |S\cap \C(w)| = s_0\mid \W = w) &= \Pr(S\in \mathcal{A}(w)\mid \W =w)\\
       &= \Pr(S\in \sigma\mathcal{A}(w)\mid \W =\sigma w)\\
       &= \Pr(S\in \mathcal{A}(w')\mid \W = w')
    \end{align*}
    as desired.

    We now evaluate the conditional expectations, beginning with $\theta_1.$ By considering the events $$\{1,2\notin S\}, \{1\in S,2\notin S\}, \{1\notin S,2\in S\},$$ we have
    \begin{align*}
       &\E[W_1' - W'_2\mid \W =w_\star, |S\cap \T| = s_1, |S\cap \C| = s_0] \\
       &= \frac{n_1-s_1}{n_1}\frac{n_0-s_0}{n_0} + \frac{s_1}{n_1}\cdot \frac{n_0-s_0}{n_0}\cdot \frac{s_1}{s_1 + s_0} + \frac{n_1-s_1}{n_1}\cdot \frac{s_0}{n_0}\cdot \frac{s_0}{s_1 + s_0}\\
       &= 1 - \frac{n_1+n_0}{n_1n_0} \frac{s_1s_0}{s_1+s_0}.
    \end{align*}
    Taking expectations, we get 
    \begin{align*}
        1 - \theta_1 &= \frac{n_1+n_0}{n_1n_0}\cdot \E\left[\frac{|S\cap \T|\cdot |S\cap \C|}{|S|}\mid \W=w_\star\right]\\
        &= \frac{n_1+n_0}{n_1n_0}\cdot \E\left[\frac{|S\cap \T|\cdot |S\cap \C|}{|S|}\right].
    \end{align*}
    
    To conclude, we show that $\theta_d\leq \theta_1$ for $1\leq d\leq \min\{n_1,n_0\}.$ Let $E$ denote the event $\{1,2\in S\}.$ Then, we have
    \begin{align*}
        \theta_d &= \E[\chi_d(\W')\mid \W=w_{\star}]\\
        &= \Pr(E'\mid \W=w_\star)\cdot \E[\chi_d(\W')\mid \W=w_\star, E'] + \Pr(E\mid \W=w_\star)\cdot \E[\chi_d(\W')\mid \W=w_\star, E]\\
        &\leq \Pr(E'\mid \W=w_\star)\cdot \E[\chi_1(\W')\mid \W=w_\star, E'] + \Pr(E\mid \W=w_\star)\cdot \E[\chi_d(\W')\mid \W=w_\star, E].
    \end{align*}
   Next, we check that $\E[\chi_d(\W')\mid \W=w_\star, E]=0$ for any $d\geq 1.$ Let $S=A$ for any $A\supseteq \{1,2\}.$ Let $\tau$ denote the permutation that swaps coordinates $1,2.$ Then, we see that
    \begin{align*}
        \Pr(\tau\W'= w' \mid \W = w, S = A) &= \Pr(\W'= \tau w' \mid \W = w, S = A)\\
        &=\mathbbm{1}_{(\tau w')_{-A}=w_{-A}}\cdot \frac{1}{\binom{|A|}{\sum_{j\in A}w_j}}\\
        &= \mathbbm{1}_{w'_{-A}=w_{-A}}\cdot \frac{1}{\binom{|A|}{\sum_{j\in A}w_j}}\\
        &= \Pr(\W'=w'\mid \W=w,S=A).
    \end{align*}
    Thus, the distributions of $\W',\tau\W'$ are the same given $\W=w_\star, S=A.$ This gives 
    \begin{align*}
        \E[\chi_d(\W')\mid \W=w_\star, S = A] = \E[\chi_d(\tau \W')\mid \W=w_\star, S = A] = -\E[\chi_d(\W')\mid \W=w_\star, S =A]
    \end{align*}
    and thus $\E[\chi_d(\W')\mid \W=w_\star, E] = 0.$ Bringing everything together, we have
    \begin{align*}
        \theta_d \leq \Pr(E'\mid \W=w_\star)\cdot \E[\chi_1(\W')\mid \W=w_\star, E'] = \E[\chi_1(\W')\mid \W=w_\star] = \theta_1.
    \end{align*}
    This concludes the proof.
\end{proof}

\begin{proof}[Proof of Proposition~\ref{spectral_CRT_k}]
    It suffices to show that the subset-sampling rule satisfies 
    \begin{align*}
        \Pr(S=A\mid \W=w) = \Pr(S=\sigma A\mid \W = \sigma w).
    \end{align*}
    Let $\mathcal{A}(w) = \{A\subseteq [m]: |A\cap \T(w)|=s_1, |A\cap \C(w)|=s_0\}.$ Then, we have
    \begin{align*}
        \Pr(S=\sigma A\mid \W = \sigma w) &= \frac{1}{\binom{n_1}{s_1}\binom{n_0}{s_0}} \mathbbm{1}_{\sigma A\in \mathcal{A}(\sigma w)}\\
        &= \frac{1}{\binom{n_1}{s_1}\binom{n_0}{s_0}} \mathbbm{1}_{\sigma A\in \sigma \mathcal{A}(w)}\\
        &= \frac{1}{\binom{n_1}{s_1}\binom{n_0}{s_0}} \mathbbm{1}_{ A\in  \mathcal{A}(w)}\\
        &= \Pr(S=A\mid \W=w).
    \end{align*}
    This concludes the proof.
\end{proof}

Finally, we consider the stratified completely randomized design. Again, the following result considers subset-sampling rules more general than those described in the main text.
\begin{lemma}
    \label{spectral_stratified_k}
    Let $I_1,\dots, I_H$ partition $[m]$ and let $\W$ follow the stratified completely randomized design on $\{0,1\}^m$. Independently across strata, $\W_{I_h}$ follows the completely randomized design with $n_{1h}$ treated and $n_{0h}$ control units so that $n_{1h}+n_{0h} = |I_h|=n_h$. Let $\T_h, \C_h$ denote the treated and control units in $I_h.$ Fix integers $1\leq s_{1h}\leq n_{1h}$ and $1\leq s_{0h}\leq n_{0h}$. Next, let $\mathbf{R}\in\{0,1\}^H$ be a random vector independent of $\W$ and write $q_h = \Pr(R_h=1).$ Let $\mu$ be the subset-sampling rule that, for each stratum $I_h$ with $R_h=1$, selects a uniformly random set of $s_{1h}$ treated units and an independent uniformly random set of $s_{0h}$ control units. The subset selections across strata are independent given $(\W,\R).$ Then,
    \begin{align*}
        \lambda = \min_{h\in [H]}q_h\cdot \frac{1/n_{1h}+1/n_{0h}}{1/s_{1h}+1/s_{0h}}.
    \end{align*}
\end{lemma}
\begin{proof}
    Let $\Omega_h\subseteq\{0,1\}^{|I_h|}$ denote the boolean slice of order $n_{1h}$. Then, the support of $\pi$ is $\Omega = \Omega_1\times\cdots\times\Omega_H$ and $\pi$ is a product distribution on $\Omega.$ For each $h\in [H],$ let $P_h$ denote the $|\Omega_h|\times |\Omega_h|$ transition matrix that updates $\W_{I_h}$ via choosing $s_{1h}$ treated and $s_{0h}$ control units, uniformly each. Since $\R\indep\W,$ given $\R$, the transition matrix from $\W$ to $\W'$ is given by 
    \begin{align*}
        P^{(\R)}=\bigotimes_{h=1}^H P_h^{(R_h)}
    \end{align*}
    where $\bigotimes$ is the Kronecker product and $P_h^{(1)}=P_h$, $P_h^{(0)} = I_{|\Omega_h|}$ (the identity matrix). 
    Unconditionally, we have
    \begin{align*}
        P = \sum_{r\in\{0,1\}^H}\Pr(\R=r)P^{(r)}.
    \end{align*}

    The matrices $P^{(r)}$ share an eigenbasis obtained by taking the Kronecker products of eigenvectors of the $P_h$. Write the eigenvalues of $P_h$ as $1=\theta_{h,0}\geq \theta_{h,1}\geq\cdots \geq 0$ (we know that eigenvalues of $P_h$ are nonnegative from the beginning of Section~\ref{spectral_proof}). Then, the eigenvalues of $P$ are given by 
    \begin{align}
        \sum_{r\in\{0,1\}^H}\Pr(\R = r)\prod_{h:r_h=1} \theta_{h,a_h} &= \E\left[\sum_{r\in\{0,1\}^H} \mathbbm{1}_{\R=r}\prod_{h:r_h=1}\theta_{h,a_h}\right]\nonumber\\
        &= \E\left[\prod_{h:R_h=1}\theta_{h,a_h}\right]\nonumber\\
        &= \E\left[\prod_{h=1}^H \left(1-R_h + R_h\theta_{h,a_h}\right)\right]\label{eq:prod_eig}
    \end{align}
    for all valid $(a_1,\dots, a_H)$. Since $(1-R_h + R_h\theta_{h,a_h})\geq 0$, the expression in \eqref{eq:prod_eig} is coordinatewise nondecreasing in $\theta$. Thus, the largest eigenvalue of $P$ is one. The second largest eigenvalue is
    \begin{align*}
        \max_{h\in [H]} \Big(\Pr(R_h=1)\theta_{h,1} + \Pr(R_h = 0)\Big)
    \end{align*}
    by choosing $(a_1,\dots, a_H)$ to be all zeros except for one coordinate, which we choose to be one. Hence, the spectral gap is 
    \begin{align*}
        \min_{h\in [H]}\Pr(R_h=1) (1-\theta_{h,1}).      
    \end{align*}
    We conclude by the expression of $1-\theta_{h,1}$ given in Proposition~\ref{spectral_CRT_k}.
\end{proof}
\begin{proof}[Proof of Proposition~\ref{strat_spec_main}]
 The result follows from applying Lemma~\ref{spectral_stratified_k} with $q_h = c_h^{-1}/c.$  
\end{proof}

\subsection{Proofs for Section \ref{Neyman_connect}}

\begin{proof}[Proof of Proposition~\ref{Neyman_connection}]
    First, we have 
    \begin{align*}
        \hat\tau - f^{(-A\sqcup B)} &= \left(\frac{1}{n_1}\sum_{i\in \T}Y_i - \frac{1}{n_0}\sum_{j\in \C}Y_j\right) - \left(\frac{1}{n_1 -s _1}\sum_{i\in \T\setminus A}Y_i - \frac{1}{n_0-s_0}\sum_{j\in \C\setminus B}Y_j\right)\\
        &= \frac{\sum_{i\in A} (Y_i - \bar{Y}_{\T})}{n_1-s_1} - \frac{\sum_{j\in B}(Y_j - \bar{Y}_{\C})}{n_0-s_0}
    \end{align*}
    where $\bar{Y}_\T = n_1^{-1}\sum_{i\in \T}Y_i$ and $\bar{Y}_\C = n_0^{-1}\sum_{j\in \C}Y_j.$ Since
    \begin{align*}
        \sum_{A\subseteq \T: |A|=s_1}\sum_{i\in A}(Y_i-\bar{Y}_\T) &= \left(\sum_{i\in \T} \binom{n_1-1}{s_1-1}Y_i\right) - s_1\binom{n_1}{s_1}\bar{Y}_\T = 0,
    \end{align*}
    we can write Equation \eqref{eq:crd_est} as 
    \begin{align*}
        \widehat{V} &= \frac{s_1+s_0}{n\binom{n_1-1}{s_1-1}\binom{n_0-1}{s_0-1}}\sum_{A\subseteq \T: |A|=s_1}\sum_{B\subseteq \C: |B| = s_0}\left(\frac{X_A^2}{(n_1-s_1)^2} + \frac{Z_B^2}{(n_0-s_0)^2}\right)\\
        &= \frac{s_1+s_0}{n\binom{n_1-1}{s_1-1}\binom{n_0-1}{s_0-1}}\left[\sum_{A\subseteq \T: |A|=s_1}\binom{n_0}{s_0}\frac{X_A^2}{(n_1-s_1)^2} +\sum_{B\subseteq \C: |B| = s_0} \binom{n_1}{s_1}\frac{Z_B^2}{(n_0-s_0)^2}\right]
    \end{align*}
    where $X_A = \sum_{i\in A}(Y_i - \bar{Y}_\T)$ and $Z_B = \sum_{j\in B}(Y_j - \bar{Y}_\C).$ 
    
    Next, note that
    \begin{align*}
        \sum_{A\subseteq \T: |A|=s_1}\sum_{i,j\in A: i\neq j}(Y_i - \bar{Y}_\T)(Y_j - \bar{Y}_\T) &= \binom{n_1-2}{s_1-2}\sum_{i,j\in \T: i\neq j}(Y_i - \bar{Y}_\T)(Y_j - \bar{Y}_\T)\\
        &= \binom{n_1-2}{s_1-2} \left[\left(\sum_{i\in \T}(Y_i - \bar{Y}_\T)\right)^2 - \sum_{i\in \T}\left(Y_i - \bar{Y}_\T\right)^2\right] \\
        &=  - \binom{n_1-2}{s_1-2}\sum_{i\in \T}\left(Y_i - \bar{Y}_\T \right)^2
    \end{align*}
    and thus 
    \begin{align*}
        \sum_{A\subseteq \T: |A|=s_1}X_A^2 &= \sum_{A\subseteq \T: |A|=s_1}\sum_{i\in A}\left(Y_i - \bar{Y}_\T\right)^2- \binom{n_1-2}{s_1-2}\sum_{i\in \T}\left(Y_i - \bar{Y}_\T\right)^2\\
        &= \left(\binom{n_1-1}{s_1-1} - \binom{n_1-2}{s_1-2}\right)\sum_{i\in \T}\left(Y_i - \bar{Y}_\T\right)^2\\
        &= \binom{n_1-2}{s_1-1}\sum_{i\in \T}\left(Y_i - \bar{Y}_\T\right)^2.
    \end{align*}
    Similarly, we obtain 
    \begin{align*}
         \sum_{B\subseteq \C: |B|=s_0}Z_B^2 &= \binom{n_0-2}{s_0-1}\sum_{j\in \C}\left(Y_j - \bar{Y}_\C\right)^2.
    \end{align*}

    Bringing everything together, we conclude
    \begin{align*}
        \widehat{V} &= \frac{s_1+s_0}{n\binom{n_1-1}{s_1-1}\binom{n_0-1}{s_0-1}}\left(\frac{\binom{n_0}{s_0}\binom{n_1-2}{s_1-1}}{(n_1-s_1)^2}\sum_{i\in \T}\left(Y_i - \bar{Y}_\T\right)^2 + \frac{\binom{n_1}{s_1}\binom{n_0-2}{s_0-1}}{(n_0-s_0)^2}\sum_{j\in \C}\left(Y_j - \bar{Y}_\C\right)^2\right)\\
        &= \frac{s_1+s_0}{n(n_1-s_1)^2}\frac{n_0}{s_0}\frac{n_1-s_1}{n_1-1}\sum_{i\in \T}\left(Y_i - \bar{Y}_\T\right)^2 + \frac{s_1+s_0}{n(n_0-s_0)^2}\frac{n_1}{s_1}\frac{n_0-s_0}{n_0-1}\sum_{j\in \C}\left(Y_j - \bar{Y}_\C\right)^2\\
        &= \alpha\cdot \frac{\hat{S}_\T^2}{n_1} + \beta\cdot \frac{\hat{S}_\C^2}{n_0}
    \end{align*}
    as desired.
\end{proof}

The following result provides further insight on the Neyman Jackknife estimator in Proposition \ref{Neyman_connection}.
\begin{proposition}\label{linear}
    Assume SUTVA and the completely randomized design with $n_1,n_0\geq 2.$ For real $a,b>0,$ consider
    \begin{align*}
        \widehat{V}(a,b) = a\cdot \frac{\hat S_\T^2}{n_1} + b\cdot \frac{\hat S_\C^2}{n_0}.
    \end{align*}
    Then, we have
    \begin{align*}
        \E[\widehat{V}(a,b)]\geq \Var(\hat\tau)
    \end{align*}
    for any table of potential outcomes if and only if $p_0/a + p_1/b\leq 1,$ where $p_0 = n_0/n$ and $p_1 = n_1/n.$ Moreover, such conservative $\widehat V(a,b)$ is admissible if and only if $p_0/a + p_1/b = 1$.
\end{proposition}

Admissibility is in the following sense. We say $\widehat{V}(a',b')$ dominates $\widehat{V}(a,b)$ if $\E[\widehat{V}(a',b')]\leq \E[\widehat{V}(a,b)]$ for all tables of potential outcomes, with strict inequality for at least one. A conservative $\widehat{V}(a,b)$ is admissible if no conservative $\widehat{V}(a',b')$ dominates it. Now consider $\widehat{V}=\widehat{V}(\alpha,\beta)$ from Proposition \ref{Neyman_connection}. A simple calculation gives 
\begin{align*}
    \frac{p_0}{\alpha} + \frac{p_1}{\beta} =  1-\frac{\s}{n} < 1
\end{align*}
where $\s$ is given in Proposition~\ref{spectral_CRT_k}. Thus, with the correction factor $(1-\s/n),$ we see that 
$$(1-\s/n)\widehat{V} = \widehat{V}((1-\s/n)\alpha, (1-\s/n)\beta)$$ is in fact admissible. In the proportional case $s_1:s_0 = n_1:n_0,$ we exactly have $(1-\s/n)\widehat{V} = \widehat{V}_{\text{Neyman}}.$

\begin{proof}[Proof of Proposition~\ref{linear}]
    We know that
    \begin{align*}
        \Var(\hat\tau) = \left(\frac{1}{n_1} - \frac{1}{n}\right)S_1^2 + \left(\frac{1}{n_0} - \frac{1}{n}\right)S_0^2 + \frac{2}{n}S_{10}
    \end{align*}
    where $S_w^2$ is the within-arm population variance and $S_{10}$ is the population covariance between $Y_i(1)$ and $Y_i(0).$ Thus, we can write the bias as 
    \begin{align*}
        n\cdot\left(\E[\widehat{V}(a,b)] - \Var(\hat\tau)\right) = \frac{a-p_0}{p_1}S_1^2 + \frac{b-p_1}{p_0}S_0^2 - 2S_{10}. 
    \end{align*}
    Write $u^2 = S_1^2$ and $v^2 = S_0^2$ so that $S_{10}\leq uv$ by Cauchy--Schwarz with equality obtained when $(Y_i(1))_{i=1}^n$ and $(Y_i(0))_{i=1}^n$ are positively collinear. Hence, $\widehat{V}(a,b)$ is conservative in expectation for all potential outcomes if and only if 
    \begin{align*}
        \frac{a-p_0}{p_1}u^2 + \frac{b-p_1}{p_0}v^2 - 2uv \geq 0
    \end{align*}
    for all $u,v \geq 0.$ This is equivalent to 
    \begin{align*}
         Ax^2 -2x + B\geq 0
    \end{align*}
    for all $x\geq 0$ with $A = (a-p_0)/p_1$ and $B = (b-p_1)/p_0.$ Completing the square, our condition is
    \begin{align*}
        A\left(x - \frac{1}{A}\right)^2 + B - \frac{1}{A}\geq 0.
    \end{align*}
    Thus, we want $A>0, AB \geq 1$, which is equivalent to 
    \begin{align*}
        \frac{p_0}{a} + \frac{p_1}{b}\leq 1
    \end{align*}
    as desired.

    For admissibility, say that $p_0/a + p_1/b = q < 1.$ Then, we see that
    \begin{align*}
        \E[\widehat{V}(a,b)] - \E[\widehat{V}(qa,qb)] = \frac{(1-q)a}{n_1}S_1^2 + \frac{(1-q)b}{n_0}S_0^2 >0
    \end{align*}
    if $S_1^2, S_0^2 > 0$ while $\widehat{V}(qa,qb)$ is still conservative (since $p_0/(qa) + p_1/(qb) = 1).$ Thus, $\widehat{V}(a,b)$ is inadmissible. Next, say that $p_0/a + p_1/b = 1.$ For sake of contradiction, assume that $\widehat{V}(a',b')$ is conservative and dominates $\widehat{V}(a,b).$ Then, we must have 
    \begin{align*}
        \E[\widehat{V}(a
        ,b)] - \E[\widehat{V}(a',b')] = \frac{a-a'}{n_1}S_1^2 + \frac{b-b'}{n_0}S_0^2 \geq 0
    \end{align*}
    for all $S_1^2,S_0^2\geq 0$ with at least one strict inequality. This means that $a>a'$ and $b\geq b'$ or $a\geq a'$ and $b>b'$ must hold. However, this implies that 
    \begin{align*}
        \frac{p_0}{a'} + \frac{p_1}{b'} > \frac{p_0}{a} + \frac{p_1}{b} = 1,
    \end{align*}
    implying $\widehat{V}(a',b')$ is not conservative. This is a contradiction, hence $\widehat{V}(a,b)$ is admissible.
\end{proof}

\subsection{Proofs for Section~\ref{NW}}\label{M_hop_sec_appendix}
\begin{proof}[Proof of Proposition~\ref{NW_connection}]
    Let $\mathbf{x}\in\mathbb{R}^n$ be the vector with $i$th entry $x_i = \psi_i - \overline{\psi}_n$ and $B$ be the circular Bartlett matrix with bandwidth $s+2M$. Here, $\overline{\psi}_n$ is the sample mean of the $\psi_i$'s. Then, the sample-centered Newey--West estimator \citep{NeweyWest1987} with bandwidth $s+2M$ and circular Bartlett weights is 
\begin{align*}
    \widehat V_{\text{NW}} = \frac{1}{n^2}\cdot \mathbf{x}'B\mathbf{x}.
\end{align*}
It remains to show that
\begin{align}
    &\widehat V 
    = \frac{s+2M}{s}\cdot \frac{1}{(n-s-2M)^2}\cdot \mathbf{x}'B \mathbf{x}.
\end{align}
Letting $D_i=\H_{S_i}$ denote the padded version of $S_i,$ we get
\begin{align*}
    \hat\tau - f^{(-S_i)} 
    &= \overline{\psi}_n - \frac{1}{n-|D_i|}\sum_{j\notin D_i}\psi_j\\
    &= \frac{1}{n-|D_i|}\sum_{j\in D_i}(\psi_j - \overline{\psi}_n).
\end{align*}
Let $\mathbf{v}_i$ denote the vector with $j$th entry $v_{ij}=\mathbbm{1}_{j\in D_i}.$ Then, we see that
\begin{align*}
    \hat\tau - f^{(-S_i)} = \frac{1}{n - s- 2M}\mathbf{x}'\mathbf{v}_i.
\end{align*}
By Equation \eqref{eq:NW_start}, we get
\begin{align*}
    \widehat V &= \frac{1}{s(n-s-2M)^2}\cdot \mathbf{x}'\left(\sum_{i=1}^n \mathbf{v}_i \mathbf{v}_i'\right) \mathbf{x}.
\end{align*}
Noting $\left(\mathbf{v}_i\mathbf{v}_i'\right)_{jk} = \mathbbm{1}_{j,k\in D_i}$, we get
\begin{align*}
\left(\sum_{i=1}^n \mathbf{v}_i \mathbf{v}_i'\right)_{jk} = \sum_{i=1}^n \mathbbm{1}_{j,k\in D_i} = 
\begin{cases}
    s + 2M - d(j,k) & d(j,k) \leq s + 2M-1,\\
    0 & d(j,k) > s + 2M -1
\end{cases}    
\end{align*}
for each $j,k\in [n]$ where $d(j,k)$ is the graph-distance on the cycle graph. These entries are exactly $(s+2M)$ times the entries of the circular Bartlett matrix with bandwidth $s+2M.$ Thus, we get 
$$\sum_{i=1}^n\mathbf{v}_i\mathbf{v}_i' = (s+2M)B$$
and conclude the proof.
\end{proof}

This connection extends to more general graphs $G_0.$ 
\begin{proposition}\label{Mhop}
    Let $G_0$ denote the underlying undirected graph for $M$-hop interference and let $r\geq 0$ denote an integer parameter. Let the subset-sampling rule draw $S=B_r(\ell)$ where $\ell$ is uniformly sampled from $[n],$ independent of $\W.$ Given $S=B_r(\ell)$ take the drop-and-recompute proxy $f^{(-S)}$ that re-averages $\psi_i$ over $i\notin B_{r+M}(\ell).$ Then, the resulting Neyman Jackknife variance estimator is
    \begin{align}
        \widehat{V} = \frac{1}{\min_{\ell\in [n]} |B_r(\ell)|}\sum_{i=1}^n\sum_{j=1}^n K_{ij}(\psi_i-\overline{\psi}_n)(\psi_j - \overline{\psi}_n), \quad K_{ij} = \sum_{\ell=1}^n \frac{\mathbbm{1}\{i,j\in B_{r+M}(\ell)\}}{(n - |B_{r+M}(\ell)|)^2} \label{NJ_network_weights}.
    \end{align}
\end{proposition}
If the $B_{r+M}(\ell)$ have common size $b$ for all $\ell,$ then we have $K_{ij} = (n-b)^{-2}|B_{r+M}(i)\cap B_{r+M}(j)|.$ These weights are exactly proportional to the overlap weights given by \citet[Equation 4.2]{kojevnikov2021bootstrap}.

\begin{proof}[Proof of Proposition~\ref{Mhop}]
    We proceed similar to the proof of Proposition~\ref{NW_connection}. First, note that 
    \begin{align*}
        \Pr(j\in S) = \frac{|B_{r}(j)|}{n}
    \end{align*}
    for each $j\in [n].$ By Lemma~\ref{spectral_ber}, we have 
    \begin{align*}
        \s = \min_{\ell\in [n]}|B_{r}(\ell)|.
    \end{align*}
    Thus, the Neyman Jackknife variance estimator is given by 
    \begin{align*}
        \widehat{V} &= \frac{n}{\min_{\ell\in [n]}|B_{r}(\ell)|}\sum_{\ell=1}^n \frac{1}{n}\cdot \left(\hat\tau - f^{(-S_\ell)}\right)^2\\
        &=  \frac{1}{\min_{\ell\in [n]}|B_{r}(\ell)|}\sum_{\ell=1}^n  \left(\hat\tau - f^{(-S_\ell)}\right)^2
    \end{align*}
    where $S_\ell = B_{r}(\ell), D_\ell = B_{r+M}(\ell),$ and 
    $$f^{(-S_\ell)} = \frac{1}{n - |D_\ell|}\sum_{i\notin D_\ell}\psi_i.$$
    To conclude, note that
    \begin{align*}
        \hat\tau - f^{(-S_\ell)} = \frac{1}{n}\sum_{i=1}^n \psi_i - \frac{1}{n - |D_\ell|}\sum_{i\notin D_\ell}\psi_i
        = \frac{1}{n-|D_\ell|}\sum_{i\in D_\ell}(\psi_i - \overline{\psi}_n).
    \end{align*}
    Letting $\mathbf{x}\in\mathbb{R}^n$ denote the vector with $i$th entry $x_i = \psi_i-\overline{\psi}_n$ and $\mathbf{v}_\ell$ denote the vector with $i$th entry $v_{\ell i} =\mathbbm{1}_{i\in D_\ell},$ we get 
    \begin{align*}
        \hat\tau - f^{(-S_\ell)} = \frac{1}{n - |B_{r+M}(\ell)|}\mathbf{x}'\mathbf{v}_\ell
    \end{align*}
    and thus 
    $$
        \widehat{V} = \frac{1}{\min_{\ell\in [n]}|B_{r}(\ell)|}\mathbf{x}'K\mathbf{x},\qquad
    K = \sum_{\ell=1}^n \frac{1}{(n - |B_{r+M}(\ell)|)^2 }\mathbf{v}_\ell\mathbf{v}_\ell'.$$
    Noting 
    $K_{ij} = \sum_{\ell=1}^n (n-|D_\ell|)^{-2}\mathbbm{1}_{i,j\in D_\ell},$ we conclude the proof.
\end{proof}

\subsection{Proofs for Section~\ref{network_designs}}

\begin{proof}[Proof of Theorem~\ref{main_indep}]
    Assume $\lambda=\s/m>0$ as otherwise the statement is trivial. Write $f^{(-A)}(\W_{-A}) = g(A,\W_{-A}).$ First, we have
    \begin{align*}
        \E[\widehat{V}] = \frac{1}{\lambda}\sum_{A\subseteq [m]} \Pr(S=A)\cdot\E\left[\big(f(\mathbf{W}) - g(A, \mathbf{W}_{-A})\big)^2\right].
    \end{align*}
    Since the conditional expectation minimizes mean squared error, we further get
    \begin{align*}
        \E[\widehat{V}]\geq \frac{1}{\lambda }\sum_{A\subseteq [m]}\Pr(S=A)\cdot\E\left[\big(f(\W) - \E\left[f(\W)\mid \W_{-A}\right]\big)^2\right].
    \end{align*}
    As $f(\W)$ has finite variance, we can write its Hoeffding decomposition \citep[Section~11.4, Lemma~11.11]{van}
    \begin{align}
        f(\W) = \sum_{B\subseteq [m]}f_B(\W_B)\label{eq:Hoeffding}.
    \end{align}
    Here, each $f_B(\W_B)$ is square-integrable and $\E[f_B(\W_B)\mid \W_C] = 0$ if $B\nsubseteq C.$ In particular, $\E[f_\emptyset] = \E[f(\W)]$ and $\E[f_B] = 0$ if $B\neq\emptyset$. Thus, we have 
    $$f(\W) - \E\left[f(\W)\mid \W_{-A}\right] = \sum_{B:B\cap A\neq\emptyset}f_B(\W_B)$$
    and also 
    $$\E\left[\big(f(\W) - \E\left[f(\W)\mid \W_{-A}\right]\big)^2\right] = \sum_{B:B\cap A\neq \emptyset}\E\left[f_B(\W_B)^2\right]$$
    by orthogonality. Finally, we swap the order of summation to get
    \begin{align*}
        \E[\widehat{V}]&\geq \frac{1}{\lambda}\sum_{A\subseteq [m]} \Pr(S=A)\sum_{B:B\cap A\neq \emptyset}\E\left[f_B(\W_B)^2\right]\\
        &= \frac{1}{\lambda}\sum_{B\subseteq [m]}\left(\sum_{A: B\cap A\neq\emptyset}\Pr(S=A)\right) \E\left[f_B(\W_B)^2\right]\\
        &= \frac{1}{\lambda}\sum_{B\subseteq [m]}\Pr(S\cap B\neq\emptyset)\cdot \E\left[f_B(\W_B)^2\right]\\
        &\geq \sum_{B\subseteq [m]:B\neq\emptyset}\E\left[f_B(\W_B)^2\right]\\
        &= \Var(f(\W)).
    \end{align*}
    This concludes the proof.
\end{proof}

\subsection{Proofs for Section~\ref{bias_variance_sec}}
\begin{proof}[Proof of Lemma~\ref{bias_var_delete}]
We begin with the bias. For $k\in [K],$ write $$\mathcal{Z}_k = \hat\tau - f^{(-V_k)},\qquad b_{ki} = \begin{cases}
        1 & i \in D_k,\\
        -\frac{|D_k|}{n-|D_k|} & i\notin D_k
    \end{cases}$$
    so that $\mathcal{Z}_k = n^{-1}\sum_{i=1}^n b_{ki}\psi_i$ and $\widehat{V} = \sum_{k=1}^K \mathcal{Z}_k^2.$ Since 
    \begin{align*}
        \E[\mathcal{Z}_k] = \frac{\sum_{i\in D_k}(\tau_i-\tau)}{n-|D_k|}, 
    \end{align*}
    simple calculations show that
    \begin{align*}
        \sum_{k=1}^K \E[\mathcal{Z}_k]^2 = \sum_{k=1}^K \left(\frac{\sum_{i\in D_k}(\tau_i-\tau)}{n-|D_k|}\right)^2 = \frac{1}{n^2}\sum_{k=1}^K \left(\sum_{i\in D_k}(\tau_i - \overline{\tau}_{-k})\right)^2 = \sum_{k=1}^K \left(\overline{\tau}_{-k}-\tau\right)^2 = H_{\mathcal{P}}.
    \end{align*}
    Thus, we get
    \begin{align*}
        \E[\widehat{V}] - H_{\mathcal{P}} = \sum_{k=1}^K \Var(\mathcal{Z}_k).
    \end{align*}
    Further define the martingale differences
    \begin{align*}
        \Delta_{ij} = \E[\psi_i\mid W_1,\dots, W_j] - \E[\psi_i\mid W_1,\dots, W_{j-1}]
    \end{align*}
    and $d_j = \frac{1}{n}\sum_{i=1}^n \Delta_{ij}$
     so that 
     \begin{align*}
         \Var(\hat\tau) = \sum_{j=1}^n \E[d_j^2] = \frac{1}{n^2}\sum_{j=1}^n \E\left[ \left(\sum_{i=1}^n \Delta_{ij}\right)^2\right].
     \end{align*}
    Similarly, we have 
    \begin{align*}
        \Var(\mathcal{Z}_k) =  \frac{1}{n^2}\sum_{j=1}^n \E\left[ \left(\sum_{i=1}^n b_{ki}\Delta_{ij}\right)^2\right].
    \end{align*}

     If $j\notin \N_i,$ we have $\Delta_{ij} = 0$ since $\W$ has independent entries. Thus, fixing $j\in V_{k(j)},$ if $\Delta_{ij}\neq 0,$ then $j\in \N_i$ which implies $i\in D_{k(j)}$ and thus $b_{k(j),i} = 1$. Hence, we can write 
     \begin{align*}
         \E[\widehat{V}] - H_{\mathcal{P}} - \Var(\hat\tau) = \frac{1}{n^2}\sum_{j=1}^n \sum_{k \in [K]\setminus\{k(j)\}}\E\left[ \left(\sum_{i=1}^n b_{ki}\Delta_{ij}\right)^2\right] \geq 0.
     \end{align*}
     For the upper bound, we begin with
     \begin{align*}
         \left(\sum_{i=1}^n b_{ki}\Delta_{ij}\right)^2 =  \left(\sum_{i: j\in \N_i} b_{ki}\Delta_{ij}\right)^2 \leq h_\star \sum_{i: j\in \N_i}b_{ki}^2 \Delta_{ij}^2
     \end{align*}
     by Cauchy--Schwarz. Thus, 
     \begin{align*}
         \E[\widehat{V}] - H_{\mathcal{P}} - \Var(\hat\tau) &= \frac{1}{n^2}\sum_{j=1}^n \sum_{k \neq k(j)}\E\left[ \left(\sum_{i=1}^n b_{ki}\Delta_{ij}\right)^2\right]\\ &\leq \frac{h_\star}{n^2}\sum_{j=1}^n \sum_{k\neq k(j)}\sum_{i: j\in \N_i}b_{ki}^2 \E[\Delta_{ij}^2]\\
         &= \frac{h_\star}{n^2}\sum_{i=1}^n \sum_{j: j\in \N_i}\E[\Delta_{ij}^2]\sum_{k\neq k(j)}b_{ki}^2.
     \end{align*}

     To conclude, fix $i,j$ such that $j\in \N_i.$ Then, there are $\nu_i$ many parts satisfying $i\in D_k.$ Each such part contributes $b_{ki}^2 = 1.$ Since $i\in D_{k(j)}$ and the sum excludes $k(j),$ we get
     \begin{align*}
         \sum_{k\neq k(j)}b_{ki}^2 = (\nu_i-1) + \sum_{k:i\notin D_k}\frac{|D_k|^2}{(n-|D_k|)^2}.
     \end{align*}
     This gives 
     \begin{align*}
         \E[\widehat{V}] - H_{\mathcal{P}} - \Var(\hat\tau)&\leq \frac{h_\star}{n^2}\sum_{i=1}^n \left((\nu_i-1) + \sum_{k:i\notin D_k}\frac{|D_k|^2}{(n-|D_k|)^2}\right)\Var(\psi_i)\\
         &\leq \frac{C^2h_\star}{n^2}\left(\sum_{i=1}^n (\nu_i-1) + \sum_{k=1}^K \frac{|D_k|^2}{n-|D_k|}\right)
     \end{align*}
     as desired.

We now consider the variance. Recall that
    \begin{align*}
        \widehat{V} = \sum_{k=1}^K \mathcal{Z}_k^2
    \end{align*}
    where $\mathcal{Z}_k = \hat\tau - f^{(-V_k)}= n^{-1}\sum_{i=1}^n b_{ki}\psi_i$. 
    Since $\widehat{V}$ is a function of $\W = (W_1,\dots, W_n)$ with independent entries, Efron--Stein gives
    \begin{align*}
        \Var(\widehat{V})\leq \frac{1}{2}\sum_{j=1}^n \E\left[\left(\widehat{V} - \widehat{V}^{(j)}\right)^2\right]
    \end{align*}
    where $g^{(j)} = g(W_1,\dots, W_{j-1},W'_j,W_{j+1},\dots, W_n)$ and $W_j'$ is an independent copy of $W_j.$ 

    Note that
    \begin{align*}
        \widehat{V} - \widehat{V}^{(j)} &= \sum_{k=1}^K (\mathcal{Z}_k - \mathcal{Z}_k^{(j)})(\mathcal{Z}_k + \mathcal{Z}_k^{(j)})\\
        &= \sum_{k=1}^K \left[\frac{1}{n}\sum_{i=1}^n b_{ki}(\psi_i - \psi_i^{(j)})\right](\mathcal{Z}_k + \mathcal{Z}_{k}^{(j)})\\
        &= \sum_{i=1}^n (\psi_i - \psi_i^{(j)})\left[\frac{1}{n}\sum_{k=1}^K b_{ki}(\mathcal{Z}_k + \mathcal{Z}_k^{(j)})\right]\\
        &= \sum_{i: j\in\N_i} (\psi_i - \psi_i^{(j)})\left[\frac{1}{n}\sum_{k=1}^K b_{ki}(\mathcal{Z}_k + \mathcal{Z}_k^{(j)})\right].
    \end{align*}
    Write $\xi_i = n^{-1}\sum_{k=1}^K b_{ki}\mathcal{Z}_k$. By Cauchy--Schwarz, we further get
    \begin{align*}
        \left(\widehat{V}-\widehat{V}^{(j)}\right)^2&\leq |\{i:j\in\N_i\}|\cdot(2C)^2 \sum_{i:j\in\N_i}\left(\xi_i + \xi_i^{(j)}\right)^2\\
        &\leq 8C^2h_\star \sum_{i:j\in\N_i} \left(\xi_i^2 + \left(\xi_i^{(j)}\right)^2\right).
    \end{align*}
    Since $\E[\xi_i^2] = \E[(\xi_i^{(j)})^2],$ we obtain
    \begin{align*}
        \Var(\widehat{V})&\leq 8C^2h_\star\sum_{j=1}^n \sum_{i:j\in\N_i}\E[\xi_i^2]\\
        &=8C^2h_\star \sum_{i=1}^n |\N_i|\E[\xi_i^2] \\
        &\leq 8C^2 h_\star \N_\star \E\left[\sum_{i=1}^n \xi_i^2\right].
    \end{align*}

    To conclude, let
    \begin{align*}
        \zeta_i = \sum_{k:i\in D_k}\frac{\mathcal{Z}_k}{n-|D_k|}.
    \end{align*}
    Then, we have
    \begin{align*}
        \xi_i = \zeta_i - \frac{1}{n}\sum_{k=1}^K \frac{|D_k|}{n-|D_k|}\mathcal{Z}_k = \zeta_i - \overline{\zeta}_n
    \end{align*}
    where $\overline{\zeta}_n = n^{-1}\sum_{i=1}^n \zeta_i.$ Moreover,
    \begin{align*}
        \sum_{i=1}^n \xi_i^2 = \sum_{i=1}^n (\zeta_i - \overline{\zeta}_n)^2 = \sum_{i=1}^n \zeta_i^2 - n\overline{\zeta}_n^2\leq \sum_{i=1}^n \zeta_i^2.
    \end{align*}
    Finally, we have
    \begin{align*}
        \zeta_i^2 \leq |\{k:i\in D_k\}|\sum_{k:i\in D_k}\frac{\mathcal{Z}_k^2}{(n-|D_k|)^2}\leq \nu_{\star} \sum_{k:i\in D_k}\frac{\mathcal{Z}_k^2}{(n-|D_k|)^2}
    \end{align*}
    which gives
    \begin{align*}
        \sum_{i=1}^n \xi_i^2 &\leq \nu_{\star}\sum_{i=1}^n \sum_{k:i\in D_k}\frac{\mathcal{Z}_k^2}{(n-|D_k|)^2}\\
        &\leq \nu_{\star}\sum_{k=1}^K \frac{|D_k|}{(n-|D_k|)^2}\mathcal{Z}_k^2\\
        &\leq \frac{\nu_{\star} d_{\star}}{(n-d_{\star})^2}\widehat{V}.
    \end{align*}
    Bringing everything together gives
    \begin{align*}
        \Var(\widehat{V})&\leq 8C^2h_\star\cdot \frac{\N_\star\nu_{\star} d_{\star}}{(n-d_{\star})^2}\cdot \E[\widehat{V}]
    \end{align*}
    as desired.
\end{proof}

\begin{proof}[Proof of Corollary~\ref{bounded_deg_bv}]
    Since $G_0$ has bounded maximum degree, we get
    \begin{align*}
        \nu_\star\leq \N_\star = O(1),\qquad h_\star = O(1).
    \end{align*}
    We begin with the bias. Note that 
    \begin{align*}
        \sum_{i=1}^n (\nu_i-1)&\leq \nu_\star |\partial_{\mathcal{P}}|,\\
        \sum_{k=1}^K \frac{|D_k|^2}{n-|D_k|}&\leq \frac{Kd_\star^2}{n-d_\star} = O\left(\frac{n}{K}\right).
    \end{align*}
    Hence, by Lemma~\ref{bias_var_delete}, we have 
    \begin{align*}
        0\leq \E[\widehat{V}] - H_{\mathcal{P}} - \Var(\hat\tau) = O(\frac{|\partial_{\mathcal{P}}|}{n^2} + \frac{1}{nK})
    \end{align*}
    as desired. For the variance, again by Lemma~\ref{bias_var_delete} we have
    \begin{align*}
        \Var(\widehat{V}) &= O\left(\frac{d_\star}{(n-d_\star)^2}\right)\cdot \E[\widehat{V}]\\
        &= O\left(\frac{1}{nK}\right)\cdot \E[\widehat{V}].
    \end{align*}
    This concludes the proof.
\end{proof}

\begin{proof}[Proof of Proposition~\ref{coverage}]
    We first show asymptotic normality. Note that the dependency graph of the $\psi_i$'s has bounded degree. Since $|\psi_i|\leq C$ and we assume $\Var(\hat\tau)\gtrsim 1/n,$ the dependency graph CLT \citep{ross2011stein} gives 
    \begin{align*}
        \frac{\hat\tau - \tau}{\sqrt{\Var(\hat\tau)}}\limd N(0,1)
    \end{align*}
    where we use $\tau = \E[\hat\tau].$ Next, we obtain coverage. Writing $z = z_{1-\alpha/2},$ we get
    \begin{align*}
        \Pr(|\hat\tau - \tau|> z\sqrt{\widehat{V}}) &= \Pr\left(\frac{|\hat\tau - \tau|}{\sqrt{\Var(\hat\tau)}}> z\sqrt{\frac{\widehat{V}}{\Var(\hat\tau)}}\right)\\
        &\leq \Pr\left(\frac{|\hat\tau - \tau|}{\sqrt{\Var(\hat\tau)}}> z\sqrt{\frac{\widehat{V}}{\E[\widehat{V}]}}\right)\\
        &\to \alpha
    \end{align*}
    where the last step follows from $\widehat{V}/\E[\widehat{V}]\limp 1$ and Slutsky's lemma. This concludes the proof.
\end{proof}

\subsection{Proofs for Section~\ref{covariate_sec}}

\begin{proof}[Proof of Proposition~\ref{oracle_bias_prop}]
    We first write 
    \begin{align*}
        \UB &= \E\left[\sum_{k=1}^K \left(f - \E[f\mid \W_{-V_k}]\right)^2\right].
    \end{align*}
    Next, for each $k$, write $f = A_k + B_k$ where
    $$A_k = \frac{1}{n}\sum_{i\in D_k}\psi_i,\qquad B_k = \frac{1}{n}\sum_{i\notin D_k}\psi_i$$
    so that $\E[f\mid \W_{-V_k}] = \E[A_k\mid \W_{-V_k}] + B_k.$ Then, we can further write 
    \begin{align*}
        \UB = \sum_{k=1}^K \E[\Var(A_k\mid\W_{-V_k})]\leq \sum_{k=1}^K \Var(A_k).
    \end{align*}
    As in our proof of Lemma~\ref{bias_var_delete}, define the martingale differences 
    $$\Delta_{ij} = \E[\psi_i\mid W_1,\dots, W_j] - \E[\psi_i\mid W_1,\dots,W_{j-1}]$$
    so that
    \begin{align*}
        \Var(A_k) = \frac{1}{n^2}\sum_{j=1}^n \E\left[\left(\sum_{i\in D_k}\Delta_{ij}\right)^2\right],\qquad \Var(\hat\tau) = \frac{1}{n^2}\sum_{j=1}^n \E\left[\left(
        \sum_{i=1}^n \Delta_{ij}\right)^2\right].
    \end{align*}
    For fixed $j$, let $k(j)\in[K]$ denote the part that contains $j$. If $i\notin D_{k(j)},$ then $j\notin \N_i$ and hence $\Delta_{ij} = 0.$ In other words, $\sum_{i\in D_{k(j)}}\Delta_{ij} = \sum_{i=1}^n \Delta_{ij}$ and we get
    \begin{align*}
        \UB - \Var(\hat\tau) &\leq \frac{1}{n^2}\sum_{j=1}^n \sum_{k\neq k(j)}\E\left[\left(\sum_{i\in D_k}\Delta_{ij}\right)^2\right]\\
        &\leq \frac{1}{n^2}\sum_{j=1}^n \sum_{k\neq k(j)}|\{i:j\in\N_i\}|\sum_{i\in D_k}\E[\Delta_{ij}^2]\\
        &\leq \frac{h_\star}{n^2}\sum_{j=1}^n \sum_{k\neq k(j)}\sum_{i\in D_k}\E[\Delta_{ij}^2]
    \end{align*}
    where the second inequality is Cauchy--Schwarz. Finally, note that
    \begin{align*}
        \sum_{j=1}^n \sum_{k\neq k(j)}\sum_{i\in D_k}\E[\Delta_{ij}^2] &= \sum_{i=1}^n \sum_{j:j\in\N_i}\sum_{k: i\in D_k, k\neq k(j)}\E[\Delta_{ij}^2].
    \end{align*}
    There are exactly $\nu_i$ many parts $k$ satisfying $i\in D_k,$ with $k(j)$ as one such part. Thus, we conclude
    \begin{align*}
        \UB - \Var(\hat\tau)&\leq \frac{h_\star}{n^2}\sum_{i=1}^n (\nu_i-1)\sum_{j:j\in\N_i}\E[\Delta_{ij}^2]\\
        &\leq \frac{h_\star}{n^2}\sum_{i=1}^n (\nu_i-1)\Var(\psi_i)\\
        &\leq \frac{C^2h_\star}{n^2}\sum_{i=1}^n (\nu_i-1)
    \end{align*}
    as desired.
    
\end{proof}

\begin{proof}[Proof of Proposition~\ref{proxy_error_prop}]
    The proxy error is given by 
    \begin{align*}
        \PE = \E\left[\sum_{k=1}^K \left(f^{(-V_k)} - \E[f\mid \W_{-V_k}]\right)^2\right].
    \end{align*}
    Write $e_{i,w}^{(-A)} = Y_i(\Z = w\cdot\mathbf{1}_n) - \hat\mu_{i,w}^{(-A)}$ for $w\in\{0,1\}$ so that $$\hat{\psi}_{i,k} = \frac{T_i}{p_i}e_{i,1}^{(-V_k)} - \frac{C_i}{q_i}e_{i,0}^{(-V_k)}$$
    and
    $$\E[\hat\psi_{i,k}]= e_{i,1}^{(-V_k)} - e_{i,0}^{(-V_k)} = \tau_i - (\hat\mu_{i,1}^{(-V_k)} - \hat\mu_{i,0}^{(-V_k)}).$$
    By Equations \eqref{eq:CE_form} and \eqref{eq:proxy_form}, we have
    \begin{align*}
        f^{(-V_k)} - \E[f\mid \W_{-V_k}] &= -\frac{1}{n}\sum_{i\in D_k}\left(\frac{\tilde{p}_i}{p_i}e_{i,1}^{(-V_k)} - \frac{\tilde{q}_i}{q_i}e_{i,0}^{(-V_k)}\right)\\
        &= -\frac{1}{n}\sum_{i\in D_k}\E[\hat\psi_{i,k}\mid \W_{-V_k}].
    \end{align*}
    Thus, we can write
    \begin{align*}
        \PE &= \frac{1}{n^2}\sum_{k=1}^K \E\left[\left(\sum_{i\in D_k}\E[\hat\psi_{i,k}\mid\W_{-V_k}]\right)^2\right]\\
        &= H_{\mathcal{P},\hat\mu} + \frac{1}{n^2}\sum_{k=1}^K \Var\left(\sum_{i\in D_k}\E[\hat\psi_{i,k}\mid\W_{-V_k}]\right)
    \end{align*}
    using $\E[\hat\psi_{i,k}] = \tau_i - (\hat\mu_{i,1}^{(-V_k)} - \hat\mu_{i,0}^{(-V_k)}).$ This gives $\PE - H_{\mathcal{P},\hat\mu}\geq 0.$ Further writing $g_{i,k} = \E[\hat\psi_{i,k}\mid\W_{-V_k}],$ define the martingale differences
    \begin{align*}
        \Delta_{ij}^{(k)} = \E[g_{i,k}\mid W_1,\dots, W_j] - \E[g_{i,k}\mid W_1,\dots, W_{j-1}]
    \end{align*}
    so that 
    \begin{align*}
        \Var\left(\sum_{i\in D_k}\E[\hat\psi_{i,k}\mid\W_{-V_k}]\right) &= \sum_{j=1}^n \E\left[\left(\sum_{i\in D_k}\Delta_{ij}^{(k)}\right)^2\right].
    \end{align*}
    Since $\hat\psi_{i,k}$ is a function of $\W_{\N_i}$ and entries of $\W$ are independent, we see that if $j\notin \N_i\setminus V_k,$ we have $\Delta_{ij}^{(k)} = 0.$ Thus, we have
    \begin{align*}
        \PE - H_{\mathcal{P},\hat\mu}&\leq \frac{1}{n^2}\sum_{k=1}^K \sum_{j=1}^n |\{i: j\in\N_i\}|\sum_{i\in D_k}\E\left[\left(\Delta_{ij}^{(k)}\right)^2\right]\\
        &\leq \frac{h_\star}{n^2}\sum_{k=1}^K \sum_{j=1}^n \sum_{i\in D_k}\E\left[\left(\Delta_{ij}^{(k)}\right)^2\right]\\
        &= \frac{h_\star}{n^2}\sum_{i=1}^n \sum_{k:i\in D_k}\Var(g_{i,k}).
    \end{align*}
    Since we can write $\hat\psi_{i,k} = a_{i,k}T_i + b_{i,k}C_i$ with $|a_{i,k}|,|b_{i,k}|\leq C'$ deterministic, we have
    \begin{align*}
        \Var(g_{i,k}) &\leq 2\Var(a_{i,k}\E[T_i\mid \W_{-V_k}]) + 2\Var(b_{i,k}\E[C_i\mid \W_{-V_k}])\\
        &\leq 2(C')^2 \left(\Var(\E[T_i\mid \W_{-V_k}]) + \Var(\E[C_i\mid \W_{-V_k}])\right).
    \end{align*}
    Next, note that
    \begin{align*}
        \sum_{k:i\in D_k}\Var(\E[T_i\mid\W_{-V_k}]) &= \sum_{k:i\in D_k}\Big(\Var(T_i) - \E[\Var(T_i\mid \W_{-V_k})]\Big)\\
        &= \nu_i \Var(T_i) - \sum_{k=1}^K\E[\Var(T_i\mid \W_{-V_k})]\\
        &\leq (\nu_i-1)\Var(T_i)
    \end{align*}
    where the last step is by Theorem~\ref{main_indep}. Since $T_i,C_i$ are indicators, we conclude 
    \begin{align*}
        \PE - H_{\mathcal{P},\hat\mu}&\leq \frac{2(C')^2h_\star}{n^2}\sum_{i=1}^n (\nu_i-1)\left(\Var(T_i)+\Var(C_i)\right)\\
        &\leq \frac{(C')^2 h_\star}{n^2}\sum_{i=1}^n (\nu_i-1)
    \end{align*}
    as desired. This concludes the proof.
\end{proof}

\section{Robustness to Approximate Measurability}\label{robust_sec}

All results using Theorem \ref{Main} assume that the computable proxy $f^{(-S)}$ is strictly measurable with respect to $(S,\W_{-S}).$ The following result shows that our framework still gives asymptotic conservativeness as long as $f^{(-S)}$ is almost measurable on $(S,\W_{-S}).$  

\begin{lemma}\label{robust}
    Let $\W\sim\pi$ be any random vector in $\mathcal{X}^m$ where $\mathcal{X}$ is finite. Let $f:\mathcal{X}^m\to\mathbb{R}$ be measurable such that $f(\W)$  has finite variance. Let $\lambda$ denote the spectral gap of Gibbs rerandomization with update set $S.$ Let $f^{(-S)}$ be any square-integrable function of $(S,\W).$ If 
    $$\E\left[(f^{(-S)} - \E[f^{(-S)}\mid S,\W_{-S}])^2\right] = o(\lambda \Var(f(\W)))$$
    as $m\to\infty,$ then we have
    \begin{align*}
        \frac{1}{\lambda}\E\left[\left(f - f^{(-S)}\right)^2\right]\geq (1 - o(1))\Var(f(\W)).
    \end{align*}
\end{lemma}
\begin{proof}
    We know that $\langle f_1,f_2\rangle = \E[f_1f_2]$ is an inner product for $L^2(\sigma(S,\W))$. By the reverse triangle inequality, we have
    \begin{align*}
        \norm{f - f^{(-S)}}\geq \norm{f - \E[f^{(-S)}\mid S,\W_{-S}]} - \norm{f^{(-S)} - \E[f^{(-S)}\mid S,\W_{-S}]}.
    \end{align*}
    By Theorem \ref{Main}, we know that
    \begin{align*}
        \norm{f - \E[f^{(-S)}\mid S,\W_{-S}]}^2\geq \lambda \Var(f).
    \end{align*}
    Thus, we get
    \begin{align*}
        \norm{f - f^{(-S)}}&\geq \sqrt{\lambda \Var(f)} - o(\sqrt{\lambda\Var(f)})\\
        &= (1-o(1))(\sqrt{\lambda\Var(f)}).
    \end{align*}
    Squaring both sides gives the desired result.
\end{proof}

The following identity is useful for interpreting the condition of Lemma \ref{robust}.
\begin{lemma}\label{robust_identity}
Let $f^{(-S)}=g(S,\W)$ and $\W'$ be obtained from $\W$ by Gibbs rerandomization with update set $S$. Then, we have
\begin{align*}
     \E\left[(f^{(-S)} - \E[f^{(-S)}\mid S,\W_{-S}])^2\right] = \frac{1}{2}\E\left[(g(S,\W) - g(S,\W'))^2\right].
\end{align*}
\end{lemma}
The right-hand side is nice as we know that $\W_{-S} = \W'_{-S}$ by construction.
\begin{proof}[Proof of Lemma \ref{robust_identity}]
    We wish to show that 
    \begin{align*}
        \E\left[(g(S,\W) - g(S,\W'))^2\right] = 2\E\left[(g(S,\W) - h(S,\W_{-S}))^2\right]
    \end{align*}
    where $h(S,\W_{-S})=\E[g(S,\W)\mid S,\W_{-S}].$ The argument is similar to the proof of Lemma \ref{Poincare_mod}. Indeed, since $\W_S$ and $\W'_S$ are i.i.d. given $(S,\W_{-S}),$ we know that 
    \begin{align}
        \E\left[(g(S,\W) - h(S,\W_{-S}))\cdot(g(S,\W') - h(S,\W_{-S}))\right]=0
    \end{align}
    as the conditional expectation given $(S,\W_{-S})$ is zero. Thus, we get
    \begin{align*}
        \E\left[(g(S,\W) - g(S,\W'))^2\right] &= \E\left[(g(S,\W) - h(S,\W_{-S}))^2\right] + \E\left[(g(S,\W') - h(S,\W_{-S}))^2\right]\\
        &= 2\E\left[(g(S,\W) - h(S,\W_{-S}))^2\right]
    \end{align*}
    where the last step follows by Lemma \ref{ed}.
\end{proof}

\subsection{Time Series with Decaying Dependence}\label{time_series} Here, we give an application of Lemma \ref{robust} for time series under the Bernoulli design. Say we have $T$ many units labeled $1,\dots, T$ and let each $t\in [T]$ independently receive treatment $W_t\in\{0,1\}$ with probability $\pi\in(0,1).$ Further assume 
\begin{align*}
    Y_t(\W)=Y_t(W_1,\dots, W_t)
\end{align*}
for each $t\in [T]$ and consider the IPW estimator 
\begin{align*}
    \hat\tau = \frac{1}{T}\sum_{t=1}^T\psi_t= \frac{1}{T}\sum_{t=1}^T \left(\frac{W_t}{\pi}-\frac{1-W_t}{1-\pi}\right)Y_t
\end{align*}
for the average direct effect. 

Let $\{V_1,\dots, V_K\}$ denote the partition of $[T]$ into consecutive blocks of length $T/K$ (we assume $T$ is a multiple of $K$). Let $S=V_\ell$ where $\ell\sim \text{Unif}\{1,2,\dots, K\},$ independent of $\W.$ This gives $\s =T/K$ by Lemma \ref{spectral_ber}. Next, for any $j<K,$ let $D_j\supseteq V_j$ denote the block that adds $r$ units to the right of $V_j$. We assume that $0\leq r\leq T/K.$ Then, for any $j\in [K],$ we define the proxy 
\begin{align*}
    f^{(-V_j)}=\frac{1}{T-|D_j|}\sum_{t\notin D_j}\psi_t
\end{align*}
where $D_K=V_K.$ Note that $f^{(-V_j)}$ is not necessarily measurable on $\W_{-V_j},$ since all units to the right of $V_j$ may depend on $\W_{V_j}.$ However, as shown below, we still obtain asymptotic conservativeness as long as potential outcomes weakly depend on treatments from the distant past and $r$ is large enough. We write $w_{t_1:t_2}=(w_t)_{t=t_1}^{t_2}.$
\begin{proposition}\label{time}
    For each $d\geq 0,$ let $\Delta_{T,d}$ be a deterministic quantity satisfying 
    \begin{align*}
        |Y_t(w)-Y_t(w')|\leq \Delta_{T,d}
    \end{align*}
    for any $w,w'\in\{0,1\}^T$ and $d\leq t\leq T$ such that $w_{t-d+1:t} = w'_{t-d+1:t}.$ Choose $S$ and $f^{(-S)}$ as described above and let 
    \begin{align*}
        \widehat{V} = K\cdot \E\left[\left(f - f^{(-S)}\right)^2\mid \W\right] =  \sum_{j=1}^K \left(f - f^{(-V_j)}\right)^2.
    \end{align*}
    If $K\to\infty$, $\Var(\hat\tau) \asymp T^{-1},$ and 
    \begin{align}
        \sum_{d\geq r+1}\Delta_{T,d}^2 = o\left(\frac{1}{K}\right),\label{eq:asymptotics}
    \end{align}
    then we have 
    \begin{align*}
        \E[\widehat{V}]\geq (1 - o(1))\Var(\hat\tau).
    \end{align*}
\end{proposition}
Our assumption on $\Delta_{T,d}$ can be thought of as a deterministic analogue of the geometric mixing setting assumed in \cite{HuWager2022}. In particular, we see that the buffer size $r$ needed for asymptotic conservativeness depends on the decay rate of $\Delta_{T,d}.$ If $\Delta_{T,d}\lesssim 1/d^{\alpha}$ for some $\alpha>1/2,$ condition \eqref{eq:asymptotics} is satisfied if $r^{2\alpha-1}\gg K.$ By contrast, if $\Delta_{T,d}\lesssim \rho^{d}$ for some $\rho\in (0,1),$ then we see that $r\gg \log(K)$ suffices.  

\begin{proof}[Proof of Proposition \ref{time}]
    For each $j\in [K],$ let $R_j\subseteq [T]$ denote the set of all units strictly to the right of $D_j.$ Then, given $S=V_j,$ we see that
    \begin{align*}
        f^{(-V_j)}(\W) - f^{(-V_j)}(\W') = \frac{1}{T-|D_j|}\sum_{t\in R_j}c_t\cdot (Y_t(\W)-Y_t(\W'))
    \end{align*}
    for all $j\in [K]$ where $c_t=W_t/\pi - (1-W_t)/(1-\pi).$ By Cauchy--Schwarz, we further see that
    \begin{align*}
        \left(f^{(-V_j)}(\W) - f^{(-V_j)}(\W')\right)^2 &\leq \frac{1}{T-|D_j|}\sum_{t\in R_j}c_t^2\cdot (Y_t(\W)-Y_t(\W'))^2\\
        &\leq \frac{C}{T-T/K-r}\sum_{t\in R_j}(Y_t(\W)-Y_t(\W'))^2
    \end{align*}
    for some constant $C>0.$ Note that the leftmost unit in $R_j$ is at least $r+1$ steps away from $V_j.$ Hence, given $S=V_j,$ we have
    \begin{align*}
        \left(f^{(-V_j)}(\W) - f^{(-V_j)}(\W')\right)^2
        &\leq \frac{C}{T-T/K-r} \sum_{d\geq r+1}\Delta_{T,d}^2.
    \end{align*}
    Taking expectations, we conclude 
    \begin{align*}
        \E\left[\left(f^{(-S)}(\W) - f^{(-S)}(\W')\right)^2\right]
        &\leq \frac{C}{T-T/K-r} \sum_{d\geq r+1}\Delta_{T,d}^2.
    \end{align*}
    By Lemma \ref{robust} along with Lemma \ref{robust_identity}, it suffices to show that the upper bound is $o(\lambda \Var(f)) = o(1/(KT)).$ Since $T/K+r = o(T),$ we conclude the proof.
\end{proof}

\section{Details for Numerical Experiments}

Our experiment code is available at \url{https://github.com/bryan1292929/neyman-jackknife-experiments}.

\subsection{Details for Section~\ref{temp_exp}}\label{temp_exp_details}

First, we describe the potential outcomes for our switchback experiment. Let $\tilde{\eta}_0$  and $\tilde{\eta}_1$ be two independent stationary AR(1) paths with autoregressive coefficient $\phi = 0.9$ and unit marginal variance:
\begin{align*}
    \tilde{\eta}_{zt} = \phi\,\tilde{\eta}_{z(t-1)} + \sqrt{1-\phi^2}\,\nu_{zt}, \qquad \tilde{\eta}_{z1} = \nu_{z1}, \qquad z \in \{0,1\},
\end{align*}
where $\nu_{zt}$ are i.i.d. standard normals. We draw the paths $\tilde{\eta}_0,\tilde{\eta}_1$ once and consider them fixed. Center both paths to empirical mean zero, replace $\tilde\eta_1$ by its residual after projecting onto $\tilde\eta_0$, and rescale each to empirical variance one. Call the results $\eta_0, \eta_1$.

For an arm correlation $c \in [-1, 1]$, define the arm baselines
$$\varepsilon_{0t} = \eta_{0t}, \qquad \varepsilon_{1t} = c\,\eta_{0t} + \sqrt{1-c^2}\,\eta_{1t}$$
so that $\varepsilon_0$ and $\varepsilon_1$ have empirical variance one and empirical correlation exactly $c$. For a treatment path $z \in \{0,1\}^T$, the outcome at step $t$ is
$$Y_t(z) = (1-z_t)\,\varepsilon_{0t} + z_t\,(2 + \varepsilon_{1t}) + \gamma \sum_{s=1}^{t-1} (1-\alpha)^{s}\, z_{t-s},$$
where the last term describes decaying carryovers with $\gamma = 0.3$ and decay rate $\alpha = 0.1$. 

Second, we describe the block Jackknife proposed by \citet{HuWager2022} which we compare against in Section~\ref{temp_exp}. Recall that
\begin{align*}
    \hat\tau =  \frac{\sum_{t=1}^TT_tY_t/p_t}{\sum_{t=1}^T T_t/p_t} - \frac{\sum_{t=1}^T C_tY_t/q_t}{\sum_{t=1}^T C_t/q_t}. 
\end{align*}
For each $j = 1,\dots, m-2,$ we define $\hat\tau^{(-j)}$ by recomputing $\hat\tau$ after dropping all outcome units in the focal part of block $j$, both parts of block $j+1$, and the burn-in part of block $j+2.$ While the first block only depends on $W_1,$ we still refer to its focal part as the latter 70 units. 

Formally, let $A_j\subseteq [T]$ denote the time units dropped as described above. Then, we let
\begin{align*}
    \hat\tau^{(-j)} =  \frac{\sum_{t\notin A_j}^TT_tY_t/p_t}{\sum_{t\notin A_j}^T T_t/p_t} - \frac{\sum_{t\notin A_j}^T C_tY_t/q_t}{\sum_{t\notin A_j}^T C_t/q_t}. 
\end{align*}
The block Jackknife variance estimator is 
\begin{align*}
    \widehat{V} = \frac{(m-3)^2}{2(m-2)^2}\sum_{j=1}^{m-2}\left(\hat\tau - \hat\tau^{(-j)}\right)^2.
\end{align*}
In particular, the proxy $\hat\tau^{(-j)}$ illustrates the difference between the block Jackknife for time units and the Neyman Jackknife for randomization units. The Neyman Jackknife estimator analogous to the above deletion rule sets $S=\{j,j+1\}$ where $j\in [m]$ is chosen uniformly at random and entries in $S$ loop from $m$ to 1. If $j\in [m-2],$ then the Neyman Jackknife masks block randomizations $W_j,W_{j+1}$ along with outcome units in block $j$, block $j+1$, and the burn-in part of block $j+2.$ In other words, the Neyman Jackknife additionally removes the burn-in part of block $j$, compared to just removing $A_j.$

Finally, all time steps in the burn-in part (or the focal part) of a block share the same exposure set. Thus, when applying the baselines from Section~\ref{related}, we treat each part as a single outcome unit with its averaged outcome. This does not change the Hájek estimator or either Jackknife.

\subsection{Details for Section~\ref{spatial_exp}}\label{spatial_exp_details}

First, we describe the potential outcomes for our network experiment. Let $\mathcal{A}$ be the neighborhood-averaging function, $(\mathcal{A}v)_i = |B_1(i)|^{-1} \sum_{j \in B_1(i)} v_j$. Draw four independent standard normal vectors in $\mathbb R^n$, apply $\mathcal{A}$ twice to each, and call the results $\tilde x_0, \tilde x_1, \tilde u_0, \tilde u_1$. We consider these fixed and center the four vectors to empirical mean zero. Orthogonalize them by Gram–Schmidt in the order $\tilde x_0, \tilde x_1, \tilde u_0, \tilde u_1$, then rescale each to empirical variance one. Call the results $x_0, x_1, u_0, u_1$. The vectors $x_0, x_1$ will serve as observed covariates, while $u_0, u_1$ serve as unobserved covariates. With $R^2 = 0.75$, let

$$\eta_0 = \sqrt{R^2}\, x_0 + \sqrt{1-R^2}\, u_0, \qquad \eta_1 = \sqrt{R^2}\, x_1 + \sqrt{1-R^2}\, u_1 .$$

For an arm correlation $c \in [-1,1]$, define the arm baselines

$$\varepsilon_{0i} = \eta_{0i}, \qquad \varepsilon_{1i} = c\, \eta_{0i} + \sqrt{1-c^2}\, \eta_{1i},$$
so that $\varepsilon_0$ and $\varepsilon_1$ have empirical variance one and empirical correlation exactly $c$. Regressing $\varepsilon_0$ on $x_0$ explains exactly a fraction $R^2$ of its variance, and likewise $\eta_1$ on $x_1$. The potential outcomes are

$$Y_i(\mathbf{0}_{B_1(i)}) = \varepsilon_{0i}, \qquad Y_i(\mathbf{1}_{B_1(i)}) = 2 + \varepsilon_{1i}.$$
We leave outcomes under partial exposure of $ B_1(i)$ unspecified. 

Second, we describe the proxy with covariates. Given $S\subseteq [n],$ we use
\begin{align*}
    f^{(-S)} = \frac{\sum_{i\notin \H_S}T_iY_i/p_i + \sum_{i\in \H_S}\tilde{p}_i\hat\mu_{i,1}^{(-S)}/p_i}{\sum_{i\notin \H_S}T_i/p_i + \sum_{i\in \H_S}\tilde{p}_i/p_i} - \frac{\sum_{i\notin \H_S}C_iY_i/q_i + \sum_{i\in \H_S}\tilde{q}_i\hat\mu_{i,0}^{(-S)}/q_i}{\sum_{i\notin \H_S}C_i/q_i + \sum_{i\in \H_S}\tilde{q}_i/q_i}
\end{align*}
where $\tilde{p}_i = \E[T_i\mid \W_{-S}]$, $\tilde{q}_i = \E[C_i\mid \W_{-S}].$\footnote{These quantities are also estimated via Monte Carlo iterations.} Letting $X_i = (1, x_{0i}, x_{1i}),$ we set
$\hat\mu_{i,w}^{(-S)} = X_i'\hat\beta_w^{(-S)}$
where $\hat\beta_1^{(-S)},\hat\beta_0^{(-S)}$ minimize
$$\sum_{j\notin \H_S}\frac{T_j}{p_j}\left(Y_j - X_j'\beta\right)^2,\qquad \sum_{j\notin \H_S}\frac{C_j}{q_j}\left(Y_j - X_j'\beta\right)^2,$$
respectively. If we instead choose
\begin{align*}
    \hat\mu_{i,1}^{(-S)} = \frac{\sum_{j\notin \H_S}T_jY_j/p_j}{\sum_{j\notin \H_S}T_j/p_j}, \qquad \hat\mu_{i,0}^{(-S)} = \frac{\sum_{j\notin \H_S}C_jY_j/q_j}{\sum_{j\notin \H_S}C_j/q_j},
\end{align*}
one can check that $f^{(-S)} = \hat\mu_{i,1}^{(-S)} - \hat\mu_{i,0}^{(-S)}.$  

Finally, when implementing the baseline method of \citet{harshaw2026variance}, we keep the semidefinite program tractable at $n = 1000$ by restricting $B - A$ to be a sum of 100 PSD matrices. In particular, we require each PSD matrix to be supported on the 1-hop closure of one cell of a $10 \times 10$ tiling of $[0,1]^2$. This keeps the bound valid but possibly looser than the unrestricted optimum.

\subsection{Details for Section~\ref{panel_exp}}\label{panel_exp_details}
First, we describe the potential outcomes. There are $T_0 = 8$ pre-treatment periods and one post-treatment period. Next, let $\mathcal{G}\mathcal{P}(0,\Sigma)$ denote the Gaussian process with mean zero and covariance function $\Sigma$ given by $\Sigma(p,p') = \exp\left(-\norm{p-p'}^2/(2\cdot 0.2^2)\right)$ for $p,p'\in[0,1]^2.$ Let $a,g^{(1)},g^{(2)},\{s^{(t)}\}_{t=1}^{T_0+1}$ be independent draws from $\mathcal{GP}(0,\Sigma)$ (we approximate the draws by 200 random cosine waves). Let $p_1,\dots, p_n$ denote the position of each unit in $[0,1]^2.$ Then, for $i\in [n]$ and $t\in\{1,\dots, T_0+1\},$ we define the control-arm outcomes 
$$Y_{it}(\mathbf 0_n) = a_i + b_t + g_i' F_t + 0.5\, s_{it} + e_{it}$$
where $a_i = a(p_i)$, $g_i = (g^{(1)}(p_i), g^{(2)}(p_i)),$ $s_{it} = s^{(t)}(p_i)$ are evaluations of the drawn fields (standardized across the $n$ units for each field), and
$b_t = 0.2\,(t-1) +  N(0, 0.3^2)$,  $F_t \sim  N(0, I_2)$,
$e_{it} \sim  N(0, 0.5^2)$. Here, $b_t,F_t$ are independent across $t$ and $e_{it}$ is independent across both $i,t.$ Intuitively, $a_i$ is a unit's fixed effect, $b_t$ is a common time effect, $g_i'F_t$ is an interactive fixed effect, $s_{it}$ is a spatial shock, and $e_{it}$ is idiosyncratic noise. We also define the treated-arm outcome
$$Y_{i,T_0+1}(\mathbf{1}_{B_1(i)}) = Y_{i,T_0+1}(\mathbf{0}_{B_1(i)}) + \theta_i + \gamma\cdot (|B_1(i)| - 1)$$
with unit effects $\theta_i\sim N(1, 0.3^2)$ independent across $i$ and spillover strength $\gamma \in\{0, 0.25, 0.5, 0.75, 1\}$.  All of the above is drawn once and held fixed. Finally, we note that our Neyman Jackknife estimator is the Monte Carlo version \eqref{eq:CRD_est_MC} with $500$ update sets drawn independently.
\end{document}